\documentclass[a4paper,UKenglish,cleveref, autoref, thm-restate]{lipics-v2019}

\usepackage{pgfplots}
\usetikzlibrary{shapes,snakes}
\usetikzlibrary{fit}
\usetikzlibrary{calc}

\usepackage{subcaption}

\usepackage{thmtools}

\usepackage{mathtools}

\usepackage{comment}
\usepackage{amssymb}
\usepackage{cleveref}

\title{The Complexity Classes of Hamming Distance Recoverable Robust Problems}

\titlerunning{The Complexity Classes of Hamming Distance Recoverable Robust Problems}

\author{Christoph Grüne}{Department of Computer Science, RWTH Aachen University, Germany}{gruene@algo.rwth-aachen.de}{https://orcid.org/0000-0002-7789-8870}{}

\authorrunning{C. Grüne}

\Copyright{Christoph Grüne}



\category{}

\relatedversion{} 

\supplement{}

\funding{This work is funded by the Deutsche Forschungsgemeinschaft (DFG, German Research Foundation) –- GRK 2236/1.}

\acknowledgements{I would like to thank Lasse Wulf, Marc Goerigk and Stefan Lendl for helpful comments and discussions on the paper especially concerning the definitions in the preliminaries. Furthermore, I would like to thank the anonymous reviewers for their feedback.}

\nolinenumbers 

\hideLIPIcs  

\EventEditors{John Q. Open and Joan R. Access}
\EventNoEds{2}
\EventLongTitle{42nd Conference on Very Important Topics (CVIT 2016)}
\EventShortTitle{CVIT 2016}
\EventAcronym{CVIT}
\EventYear{2016}
\EventDate{December 24--27, 2016}
\EventLocation{Little Whinging, United Kingdom}
\EventLogo{}
\SeriesVolume{42}
\ArticleNo{23}

\begin{document}

\newcommand{\NP}{\textsl{NP}}
\newcommand{\PSPACE}{\textsl{PSPACE}}
\newcommand{\PTIME}{\textsl{PTIME}}

\newcommand{\CombinatorialProblem}{\ensuremath{P_{C}}}
\newcommand{\CombinatorialGraphProblem}{\ensuremath{P_{CG}}}
\newcommand{\hdrr}[1][P_A]{\ensuremath{#1^{H\hspace{-0.5mm}D\hspace{-0.5mm}R\hspace{-0.5mm}R}}}
\newcommand{\mhdrr}[1][P_A]{\ensuremath{#1^{m\hspace{-0.05mm}\text{-}\hspace{-0.25mm}H\hspace{-0.5mm}D\hspace{-0.5mm}R\hspace{-0.5mm}R}}}
\newcommand{\mplushdrr}[1][P_A]{\ensuremath{#1^{(m\hspace{-0.25mm}+\hspace{-0.25mm}1)\hspace{-0.05mm}\text{-}\hspace{-0.25mm}H\hspace{-0.5mm}D\hspace{-0.5mm}R\hspace{-0.5mm}R}}}
\newcommand{\mminushdrr}[1][P_A]{\ensuremath{#1^{(m\hspace{-0.25mm}-\hspace{-0.25mm}1)\hspace{-0.05mm}\text{-}\hspace{-0.25mm}H\hspace{-0.5mm}D\hspace{-0.5mm}R\hspace{-0.5mm}R}}}

\newcommand{\LiteralGadget}{\ensuremath{G_{\ell}}}
\newcommand{\VariableGadget}{\ensuremath{G_{var}}}
\newcommand{\ClauseGadget}{\ensuremath{G_{c}}}
\newcommand{\ExtensionGadget}{\ensuremath{G_{ext}}}
\newcommand{\IdGadget}{\ensuremath{G_{id}}}
\newcommand{\FakeClauseGadget}{\ensuremath{G_{\hspace{-0.5mm}f\hspace{-0.4mm}c}}}
\newcommand{\DependencyRevealGadget}{\ensuremath{G_{dr}}}

\newcommand{\todo}[1]{\textcolor{red}{\textbf{TODO: #1}}\\}
\newcommand{\todolater}[1]{\textcolor{orange}{\textbf{TODO (later): #1}}\\}
\newcommand{\todolayout}[1]{\textcolor{violet}{\textbf{TODO (layout): #1}}\\}
\newcommand{\cit}{\textsuperscript{[citation needed]}}

\newcommand{\sat}{\textsc{3Sat\-is\-fi\-a\-bil\-i\-ty}}
\newcommand{\sats}{\textsc{3Sat}}
\newcommand{\qsat}{\textsc{TQBF}}
\newcommand{\qsatgame}{\textsc{TQBF Game}}

\newcommand{\ustcon}{\textsc{UstCon}}
\newcommand{\ustconlong}{\textsc{Undirected \allowbreak s-t-Connectivity}}

\newcommand{\vc}{\textsc{Vertex \allowbreak Cover}}
\newcommand{\ds}{\textsc{Dominating \allowbreak Set}}
\newcommand{\fas}{\textsc{Feedback \allowbreak Arc \allowbreak Set}}
\newcommand{\fvs}{\textsc{Feedback \allowbreak Vertex \allowbreak Set}}
\newcommand{\hs}{\textsc{Hitting \allowbreak Set}}

\newcommand{\is}{\textsc{Independent \allowbreak Set}}
\newcommand{\cl}{\textsc{Clique}}

\newcommand{\sss}{\textsc{Subset \allowbreak Sum}}
\newcommand{\ks}{\textsc{Knapsack}}
\newcommand{\pa}{\textsc{Partition}}
\newcommand{\sms}{\textsc{Single \allowbreak Machine \allowbreak Scheduling}}
\newcommand{\tms}{\textsc{Two \allowbreak Machine \allowbreak Scheduling}}

\newcommand{\hp}{\textsc{Ha\-mil\-to\-ni\-an \allowbreak Path}}
\newcommand{\dhp}{\textsc{Directed \allowbreak Ha\-mil\-to\-ni\-an \allowbreak Path}}
\newcommand{\uhp}{\textsc{Undirected \allowbreak Ha\-mil\-to\-ni\-an \allowbreak Path}}
\newcommand{\udhp}{\textsc{(Un)directed \allowbreak Ha\-mil\-to\-ni\-an \allowbreak Path}}
\newcommand{\hc}{\textsc{Ha\-mil\-to\-ni\-an \allowbreak Cycle}}
\newcommand{\dhc}{\textsc{Directed \allowbreak Ha\-mil\-to\-ni\-an \allowbreak Cycle}}
\newcommand{\uhc}{\textsc{Undirected \allowbreak Ha\-mil\-to\-ni\-an \allowbreak Cycle}}
\newcommand{\udhc}{\textsc{(Un)directed \allowbreak Ha\-mil\-to\-ni\-an \allowbreak Cycle}}
\newcommand{\tsp}{\textsc{Traveling \allowbreak Salesman}}

\newcommand{\tdm}{\textsc{3Dimensional \allowbreak Matching}}
\newcommand{\xtc}{\textsc{Exact \allowbreak Cover \allowbreak By \allowbreak 3-Sets}}

\newcommand{\stt}{\textsc{Steiner \allowbreak Tree}}

\newcommand{\tddp}{\textsc{2Dis\-joint \allowbreak Directed \allowbreak Path}}
\newcommand{\kddp}{k\textsc{Dis\-joint \allowbreak Directed \allowbreak Path}}

\newcommand{\col}{\textsc{Coloring}}
\newcommand{\tcol}{\textsc{3Coloring}}
\newcommand{\kcol}{k\textsc{Coloring}}
\newcommand{\cn}{\textsc{Chromatic \allowbreak Number}}
\newcommand{\cc}{\textsc{Clique \allowbreak Cover}}

\maketitle

\begin{abstract}
	The well-known complexity class \NP{} contains combinatorial problems,
whose optimization counterparts are important for many practical settings.
In reality, however, uncertainty in the input data is a usual phenomenon, which is typically not covered in \NP{} problems.

One concept to model the uncertainty in the input data, is \textit{recoverable robustness}.
The instance of the recoverable robust version of a combinatorial problem $P$ is split into a base scenario $\sigma_0$ and an uncertainty scenario set $\textsf{S}$.
The task is to calculate a solution $\texttt{s}_0$ for the base scenario $\sigma_0$ and solutions $\texttt{s}$ for all uncertainty scenarios $\sigma \in \textsf{S}$ such that $\texttt{s}_0$ and $\texttt{s}$ are not too far away from each other according to a distance measure, so $\texttt{s}_0$ can be easily adapted to $\texttt{s}$.

We analyze the complexity of \textit{Hamming distance recoverable robust} versions of problems in \NP{} for different scenario encodings.
The complexity is primarily situated in the lower levels of the polynomial hierarchy.
The main contribution of the paper is a gadget reduction framework that reveals that the recoverable robust version of problems in a large class of combinatorial problems is $\Sigma^p_{3}$-complete.
We show that this class includes over 20 problems such as Vertex Cover, Independent Set, Hamiltonian Path or Subset Sum.
We expect that the number of problems can be easily extended with the help of the gadget reduction framework.
Additionally, we expand the results to $\Sigma^p_{2m+1}$-completeness for multi-stage recoverable robust problems with $m \in \mathbb{N}$ stages.

\end{abstract}

\newpage

\section{Introduction}
The concept of \textit{robustness} in the field of optimization problems comprises a collection of models that consider uncertainties in the input.
These uncertainties may for example arise from faulty or inaccurate sensors or from a lack of knowledge.
Robustness measures can model these types of uncertainty that occur in practical optimization instances into an \textit{uncertainty set}.
The goal is to find solutions that are stable over all possible \emph{scenarios} in the uncertainty set. That is, these solutions remain good but not necessarily optimal regardless what the uncertainties turn out to be in reality.

One specific robustness concept is \textit{recoverable robustness}, which is a recently introduced concept \cite{DBLP:series/lncs/LiebchenLMS09} by Liebchen et al.
The input of a recoverable robust version of a problem $P$ is a \textit{base scenario} $\sigma_0$, which is an instance of problem $P$, as well as a set of \textit{uncertainty scenarios} $\textsf{S}$, whose members are again instances of $P$.
The set of uncertainty scenarios $\textsf{S}$ is the uncertainty set of the problem.
We are asked to compute a base solution $\texttt{s}_0$ to the base scenario $\sigma_0$ and to compute recovery solutions $\texttt{s}$ to all members of the uncertainty scenarios $\sigma \in \textsf{S}$ such that $\texttt{s}_0$ and $\texttt{s}$ are not too far away from each other according to a distance measure.
The solution on the base scenario does not directly include the uncertainties but needs to include the potential to adapt the base solution $\texttt{s}_0$ to solutions $\texttt{s}$ within the given distance between the solutions.
Thus, the base solution $\texttt{s}_0$ may be restricted by these possibly harmful scenarios.

From a worst-case-analysis point of view, we assume that the uncertainty scenarios are chosen by an adversary.
The algorithm computes a base solution with the potential to adapt to all scenarios.
Then, the adversary chooses the most harmful scenario based on the base solution.
Finally, the algorithm computes a recovery solution to adapt to the chosen scenario.

A more general concept is \textit{multi-stage recoverable robustness}, in which not only one set of uncertainty scenarios is provided but $m$ sets of scenarios.
This concept was introduced by Cicerone et al. \cite{DBLP:journals/isci/CiceroneSSS12}.
The \textit{$m$-stage recoverable robust problem} asks to solve the recoverable robust problem on the individual sets of scenarios inductively.
That is, a base solution $\texttt{s}_0$ has to be found such that one can recover from $\texttt{s}_0$ for the first set of scenarios $\textsf{S}_1$ to a solution $\texttt{s}_1$ such that one can recover from $\texttt{s}_1$ for the second set of scenarios $\textsf{S}_2$ and so forth such that one can recover from $\texttt{s}_{m-1}$ for the $m$-th set of scenarios $\textsf{S}_m$ to a solution $\texttt{s}_m$.

\paragraph*{Related Work}
Recoverable robustness is used in many practical settings such as different optimization areas in air transport \cite{DBLP:journals/ors/DijkSP19,DBLP:journals/transci/FroylandMW14,DBLP:journals/cor/MaherDS14} or in railway optimization, for which a survey can be found in \cite{DBLP:journals/eor/LusbyLB18}.
Considered problems in railway optimization are to be found on all stages of railway operation, such as network design \cite{TONISSEN2018360,CADARSO20121288}, rolling stock planning \cite{DBLP:conf/atmos/CacchianiCGKM08,DBLP:journals/transci/CacchianiCGKMT12}, shunting \cite{DBLP:journals/aor/CiceroneDSFN09} and timetabling \cite{DBLP:series/lncs/CiceroneDSFNSS09,DBLP:journals/jco/CiceroneDSFN09,DBLP:journals/tc/DAngeloSNP11,DBLP:conf/iwoca/DAngeloSN09,DBLP:conf/atmos/GoerigkHMSS13,DBLP:series/lncs/Busing09}.
Our focus lies on the complexity of recoverable robust problems.
In parallel to this paper, Goerigk et al. \cite{DBLP:journals/corr/abs-2209-01011} analyzed the Hamming distance recoverable robust independent set, TSP and vertex cover.
Hamming distance means that at most $k$ elements may be added to or deleted from the base solution in total to obtain a recovery solution.
They showed the $\Sigma^p_3$-hardness of the variant with discrete budgeted uncertainty over the costs of the elements.
To the best of the author's knowledge, this is the only contribution investigating the complexity within the polynomial hierarchy beyond \NP-hardness.
All other contributions study primarily algorithms and analyze the problems only on their \NP-hardness or their approximability, where different distance measures between the solutions are of interest.
The concept of \textit{$k$-dist recoverable robustness}, allowing at most $k$ new elements in recovery solutions, was introduced in \cite{DBLP:journals/networks/Busing12} but was also used in \cite{DBLP:journals/jco/HradovichKZ17}.
Besides the $k$-dist measures, there are also measures which limit the number of deleted elements \cite{DBLP:conf/inoc/BusingKK11} or exchanged \cite{DBLP:journals/ol/ChasseinG16} elements.
Furthermore, combinations of these distance measures are analyzed as well in the literature \cite{DBLP:journals/ol/BusingKK11}.
Further usages of Hamming distance recoverable robustness can be found in \cite{DBLP:journals/networks/DouradoMPRPA15}.
Among the studied recoverable robust problems is Knapsack, which is \NP-hard for different distance measures between the solutions \cite{DBLP:journals/ejco/BusingGKK19,DBLP:conf/inoc/BusingKK11,DBLP:journals/ol/BusingKK11}.
Recoverable robust versions of problems that are in \PTIME{} are shown to be \NP-complete as well such as Shortest Path, which is \NP-hard for $k$-dist \cite{DBLP:journals/networks/Busing12}, or Matching \cite{DBLP:journals/networks/DouradoMPRPA15}.
Furthermore, the recoverable robust Single Machine Scheduling problem is 2-approximable \cite{DBLP:journals/dam/BoldG22} and the recoverable robust TSP is 4-approximable \cite{DBLP:journals/ol/ChasseinG16}.
Moreover, a recoverable robust version of Spanning Tree \cite{DBLP:journals/jco/HradovichKZ17} is shown to be in \PTIME.

For the complexity analysis, we introduce a gadget reduction framework.
Different gadget reduction concepts were studied for example by Agrawal et al.~\cite{DBLP:journals/cc/AgrawalAIPR01}, who defined gadget reductions under $\sf AC^0$ for NP-completeness mapping one bit of the input of one problem to a bounded number of bits in the other problem.
A further form of gadget reduction was introduced by Trevisan et al.~\cite{DBLP:conf/focs/TrevisanSSW96}.
They formalize gadgets with constraint families to compute optimal gadgets via linear programming for gap-preserving reductions.

\paragraph*{Contribution}
We study Hamming distance recoverable robust problems with different forms of elemental uncertainty.
That is, it is uncertain whether an element (e.g. a vertex or object) is included in a scenario or not.
This form of uncertainty is different to cost uncertainty, where all elements are present in all scenarios but the costs of the elements are uncertain.
We show that recoverable robust versions of typical \NP-complete combinatorial problems with $xor$-dependencies or $\Gamma$-set scenarios are $\Sigma^p_{3}$-complete and the corresponding multi-stage recoverable robust versions are $\Sigma^p_{2m+1}$-complete, where $m \in \mathbb{N}$ is the number of stages.

We do this by defining a gadget reduction framework, which uses a specific definition of combinatorial problems.
These problems are defined over \emph{combinatorial elements}, which are defined over a \emph{universe} $U$, and \emph{nested relations} $R(U)$ over that universe.
We show that this framework is able to ``upgrade'' many already exsiting NP-hardness reductions by applying it to over 20 well-known problems.
Thus, we expect that the results are easily extendable beyond those problems.

In order to explain the idea of these \emph{universe gadget reductions}, consider a problem $A$ for which we want to reduce to another problem $B$.
A gadget reduction creates a gadget for each combinatorial element from $U_A$ and $R_A(U_A)$ to simulate the behavior of this element in $B$.
This gadget consists of universe elements from $U_B$ or relation elements from $R_B(U_B)$ and is disjoint from all other gadgets.
That is, no element from $B$ is in two gadgets at the same time.
Additionally, we demand that if we remove a combinatorial element in $A$, we are able to remove the corresponding gadget in $B$ without invalidating the correctness of the reduction.
This form of reduction preserves the scenarios structurally independent of the underlying encoding.
Thus, this gadget reduction framework allows for reductions between Hamming distance recoverable robust problems.
Indeed these properties are already achieved by typical polynomial reductions (or slight modifications) of it.

\paragraph*{Paper Summary}
In \Cref{sec:preliminaries}, we define necessary complexity theoretical concepts.
In \Cref{sec:combinatorialProblemFramework}, we build a framework for combinatorial decision problems to define Hamming distance recoverable robust problems.
Then, we consider typical problems, which are in \NP{} or \NP-complete, and analyze their complexity for polynomially computable scenario encodings in \Cref{sec:sigma1}.
\Cref{sec:sigma3} consists of the complexity analysis of succinctly encoded scenarios as well as multi-stage recoverable robustness.
At last in \Cref{sec:class}, we establish a whole class of Hamming distance recoverable robust problems by using our combinatorial decision problem framework and by introducing universe gadget reductions.
With \Cref{sec:conclusion}, we conclude the paper.
\section{Preliminaries}\label{sec:preliminaries}
We define a \emph{language} $L$ as a subset of $\{0,1\}^*$.
The class $\Sigma^k_p$ contains all languages $L$ such that there is a Turing machine $V$ (the ``verifier'') and polynomial $p$ such that for all $x \in \{0,1\}^*$, it holds $x \in L$ iff
$$
	\exists y_1 \in \{0,1\}^{p(|x|)} \forall y_2 \in \{0,1\}^{p(|x|)} \ldots  Q y_k \in \{0,1\}^{p(|x|)} \ V(x, y_1, y_2, \ldots, y_k) = 1,
$$
where $Q = \exists$, if $k$ odd, and $Q=\forall$, else.
This family of classes is part of the \emph{polynomial-time hierarchy} defined by Stockmeyer \cite{DBLP:journals/tcs/Stockmeyer76}.
A \emph{many-one reduction} (or \emph{Karp reduction}) from some language $L_1$ to some language $L_2$ is defined as a function $f: \{0,1\}^* \rightarrow \{0,1\}^*$ such that $x \in L_1$ iff $f(x) \in L_2$ for all $x \in \{0,1\}^*$.
A language $L_1$ is $\Sigma^k_p$-hard if all languages $L_2 \in \Sigma^k_p$ can be reduced to $L_1$ via a polynomial time many-one reduction.
A problem is $\Sigma^k_p$-complete if it is contained in $\Sigma^k_p$ and $\Sigma^k_p$-hard.
The canonical complete problems for $\Sigma^k_p$ are $\exists_1\forall_2\ldots\exists_k\textsc{CNF-Sat}$,  for odd $k$, and $\exists_1\forall_2\ldots\forall_k\textsc{DNF-Sat}$, for even $k$.

\section{Combinatorial Problem Framework}\label{sec:combinatorialProblemFramework}
In theoretical computer science, problems are defined as languages, which consist of all YES-instances of the problem.
The instances are encoded as words from $\{0,1\}^*$.
For combinatorial problems, we may assume that an instance contains a universe $U = \{1, \ldots, n\}$, which consists of the encoding atoms of the instance.
Furthermore, an instance includes (nested) relations between these atoms.
To encode the relations, the atoms are used together with a delimiter symbol.

One example of such a problem is the problem \textsc{Undirected \allowbreak s-t-Connec\-tiv\-i\-ty} (\ustcon).
Its input is an undirected graph $G = (V,E)$ together with two vertices $s, t \in V$.
The corresponding instance is then encoded by the vertices $V=U$ as universe and three relations $s,t \subseteq V$ and $E \subseteq V\times V$.
The instance is a YES-instance iff there is path from $s$ to $t$ in $G$.
Another example is the problem \textsc{Vertex Cover}.
Again, the vertices $V=U$ are the universe and $E \subseteq V\times V$ is a relation.
The instance is a YES-instance iff there is a small vertex cover in $G$.

In mathematical optimization, a problem is often defined over its feasible solutions $F$ together with a cost function $c$.
The goal is then to find a solution that achieves the minimum (resp. maximum) costs of all feasible solutions.
Oftentimes, an additional ground set of combinatorial elements $X$ is given.
For simplicity, the feasible solutions are then combinations of that ground set, that is $F(X) \subseteq 2^X$.
We apply this to \ustcon{} by interpreting the edges as the ground set $X = E$ and all paths $F(X) \subseteq 2^E$ from $s$ to $t$ as the feasible solutions.
For \textsc{Vertex Cover}, we define the vertices as ground set $X = V$ and the feasible solutions $F(X) \subseteq 2^V$ are all small vertex covers in the graph.
For simplicity, we ignore cost or weight functions and ask for the mere existence of a solution (here: a path, a small vertex cover).

While this is not a general definition, many typical combinatorial problems can be defined this way such as \is{} (an independent set is a subset of vertices), \hp{} (a Hamiltonian path is a subset of edges), \sss{} (a solution for subset sum is a subset of numbers).

We distinguish the natural encoding universe $U$ from the solution ground set $X$ over which the solutions are defined.
With that, we reach a larger class of problems.
In \vc, the encoding universe $U =V$ is the same as the solution ground set $X=V$, because a vertex cover is a set of vertices and a graph is a set of vertices which are in relation via edges.
In contrast, the instances of \ustcon{} are still graphs while the solutions are subsets of edges.
Thus for \ustcon{}, the solution ground set and the universe do not coincide.

We begin with the definition of nested relations in order to define the instances of combinatorial problems.
With these nested relations, we are able to define all possible associations of universe elements as well as between universe elements and relational elements.
Thus in a graph $G = (V,E)$, we are not only able to for example encode edges $E \subseteq V \times V$ but also an incidence relation $I \subseteq V \times E$ or the neighborhood relation $N \subseteq V^{\leq|V|}$.

\begin{definition}[Nested Relations]\label{intro:def:relationSet}
    Let $U$ be a set.
    Then $\mathcal{R}(U)$ is the set of nested relations over $U$ defined by the smallest set fulfilling:
    \begin{align}
        U &\in \mathcal{R}(U) & \\
        A &\in \mathcal{R}(U), & \text{if} \ A \subseteq B \ \text{for some B} \ \in \mathcal{R}(U)\\
        \bigtimes_i A_i &\in \mathcal{R}(U), & \text{if for all} \ i, \ A_i \in \mathcal{R}(U)
    \end{align}
    We denote the set of relation elements that include $r \in A \in \mathcal{R}(U)$ by $R(r)$.
\end{definition}

With access to all nested relations over the universe, we are able to define not only a variety of problems but we are also able to meaningfully define gadget reductions between problems.
The solution ground set $X=R$ is then a subset of relational elements of one (nested) relation $R \in \mathcal{R}(U)$ over the gadget reduction universe $U$.
Thus the solutions are of the form $F(R) \subseteq 2^{R}$.

\begin{definition}[Combinatorial Decision Problem]\label{intro:def:combinatorialProblem}
	A combinatorial decision problem $P_A$ is a set of tuples $(U_A, R_A, F_A(R_A))$ with the set of universe elements $U_A$, relations $R_A \in \mathcal{R}(U_A)^r, r \in \mathbb{N}$, and the set of feasible solutions $F_A(R_A) \subseteq 2^{R^i_A}$ for some $1 \leq i \leq r$.
    We assume that $R^1_A =U_A$.
	We call $R_A$ the instance of the problem and $R_A$ is a YES-instance if and only if $F_A(R_A) \neq \emptyset$.
	We use an index set $I_A$ to easily address the members of the tuple $R_A$.
\end{definition}

For simplicity, we may omit the problem in the index of $U_A$, $R_A$ and $F_A(R_A)$ as well as the dependence of the feasible solutions $F(R)$ on the relations $R$ and write $F$.
For a better understanding, we again use \ustcon{} as an example.

\begin{example}[Undirected $s$-$t$-Connectivity Problem]
  The \textit{input} of \ustcon{} is a graph $G = (V, E)$ and two vertices $s, t \in V$.
  A \textit{feasible solution} is a path from $s$ to $t$ in $G$.
  This translates to the following tuple $(U, R, F)$.
  The universe $U$ consists of the vertices $V$.
  The relations in $R$ are the edges $E$ and the vertices $s$ and $t$, that is, $R = (V, E, s, t)$.
  The feasible solutions are all $s$-$t$-paths $p \in F \subseteq 2^E$ in $G$ defined as subsets of edges.
\end{example}

Observe that for combinatorial problems, the encoding of the input and the solutions depends only on the universe of elements.
Thus, the universe elements in $U$ build the atoms of the problem.
The (nested) relations $R$ model the relations between these atoms.
The feasible solutions $F$ model all possible combinations of solution elements that are feasible.

\subsection{Scenarios for Robust Problems}
Before we are able to define recoverable robust problems, we need to define scenarios.
Scenarios are a central concept in robust optimization, which model the uncertainty.
A Hamming distance recoverable robust problem $\hdrr[P_A]$ is based on a combinatorial problem $P_A$.
We then define a scenario as follows.

\begin{definition}[Scenarios]
    A scenario of the Hamming distance recoverable robust problem \hdrr{} is a problem instance $(U_A, R_A, F_A(R_A))$ of the base problem $P_A$.
\end{definition}

\paragraph*{Encoding of Scenarios}
For scenarios, we use explicit encodings, implicit encodings or succinct encodings.
We consider elemental uncertainty, for which it is uncertain whether a combinatorial element is part of a scenario or not.
Thus, all of these encodings are based on combinatorial elements of an instance, which include the universe and all relation elements.
This is different to uncertainty over the costs of elements, where the underlying combinatorial elements remain the same for all scenarios.
If a combinatorial element is not part of a scenario, then all relation elements that include this combinatorial element are discarded as well in the scenario.
For example, if a vertex $v$ in a graph problem is discarded, then all edges incident to $v$ are discarded, too.
We denote this removal of combinatorial elements with $U \setminus \{r\}$ and $R \setminus R(r)$, where the removal of $r$ removes all relation elements $R(r)$ that contain $r$.
We call the elements that are part of the current scenario the \textit{active} elements, otherwise we call the elements \textit{inactive}.

First, we will use explicit encodings by providing the complete instance encoding over the base problem $P_A$.
Additionally, we use implicit encodings by providing a set of all elements that are different from base scenario $\sigma_0$.
Furthermore, we address succinct encodings of scenarios as well.
These encodings usually encode an exponential number of scenarios in polynomial space.
The popular concept of discrete budgeted uncertainty, also known as $\Gamma$-scenarios, \cite{bertsimas2004robust} falls into this last category as well as later defined $xor$-dependencies, which use logical operators between the elements to encode which element is active, i.e. part of a scenario.

\subsection{Hamming Distance Recoverable Robust Problems}

Now, we define Hamming distance recoverable robust problems.
For this, we need a definition of the Hamming distance over a set.

\begin{definition}[Hamming Distance of Sets]
    Let $A, B$ be two sets. Then, we define the Hamming distance $H(A, B)$ of set $A$ and $B$ to be
    $$
        H(A, B) := |A \vartriangle B | = |\{x \mid \text{either } x \in A \text{ or } x \in B \}|
    $$
\end{definition}

Intuitively, a Hamming distance recoverable robust problem \hdrr{} is based on a nominal combinatorial decision problem $P_A$, e.g. \ustcon{}.
We distinguish the \textit{base scenario} from \textit{uncertainty scenarios}.
The base scenario $\sigma_0$ is the instance on which the first solution $\texttt{s}_0$ has to be computed.
The uncertainty scenarios $\sigma \in \textsf{S}$ are the scenarios for which the solution $\texttt{s}$, that has to be adapted from $\texttt{s}_0$, have to be computed.
All scenarios of a problem may share universe elements or relation elements.
In conclusion, we not only have to find a solution for one instance, but for one base scenario $\sigma_0$ and for all uncertainty scenarios in $\textsf{S}$.
That is, we can recover from every possible scenario with a new solution to the problem.
The solutions to the uncertainty scenarios, nonetheless, may have a Hamming distance of at most $\kappa$ to the solution of the base scenario.
We always define the Hamming distance over the solution ground set $X$ between the solutions from $F(X) \subseteq 2^X$.
Formally, we obtain the following definition.

\begin{definition}[Hamming Distance Recoverable Robust Problem]
    A Hamming distance recoverable robust problem \hdrr{} is a combinatorial problem based on a combinatorial problem $P_A$. \hdrr{} is defined as a set of tuples $(U, R, F(R))$ with
    \begin{description}
        \itemsep0em
        \item $U = U_0 \cup \bigcup_{\sigma \in \textsf{S}} U_\sigma$ is the universe. The universe is the union over all universe elements that occur in the scenarios.
        \item $R = (R_0, (R_\sigma)_{\sigma\in\textsf{S}}) = ((U_0, R^2_0, \ldots, R^r_0), (U_\sigma, R^2_\sigma, \ldots, R^r_\sigma)_{\sigma \in \textsf{S}})$ are the relations. The relations are separate for each scenario.
        \item $F(R) = \{(\texttt{s}_0, (\texttt{s}_\sigma)_{\sigma \in \textsf{S}}) \in  F_0(R_0) \times (F_\sigma(R_\sigma))_{\sigma\in\textsf{S}} \mid H(\texttt{s}_0, \texttt{s}_\sigma) \leq \kappa \ \text{for all} \ \sigma \in \textsf{S}\})$ are the feasible solutions.
    The Hamming distance $H(\texttt{s}, \texttt{s}')$ is defined over the elements in the solutions $\texttt{s}, \texttt{s}'$.
    \end{description}
\end{definition}

The feasible solutions are not subsets of some relation $R$ but consist of tuples including the solution for each scenario in $F$, which also adhere to the Hamming distance.
In general, we assume that the bound on the Hamming distance $\kappa$ is part of the input.

Observe that the specifications are no restriction because every decision problem can be formulated as one base scenario and no uncertainty scenarios, that is $\textsf{S} = \emptyset$.
On the other hand, the base problem $P_A$ is a restriction of \hdrr{} by setting $\textsf{S} = \emptyset$.
Furthermore, the base scenario is defined by $\sigma_0 = (U_0,R_0,F_0)$ and all uncertainty scenarios $\sigma \in \textsf{S}$ are defined by $\sigma = (U_\sigma,R_\sigma,F_\sigma)$.
Again, we provide an example for a better understanding of the definition and again, we use the problem \ustcon{}.

\begin{example}[Hamming Distance Recoverable Robust \ustcon{}]
	Let $G = (V, E)$ be a graph, $s, t \in V$ and $\kappa \in \mathbb N$.
	\hdrr[\ustcon] is a Hamming distance recoverable robust problem with feasible solutions $F\subseteq 2^E$.
	Thus, the Hamming distance is defined over the edges.
	The start and end vertices $s$ and $t$ remain the same for all scenarios.
	The input $R$ contains the following: Each scenario $\sigma \in \textsf{S}$ encodes the set of active vertices $V_\sigma$ and edges $E_\sigma$.
	The feasible solutions $F$ consists of all $s$-$t$-paths $(p_0, p_{\sigma \in \textsf{S}}) \in 2^{E_{\sigma_0}} \times 2^{E_{\sigma \in \textsf{S}}}$ such that $H(p_0, p_\sigma) \leq \kappa, \text{ for all } \sigma \in \textsf{S}$.
	In other words, the question is
	$$
		\exists p_0 \in 2^{E_{\sigma_0}}: \ \forall {\sigma \in \textsf{S}}: \ \exists p_{\sigma} \in 2^{E_{\sigma}}: \ p_0 \in F_0, \ p_\sigma \in F_\sigma \ \text{and} \ H(p, p_\sigma) \leq \kappa.
	$$
\end{example}

\subsection{Combinatorial Problems with Partitions as Solutions}
As already stated in the introduction of this section, \Cref{intro:def:combinatorialProblem} is not a general definition for combinatorial decision problems.
For example coloring (which asks for an independent set cover) or clique cover as well as many other problems are not covered by this definition, because these have partitions as solutions and not subsets of some relation.
In order to meaningfully integrate these kind of problems into this framework, we need to adapt the definition of combinatorial decision problems as well as the definition for the Hamming distance between solutions because the solutions are not sets but partitions.

\begin{definition}[Combinatorial Decision Problem with Partition Solutions]\label{intro:def:combinatorialProblemPartition}
	A combinatorial decision problem with partition solutions $P_A$ is a set of tuples $(U_A, R_A, F_A(R_A))$ with the set of universe elements $U_A$, relations $R_A \in \mathcal{R}(U_A)^r, r \in \mathbb{N}$, and the set of feasible solutions $F_A(R_A) \subseteq (2^{R^i_A})^k$ which are $k$-partitions of $R^i_A$ for some $1 \leq i \leq r$.
	We assume that $R^1_A =U_A$.
	We call $R_A$ the instance of the problem and $R_A$ is a YES-instance if and only if $F_A(R_A) \neq \emptyset$.
	We use an index set $I_A$ to easily address the members of the tuple $R_A$.
\end{definition}

The only change in the definition in comparison to \Cref{intro:def:combinatorialProblem} is that the set of feasible solutions is defined as $F_A(R_A) \subseteq (2^{R_A^i})^k$ such that $F_A(R_A)$ consists of $k$-partitions.
We then define the Hamming distance between two $k$-partitions to be the sum of the Hamming distances of the sets of the two partitions.

\begin{definition}[Hamming Distance of Partitions]
	Let $A = (A_1,\ldots,A_k) \subseteq S^k$ and $B = (B_1, \ldots, B_k) \subseteq S^k$ be two $k$-partitions of the set $S$.
	Then, we define the Hamming distance $H_P(A, B)$ of partitions $A$ and $B$ to be
    	$$
        		H_P(A, B) := \sum^k_{i=1} H(A_i,B_i),
    	$$
	where $H(A_i,B_i)$ is the Hamming distance over the sets $A_i, B_i$.
\end{definition}
	
	Accordingly, we also define Hamming distance recoverable robust versions of combinatorial problems with partitions as solutions.

\begin{definition}[Hamming Distance Recoverable Robust Problem with partition solutions]\hfill\\
	A Hamming distance recoverable robust problem \hdrr{} is a combinatorial problem based on a combinatorial problem $P_A$.
	\hdrr{} is defined as a set of tuples $(U, R, F(R))$ with
    	\begin{description}
		\itemsep0em
		\item $U = U_0 \cup \bigcup_{\sigma \in \textsf{S}} U_\sigma$ is the universe. The universe is the union over all universe elements that occur in the scenarios.
		\item $R = (R_0, (R_\sigma)_{\sigma\in\textsf{S}}) = ((U_0, R^2_0, \ldots, R^r_0), (U_\sigma, R^2_\sigma, \ldots, R^r_\sigma)_{\sigma \in \textsf{S}})$ are the relations. The relations are separate for each scenario.
		\item $F(R) = \{(\texttt{s}_0, (\texttt{s}_\sigma)_{\sigma \in \textsf{S}}) \in  F_0(R_0) \times (F_\sigma(R_\sigma))_{\sigma\in\textsf{S}} \mid H_P(\texttt{s}_0, \texttt{s}_\sigma) \leq \kappa \ \text{for all} \ \sigma \in \textsf{S}\})$ are the feasible solutions.
		The Hamming distance $H_P(\texttt{s}, \texttt{s}')$ is defined over the solution partitions $\texttt{s}, \texttt{s}'$.
	\end{description}
\end{definition}

All of the following results on Hamming distance recoverable robust problems also hold for these kinds of problems.

\section{Recoverable Robust Problems with Polynomially Computable Scenario Encodings}\label{sec:sigma1}
We now consider problems with polynomially computable scenario encodings.
A scenario encoding is polynomially computable if the set of scenarios is transformable into a set of explicitly encoded instances in polynomial time.
(Consequently, the number of scenarios is bounded by a polynomial.)

\begin{lemma}\label{theorem:inNP}
Let $P_A \in \NP$. Then \hdrr{} $\in \NP$ if the set of scenarios $S$ of \hdrr{} is polynomially computable.
\end{lemma}
\begin{proof}
We construct a polynomial time verifier that receives the instance as well as an ($\exists$-quantified) string $y_1$ as input.
The string $y_1$ represents the solution $\texttt{s}_0$ to the base scenario and the solution to the uncertainty scenarios $(\texttt{s}_\sigma)_{\sigma \in \textsf{S}}$.
These are encoded as list of active combinatorial elements.
Because the number of scenarios is polynomially bounded in the length of the input and the solution to a scenario is a subset of the active elements, the string $y_1$ is also polynomially bounded in the input length.
Furthermore, the certificate is verifiable in polynomial time by the following algorithm.
First we compute the explicit encoding of all scenarios in polynomial time.
We then verify whether $\texttt{s}_0$ is a solution to $\sigma_0$ and whether $\texttt{s}_\sigma$ is a solution to $\sigma$ for all $\sigma \in \textsf{S}$.
This is doable by using the existing verifier for the base problem that exists because the problem is in \NP.
At last, we check $H(\texttt{s}_0, \texttt{s}_\sigma) \leq \kappa$ for all $\sigma \in \textsf{S}$.
\end{proof}

Besides general polynomially computable scenarios, we may consider the popular concept of $\Gamma$-scenarios.
These consist of all scenarios that deviate in at most $\Gamma$ many elements from the base instance corresponding to a set of activatable elements.
If $\Gamma$ is constant, we may use \Cref{theorem:inNP} to obtain the following result.

\begin{corollary}\label{corollary:Gamma}
Let $P_A \in \NP$.
Then \hdrr{} $\in \NP$ if the set of scenarios $\textsf{S}$ of \hdrr{} consists of all possible $\Gamma$-scenarios for a constant $\Gamma$.
\end{corollary}

The following theorem follows from \Cref{theorem:inNP} and \Cref{theorem:inNP} by reusing the original reduction to $P_A$ and setting the scenario set $\textsf{S} = \emptyset$.

\begin{theorem}\label{lemma9}
    Let $P_A$ be an \NP-complete problem. Then, \hdrr{} is \NP-complete if the set of scenarios $\textsf{S}$ of \hdrr{} is polynomially computable.
\end{theorem}
\begin{proof}
    The reduction from $P_A$ is trivial because the scenarios can be set to $\textsf{S} = \emptyset$ showing the hardness of \hdrr{}.
    On the other hand, \Cref{theorem:inNP} proves the containment.
\end{proof}
	
\subsection{Reduction for \ustconlong{}}

\begin{theorem}
There is a deterministic logarithmic space computable reduction from \sat{} to \hdrr[\ustcon] with one base and one uncertainty scenario.
\end{theorem}
\begin{proof}
First of all, there is a reduction from \sat{} to \dhc{} presented by Arora and Barak \cite{DBLP:books/daglib/0023084} and a reduction from \dhc{} to \uhc{} presented by Karp \cite{DBLP:conf/coco/Karp72}.
These reductions are computable in logarithmic space.
We use these reductions to develop a reduction from \sat{} to \hdrr[\ustcon{}].

We can either define the scenarios over vertices or over edges.
This, however, is in this reduction realm equivalent, because we can easily introduce a vertex for every edge, such that for the deletion of such an vertex the former edge is deleted.
On the other hand, we can delete all incident edges of a vertex to exclude the vertex from a possible solution.
For the sake of simplicity, we use edge scenarios in the reduction.
Accordingly, the Hamming distance of the solution is based on the edges.

We now provide a reduction from \uhc{} to \hdrr[\ustcon{}].
Let $G = (V, E)$ be a graph of the \uhc{} instance.
We map the graph $G$ to a graph $G'$, a base scenario $\sigma_0$ and a uncertainty scenario $\sigma_1$, which together define the \hdrr[\ustcon{}] instance.
A simple example instance, which we use for explaining the construction, can be found in \Cref{fig:Sigma1:ExampleGraph}.

\begin{figure}[!ht]
    \centering
    \scalebox{1.382}{
				\begin{tikzpicture}[scale=1,
						node/.style = {shape=circle, draw, inner sep=0pt, minimum size=0.25cm},
						dashednode/.style = {shape=circle, draw=blue, dashed, inner sep=0pt, minimum size=0.25cm},
						textnode/.style = {shape=circle, draw, inner sep=0pt, minimum size=0.4cm},
						dashedtextnode/.style = {shape=circle, draw=blue, dashed, inner sep=0pt, minimum size=0.4cm},
						box/.style = {rectangle, fill=gray!20, rounded corners, fill opacity=1, inner sep=1pt}]
            \node[node] (v1) at (0,0) {}; \node at ($(v1) + (-0.33,0)$) {$1$};
            \node[node] (v2) at (1,0) {}; \node at ($(v2) + (0.33,0)$) {$2$};
            \node[node] (v3) at (1,1) {}; \node at ($(v3) + (0.33,0)$) {$3$};
            \node[node] (v4) at (0,1) {}; \node at ($(v4) + (-0.33,0)$) {$4$};
            
            \path [-, black, thick] (v1) edge[bend left=0] node[above] {} (v2);
            \path [-, black, thick] (v2) edge[bend left=0] node[above] {} (v4);
            \path [-, black, thick] (v4) edge[bend left=0] node[above] {} (v3);
            \path [-, black, thick] (v3) edge[bend left=0] node[above] {} (v1);
        \end{tikzpicture}
    }
    \caption{Example Instance $G$ for \uhc}
    \label{fig:Sigma1:ExampleGraph}
\end{figure}

First, all vertices $v \in V$ are duplicated $|V| + 3$ times to connect them to one path that includes $|V| + 2$ edges.
Let $v_i^a$ and $v_i^b$ the end vertices of a vertex path of vertex $v_i$.
In \Cref{fig:Sigma1:NodeDuplication}, the duplication procedure is depicted.
We call these \textit{vertex paths}.

\begin{figure}[!ht]
    \centering
    \begin{subfigure}[b]{0.48\textwidth}
        \centering
        \scalebox{1.382}{
				\begin{tikzpicture}[scale=1,
						node/.style = {shape=circle, draw, inner sep=0pt, minimum size=0.25cm},
						dashednode/.style = {shape=circle, draw=blue, dashed, inner sep=0pt, minimum size=0.25cm},
						textnode/.style = {shape=circle, draw, inner sep=0pt, minimum size=0.4cm},
						dashedtextnode/.style = {shape=circle, draw=blue, dashed, inner sep=0pt, minimum size=0.4cm},
						box/.style = {rectangle, fill=gray!20, rounded corners, fill opacity=1, inner sep=1pt}]
                    \node[node] (v1) at (0,0) {}; \node at ($(v1) + (-0.33,0)$) {$1$};
                    \node[node] (v2) at (1,0) {}; \node at ($(v2) + (0.33,0)$) {$2$};
                    \node[node] (v3) at (1,1) {}; \node at ($(v3) + (0.33,0)$) {$3$};
                    \node[node] (v4) at (0,1) {}; \node at ($(v4) + (-0.33,0)$) {$4$};
                \end{tikzpicture}
        }
        \subcaption{Example instance vertices without edges.}
    \end{subfigure}
    \hfill
    \begin{subfigure}[b]{0.48\textwidth}
        \centering
        \scalebox{1.382}{
    \begin{tikzpicture}[scale=0.33,
    		node/.style = {shape=circle, draw, inner sep=0pt, minimum size=0.2cm},
    		dashednode/.style = {shape=circle, draw=blue, dashed, inner sep=0pt, minimum size=0.25cm},
    		textnode/.style = {shape=circle, draw, inner sep=0pt, minimum size=0.4cm},
    		dashedtextnode/.style = {shape=circle, draw=blue, dashed, inner sep=0pt, minimum size=0.4cm},
    		box/.style = {rectangle, fill=gray!20, rounded corners, fill opacity=1, inner sep=1pt}]
        \node at (1,-3.5) {$1$};
        \node at (11,-3.5) {$2$};
        \node at (11,3.5) {$3$};
        \node at (1,3.5) {$4$};
        
        \node[node] (v11) at (0,-1.0) {};
        \node[node] (v12) at (0.5,-2.0) {};
        \node[node] (v13) at (1.5,-2.5) {};
        \node[node] (v14) at (2.5,-3.0) {};
        \node[node] (v15) at (3.5,-3.5) {};
        \node[node] (v16) at (5,-3.5) {};
        
        \path [-, black, thick] (v11) edge[bend left=0] node[above] {} (v12);
        \path [-, black, thick] (v12) edge[bend left=0] node[above] {} (v13);
        \path [-, black, thick] (v13) edge[bend left=0] node[above] {} (v14);
        \path [-, black, thick] (v14) edge[bend left=0] node[above] {} (v15);
        \path [-, black, thick] (v15) edge[bend left=0] node[above] {} (v16);
        
        \node[node] (v21) at (7,-3.5) {};
        \node[node] (v22) at (8.5,-3.5) {};
        \node[node] (v23) at (9.5,-3.0) {};
        \node[node] (v24) at (10.5,-2.5) {};
        \node[node] (v25) at (11.5,-2.0) {};
        \node[node] (v26) at (12,-1.0) {};
        
        \path [-, black, thick] (v21) edge[bend left=0] node[above] {} (v22);
        \path [-, black, thick] (v22) edge[bend left=0] node[above] {} (v23);
        \path [-, black, thick] (v23) edge[bend left=0] node[above] {} (v24);
        \path [-, black, thick] (v24) edge[bend left=0] node[above] {} (v25);
        \path [-, black, thick] (v25) edge[bend left=0] node[above] {} (v26);
            
        \node[node] (v31) at (7,3.5) {};
        \node[node] (v32) at (8.5,3.5) {};
        \node[node] (v33) at (9.5,3.0) {};
        \node[node] (v34) at (10.5,2.5) {};
        \node[node] (v35) at (11.5,2.0) {};
        \node[node] (v36) at (12.0,1.0) {};
        
        \path [-, black, thick] (v31) edge[bend left=0] node[above] {} (v32);
        \path [-, black, thick] (v32) edge[bend left=0] node[above] {} (v33);
        \path [-, black, thick] (v33) edge[bend left=0] node[above] {} (v34);
        \path [-, black, thick] (v34) edge[bend left=0] node[above] {} (v35);
        \path [-, black, thick] (v35) edge[bend left=0] node[above] {} (v36);
            
        \node[node] (v41) at (0,1.0) {};
        \node[node] (v42) at (0.5,2.0) {};
        \node[node] (v43) at (1.5,2.5) {};
        \node[node] (v44) at (2.5,3.0) {};
        \node[node] (v45) at (3.5,3.5) {};
        \node[node] (v46) at (5,3.5) {};
        
        \path [-, black, thick] (v41) edge[bend left=0] node[above] {} (v42);
        \path [-, black, thick] (v42) edge[bend left=0] node[above] {} (v43);
        \path [-, black, thick] (v43) edge[bend left=0] node[above] {} (v44);
        \path [-, black, thick] (v44) edge[bend left=0] node[above] {} (v45);
        \path [-, black, thick] (v45) edge[bend left=0] node[above] {} (v46);
    \end{tikzpicture}
}
        \subcaption{Multiplied vertices connected to vertex paths.}
    \end{subfigure}
    \caption{Duplication of Nodes}
    \label{fig:Sigma1:NodeDuplication}
\end{figure}

We now define the base scenario $\sigma_0$ and the uncertainty scenario $\sigma_1$.
For the base scenario $\sigma_0$, we design a simple to solve instance, which forces the solution to include all edges of the vertex paths.
For this, we connect the vertex paths to a simple cycle by introducing an edge connecting two vertex paths.
That is, we introduce edges
$$
    \left\{v^i_b, v^{(i+1 \ \text{mod} \ |V(G)|)}_a\right\} \ \text{for} \ 1 \leq i \leq |V(G)|.
$$
The simple to find solution is this cycle.
We further have to introduce two vertices $s$ and $t$.
For this, we choose one vertex path, delete the edge in the middle of the path and designate the incident vertices of the delete edge as $s$ and $t$.
The base scenario $\sigma_0$ can be found in \Cref{fig:Sigma1:BaseScenario}.

For the uncertainty scenario $\sigma_1$, we deactivate the edges between the vertex paths but not the vertex paths themselves.
We then set $\kappa = |V(G)|$.
This forces the solution to the uncertainty scenario to include the vertex paths, as only $|V(G)|$ edges can be altered while a vertex path has at least $|V(G)| + 1$ (including the one with $s$ and $t$).
Furthermore, we map and activate the actual edges of $G$.
For this, we quadruplicate the edges of $G$. The edge $\{v_i, v_j\}$ is quadruplicated to $\{v_i^a, v_j^a\}$, $\{v_i^a, v_j^b\}$, $\{v_i^b, v_j^a\}$ and $\{v_i^b, v_j^b\}$.
Thus, each vertex path can be ordered in both ways in a possible Hamiltonian cycle.
This is depicted in \Cref{fig:Sigma1:RecoveryScenario}.

\begin{figure}[!ht]
  \centering
  \begin{subfigure}[b]{0.48\textwidth}
    \centering
    \scalebox{1.382}{
    \begin{tikzpicture}[scale=0.33,
    		node/.style = {shape=circle, draw, inner sep=0pt, minimum size=0.2cm},
    		dashednode/.style = {shape=circle, draw=blue, dashed, inner sep=0pt, minimum size=0.25cm},
    		textnode/.style = {shape=circle, draw, inner sep=0pt, minimum size=0.4cm},
    		dashedtextnode/.style = {shape=circle, draw=blue, dashed, inner sep=0pt, minimum size=0.4cm},
    		box/.style = {rectangle, fill=gray!20, rounded corners, fill opacity=1, inner sep=1pt}]
        \node at (1,-3.5) {$1$};
        \node at (11,-3.5) {$2$};
        \node at (11,3.5) {$3$};
        \node at (1,3.5) {$4$};
        
    \node[node] (v11) at (0,-1.0) {};
    \node[node] (v12) at (0.5,-2.0) {};
    \node[node,shape=star,star points=6] (v13) at (1.5,-2.5) {};
    \node[node,shape=diamond] (v14) at (2.5,-3.0) {};
    \node[node] (v15) at (3.5,-3.5) {};
    \node[node] (v16) at (5,-3.5) {};
    
    \path [-, black, thick] (v11) edge[bend left=0] node[above] {} (v12);
    \path [-, black, thick] (v12) edge[bend left=0] node[above] {} (v13);
    \path [-, black, thick] (v14) edge[bend left=0] node[above] {} (v15);
    \path [-, black, thick] (v15) edge[bend left=0] node[above] {} (v16);
    
    \node[node] (v21) at (7,-3.5) {};
    \node[node] (v22) at (8.5,-3.5) {};
    \node[node] (v23) at (9.5,-3.0) {};
    \node[node] (v24) at (10.5,-2.5) {};
    \node[node] (v25) at (11.5,-2.0) {};
    \node[node] (v26) at (12,-1.0) {};
    
    \path [-, black, thick] (v21) edge[bend left=0] node[above] {} (v22);
    \path [-, black, thick] (v22) edge[bend left=0] node[above] {} (v23);
    \path [-, black, thick] (v23) edge[bend left=0] node[above] {} (v24);
    \path [-, black, thick] (v24) edge[bend left=0] node[above] {} (v25);
    \path [-, black, thick] (v25) edge[bend left=0] node[above] {} (v26);
        
    \node[node] (v31) at (7,3.5) {};
    \node[node] (v32) at (8.5,3.5) {};
    \node[node] (v33) at (9.5,3.0) {};
    \node[node] (v34) at (10.5,2.5) {};
    \node[node] (v35) at (11.5,2.0) {};
    \node[node] (v36) at (12.0,1.0) {};
    
    \path [-, black, thick] (v31) edge[bend left=0] node[above] {} (v32);
    \path [-, black, thick] (v32) edge[bend left=0] node[above] {} (v33);
    \path [-, black, thick] (v33) edge[bend left=0] node[above] {} (v34);
    \path [-, black, thick] (v34) edge[bend left=0] node[above] {} (v35);
    \path [-, black, thick] (v35) edge[bend left=0] node[above] {} (v36);
        
    \node[node] (v41) at (0,1.0) {};
    \node[node] (v42) at (0.5,2.0) {};
    \node[node] (v43) at (1.5,2.5) {};
    \node[node] (v44) at (2.5,3.0) {};
    \node[node] (v45) at (3.5,3.5) {};
    \node[node] (v46) at (5,3.5) {};
    
    \path [-, black, thick] (v41) edge[bend left=0] node[above] {} (v42);
    \path [-, black, thick] (v42) edge[bend left=0] node[above] {} (v43);
    \path [-, black, thick] (v43) edge[bend left=0] node[above] {} (v44);
    \path [-, black, thick] (v44) edge[bend left=0] node[above] {} (v45);
    \path [-, black, thick] (v45) edge[bend left=0] node[above] {} (v46);
    
    \path [-, black, thick, dashed] (v16) edge[bend left=0] node[above] {} (v21);
    \path [-, black, thick, dashed] (v26) edge[bend left=0] node[above] {} (v36);
    \path [-, black, thick, dashed] (v31) edge[bend left=0] node[above] {} (v46);
    \path [-, black, thick, dashed] (v41) edge[bend left=0] node[above] {} (v11);
\end{tikzpicture}
}
    \subcaption{The base scenario. Vertex $s$ is star shaped and vertex $t$ diamond shaped. The dashed edges connect the vertex paths.}
    \label{fig:Sigma1:BaseScenario}
  \end{subfigure}
  \hfill
  \begin{subfigure}[b]{0.48\textwidth}
    \centering
    \scalebox{1.382}{
    \begin{tikzpicture}[scale=0.33,
    		node/.style = {shape=circle, draw, inner sep=0pt, minimum size=0.2cm},
    		dashednode/.style = {shape=circle, draw=blue, dashed, inner sep=0pt, minimum size=0.25cm},
    		textnode/.style = {shape=circle, draw, inner sep=0pt, minimum size=0.4cm},
    		dashedtextnode/.style = {shape=circle, draw=blue, dashed, inner sep=0pt, minimum size=0.4cm},
    		box/.style = {rectangle, fill=gray!20, rounded corners, fill opacity=1, inner sep=1pt}]
        \node at (1,-3.5) {$1$};
        \node at (11,-3.5) {$2$};
        \node at (11,3.5) {$3$};
        \node at (1,3.5) {$4$};
        
    \node[node] (v11) at (0,-1.0) {};
    \node[node] (v12) at (0.5,-2.0) {};
    \node[node,shape=star,star points=6] (v13) at (1.5,-2.5) {};
    \node[node,shape=diamond] (v14) at (2.5,-3.0) {};
    \node[node] (v15) at (3.5,-3.5) {};
    \node[node] (v16) at (5,-3.5) {};
    
    \path [-, black, thick] (v11) edge[bend left=0] node[above] {} (v12);
    \path [-, black, thick] (v12) edge[bend left=0] node[above] {} (v13);
    \path [-, black, thick] (v14) edge[bend left=0] node[above] {} (v15);
    \path [-, black, thick] (v15) edge[bend left=0] node[above] {} (v16);
    
    \node[node] (v21) at (7,-3.5) {};
    \node[node] (v22) at (8.5,-3.5) {};
    \node[node] (v23) at (9.5,-3.0) {};
    \node[node] (v24) at (10.5,-2.5) {};
    \node[node] (v25) at (11.5,-2.0) {};
    \node[node] (v26) at (12,-1.0) {};
    
    \path [-, black, thick] (v21) edge[bend left=0] node[above] {} (v22);
    \path [-, black, thick] (v22) edge[bend left=0] node[above] {} (v23);
    \path [-, black, thick] (v23) edge[bend left=0] node[above] {} (v24);
    \path [-, black, thick] (v24) edge[bend left=0] node[above] {} (v25);
    \path [-, black, thick] (v25) edge[bend left=0] node[above] {} (v26);
        
    \node[node] (v31) at (7,3.5) {};
    \node[node] (v32) at (8.5,3.5) {};
    \node[node] (v33) at (9.5,3.0) {};
    \node[node] (v34) at (10.5,2.5) {};
    \node[node] (v35) at (11.5,2.0) {};
    \node[node] (v36) at (12.0,1.0) {};
    
    \path [-, black, thick] (v31) edge[bend left=0] node[above] {} (v32);
    \path [-, black, thick] (v32) edge[bend left=0] node[above] {} (v33);
    \path [-, black, thick] (v33) edge[bend left=0] node[above] {} (v34);
    \path [-, black, thick] (v34) edge[bend left=0] node[above] {} (v35);
    \path [-, black, thick] (v35) edge[bend left=0] node[above] {} (v36);
        
    \node[node] (v41) at (0,1.0) {};
    \node[node] (v42) at (0.5,2.0) {};
    \node[node] (v43) at (1.5,2.5) {};
    \node[node] (v44) at (2.5,3.0) {};
    \node[node] (v45) at (3.5,3.5) {};
    \node[node] (v46) at (5,3.5) {};
    
    \path [-, black, thick] (v41) edge[bend left=0] node[above] {} (v42);
    \path [-, black, thick] (v42) edge[bend left=0] node[above] {} (v43);
    \path [-, black, thick] (v43) edge[bend left=0] node[above] {} (v44);
    \path [-, black, thick] (v44) edge[bend left=0] node[above] {} (v45);
    \path [-, black, thick] (v45) edge[bend left=0] node[above] {} (v46);
    
    \path [-, black, thick, dashed] (v11) edge[bend left=0] node[above] {} (v21);
    \path [-, black, thick, dashed] (v11) edge[bend left=0] node[above] {} (v26);
    \path [-, black, thick, dashed] (v11) edge[bend left=0] node[above] {} (v31);
    \path [-, black, thick, dashed] (v11) edge[bend left=0] node[above] {} (v36);
    \path [-, black, thick, dashed] (v16) edge[bend left=0] node[above] {} (v21);
    \path [-, black, thick, dashed] (v16) edge[bend left=0] node[above] {} (v26);
    \path [-, black, thick, dashed] (v16) edge[bend left=0] node[above] {} (v31);
    \path [-, black, thick, dashed] (v16) edge[bend left=0] node[above] {} (v36);
    \path [-, black, thick, dashed] (v21) edge[bend left=0] node[above] {} (v41);
    \path [-, black, thick, dashed] (v21) edge[bend left=0] node[above] {} (v46);
    \path [-, black, thick, dashed] (v26) edge[bend left=0] node[above] {} (v41);
    \path [-, black, thick, dashed] (v26) edge[bend left=0] node[above] {} (v46);
    \path [-, black, thick, dashed] (v31) edge[bend left=0] node[above] {} (v41);
    \path [-, black, thick, dashed] (v31) edge[bend left=0] node[above] {} (v46);
    \path [-, black, thick, dashed] (v36) edge[bend left=0] node[above] {} (v41);
    \path [-, black, thick, dashed] (v36) edge[bend left=0] node[above] {} (v46);
\end{tikzpicture}
}
    \subcaption{The uncertainty scenario. The dashed edges are the quadruplicated edges of $G$. The dashed edges of the base scenario are deleted.}
    \label{fig:Sigma1:RecoveryScenario}
  \end{subfigure}
  \caption{The base scenarios and the uncertainty scenario.}
\end{figure}

On one hand, the construction of the base scenario $\sigma_0$ forces the base solution $\texttt{s}_0$ to be the cycle itself. The solution $\texttt{s}_0$ is presented in \Cref{fig:Sigma1:ScenarioSolutions:a}.
On the other hand, the vertex paths force the solution of the uncertainty scenario to go over all vertex paths because of setting $\kappa = |V(G)|$ prevents the solution $\texttt{s}_1$ from evading these paths.
A selection of possible solutions for the uncertainty scenario are shown in \Cref{fig:Sigma1:ScenarioSolutions:b}.

\begin{figure}[!ht]
    \centering
    \begin{subfigure}[b]{0.48\textwidth}
        \centering
        \scalebox{1.382}{
    \begin{tikzpicture}[scale=0.33,
    		node/.style = {shape=circle, draw, inner sep=0pt, minimum size=0.2cm},
    		dashednode/.style = {shape=circle, draw=blue, dashed, inner sep=0pt, minimum size=0.25cm},
    		textnode/.style = {shape=circle, draw, inner sep=0pt, minimum size=0.4cm},
    		dashedtextnode/.style = {shape=circle, draw=blue, dashed, inner sep=0pt, minimum size=0.4cm},
    		box/.style = {rectangle, fill=gray!20, rounded corners, fill opacity=1, inner sep=1pt}]
        \node at (1,-3.5) {$1$};
        \node at (11,-3.5) {$2$};
        \node at (11,3.5) {$3$};
        \node at (1,3.5) {$4$};
        
    \node[node] (v11) at (0,-1.0) {};
    \node[node] (v12) at (0.5,-2.0) {};
    \node[node,shape=star,star points=6] (v13) at (1.5,-2.5) {};
    \node[node,shape=diamond] (v14) at (2.5,-3.0) {};
    \node[node] (v15) at (3.5,-3.5) {};
    \node[node] (v16) at (5,-3.5) {};
    
    \path [-, black, thick] (v11) edge[bend left=0] node[above] {} (v12);
    \path [-, black, thick] (v12) edge[bend left=0] node[above] {} (v13);
    \path [-, black, thick] (v14) edge[bend left=0] node[above] {} (v15);
    \path [-, black, thick] (v15) edge[bend left=0] node[above] {} (v16);
    
    \node[node] (v21) at (7,-3.5) {};
    \node[node] (v22) at (8.5,-3.5) {};
    \node[node] (v23) at (9.5,-3.0) {};
    \node[node] (v24) at (10.5,-2.5) {};
    \node[node] (v25) at (11.5,-2.0) {};
    \node[node] (v26) at (12,-1.0) {};
    
    \path [-, black, thick] (v21) edge[bend left=0] node[above] {} (v22);
    \path [-, black, thick] (v22) edge[bend left=0] node[above] {} (v23);
    \path [-, black, thick] (v23) edge[bend left=0] node[above] {} (v24);
    \path [-, black, thick] (v24) edge[bend left=0] node[above] {} (v25);
    \path [-, black, thick] (v25) edge[bend left=0] node[above] {} (v26);
        
    \node[node] (v31) at (7,3.5) {};
    \node[node] (v32) at (8.5,3.5) {};
    \node[node] (v33) at (9.5,3.0) {};
    \node[node] (v34) at (10.5,2.5) {};
    \node[node] (v35) at (11.5,2.0) {};
    \node[node] (v36) at (12.0,1.0) {};
    
    \path [-, black, thick] (v31) edge[bend left=0] node[above] {} (v32);
    \path [-, black, thick] (v32) edge[bend left=0] node[above] {} (v33);
    \path [-, black, thick] (v33) edge[bend left=0] node[above] {} (v34);
    \path [-, black, thick] (v34) edge[bend left=0] node[above] {} (v35);
    \path [-, black, thick] (v35) edge[bend left=0] node[above] {} (v36);
        
    \node[node] (v41) at (0,1.0) {};
    \node[node] (v42) at (0.5,2.0) {};
    \node[node] (v43) at (1.5,2.5) {};
    \node[node] (v44) at (2.5,3.0) {};
    \node[node] (v45) at (3.5,3.5) {};
    \node[node] (v46) at (5,3.5) {};
    
    \path [-, black, thick] (v41) edge[bend left=0] node[above] {} (v42);
    \path [-, black, thick] (v42) edge[bend left=0] node[above] {} (v43);
    \path [-, black, thick] (v43) edge[bend left=0] node[above] {} (v44);
    \path [-, black, thick] (v44) edge[bend left=0] node[above] {} (v45);
    \path [-, black, thick] (v45) edge[bend left=0] node[above] {} (v46);

    \path [-, black, thick, dashed] (v16) edge[bend left=0] node[above] {} (v21);
    \path [-, black, thick, dashed] (v26) edge[bend left=0] node[above] {} (v36);
    \path [-, black, thick, dashed] (v31) edge[bend left=0] node[above] {} (v46);
    \path [-, black, thick, dashed] (v41) edge[bend left=0] node[above] {} (v11);
\end{tikzpicture}
}
        \subcaption{The base scenario $\sigma$ with solution $\texttt{s}_0$ in black for example graph $G$.}
        \label{fig:Sigma1:ScenarioSolutions:a}
    \end{subfigure}
    \hfill
    \begin{subfigure}[b]{0.48\textwidth}
        \centering
        \scalebox{1.382}{
    \begin{tikzpicture}[scale=0.33,
    		node/.style = {shape=circle, draw, inner sep=0pt, minimum size=0.2cm},
    		dashednode/.style = {shape=circle, draw=blue, dashed, inner sep=0pt, minimum size=0.25cm},
    		textnode/.style = {shape=circle, draw, inner sep=0pt, minimum size=0.4cm},
    		dashedtextnode/.style = {shape=circle, draw=blue, dashed, inner sep=0pt, minimum size=0.4cm},
    		box/.style = {rectangle, fill=gray!20, rounded corners, fill opacity=1, inner sep=1pt}]
        \node at (1,-3.5) {$1$};
        \node at (11,-3.5) {$2$};
        \node at (11,3.5) {$3$};
        \node at (1,3.5) {$4$};
        
    \node[node] (v11) at (0,-1.0) {};
    \node[node] (v12) at (0.5,-2.0) {};
    \node[node,shape=star,star points=6] (v13) at (1.5,-2.5) {};
    \node[node,shape=diamond] (v14) at (2.5,-3.0) {};
    \node[node] (v15) at (3.5,-3.5) {};
    \node[node] (v16) at (5,-3.5) {};
    
    \path [-, black, thick] (v11) edge[bend left=0] node[above] {} (v12);
    \path [-, black, thick] (v12) edge[bend left=0] node[above] {} (v13);
    \path [-, black, thick] (v14) edge[bend left=0] node[above] {} (v15);
    \path [-, black, thick] (v15) edge[bend left=0] node[above] {} (v16);
    
    \node[node] (v21) at (7,-3.5) {};
    \node[node] (v22) at (8.5,-3.5) {};
    \node[node] (v23) at (9.5,-3.0) {};
    \node[node] (v24) at (10.5,-2.5) {};
    \node[node] (v25) at (11.5,-2.0) {};
    \node[node] (v26) at (12,-1.0) {};
    
    \path [-, black, thick] (v21) edge[bend left=0] node[above] {} (v22);
    \path [-, black, thick] (v22) edge[bend left=0] node[above] {} (v23);
    \path [-, black, thick] (v23) edge[bend left=0] node[above] {} (v24);
    \path [-, black, thick] (v24) edge[bend left=0] node[above] {} (v25);
    \path [-, black, thick] (v25) edge[bend left=0] node[above] {} (v26);
        
    \node[node] (v31) at (7,3.5) {};
    \node[node] (v32) at (8.5,3.5) {};
    \node[node] (v33) at (9.5,3.0) {};
    \node[node] (v34) at (10.5,2.5) {};
    \node[node] (v35) at (11.5,2.0) {};
    \node[node] (v36) at (12.0,1.0) {};
    
    \path [-, black, thick] (v31) edge[bend left=0] node[above] {} (v32);
    \path [-, black, thick] (v32) edge[bend left=0] node[above] {} (v33);
    \path [-, black, thick] (v33) edge[bend left=0] node[above] {} (v34);
    \path [-, black, thick] (v34) edge[bend left=0] node[above] {} (v35);
    \path [-, black, thick] (v35) edge[bend left=0] node[above] {} (v36);
        
    \node[node] (v41) at (0,1.0) {};
    \node[node] (v42) at (0.5,2.0) {};
    \node[node] (v43) at (1.5,2.5) {};
    \node[node] (v44) at (2.5,3.0) {};
    \node[node] (v45) at (3.5,3.5) {};
    \node[node] (v46) at (5,3.5) {};
    
    \path [-, black, thick] (v41) edge[bend left=0] node[above] {} (v42);
    \path [-, black, thick] (v42) edge[bend left=0] node[above] {} (v43);
    \path [-, black, thick] (v43) edge[bend left=0] node[above] {} (v44);
    \path [-, black, thick] (v44) edge[bend left=0] node[above] {} (v45);
    \path [-, black, thick] (v45) edge[bend left=0] node[above] {} (v46);
    
    \path [-, red, thick, dashed] (v11) edge[bend left=0] node[above] {} (v21);
    \path [-, green, thick, dashed] (v11) edge[bend left=0] node[above] {} (v26);
    \path [-, blue, thick, dashed] (v11) edge[bend left=0] node[above] {} (v31);
    \path [-, brown, thick, dashed] (v11) edge[bend left=0] node[above] {} (v36);
    \path [-, brown, thick, dashed] (v16) edge[bend left=0] node[above] {} (v21);
    \path [-, blue, thick, dashed] (v16) edge[bend left=0] node[above] {} (v26);
    \path [-, green, thick, dashed] (v16) edge[bend left=0] node[above] {} (v31);
    \path [-, red, thick, dashed] (v16) edge[bend left=0] node[above] {} (v36);
    \path [-, blue, thick, dashed] (v21) edge[bend left=0] node[above] {} (v41);
    \path [-, green, thick, dashed] (v21) edge[bend left=0] node[above] {} (v46);
    \path [-, brown, thick, dashed] (v26) edge[bend left=0] node[above] {} (v41);
    \path [-, red, thick, dashed] (v26) edge[bend left=0] node[above] {} (v46);
    \path [-, red, thick, dashed] (v31) edge[bend left=0] node[above] {} (v41);
    \path [-, brown, thick, dashed] (v31) edge[bend left=0] node[above] {} (v46);
    \path [-, green, thick, dashed] (v36) edge[bend left=0] node[above] {} (v41);
    \path [-, blue, thick, dashed] (v36) edge[bend left=0] node[above] {} (v46);
\end{tikzpicture}
}
        \subcaption{The uncertainty scenarios with four solutions (red, green, blue and brown) of the possible eight solutions for example graph $G$.}
        \label{fig:Sigma1:ScenarioSolutions:b}
    \end{subfigure}
    \caption{Both scenarios with their respective solutions.}
    \label{fig:Sigma1:ScenarioSolutions}
\end{figure}

The reduction is clearly computable in logarithmic space, because we only have to count the number of vertices in the duplication procedure.
The connection to the cycle is also directly possible if the number of vertices known.
At last, the introduction of the edges for the base scenario is only a copy procedure based on the original graph, which is directly computable if the number of vertices is known.

Furthermore, the reduction is correct.
First of all, the only solution for the base scenario is the path from $s$ to $t$ over the former cycle in $\sigma_0$.
If a Hamiltonian cycle exists in the graph, then it is possible find a correspondent solution $\texttt{s}_1$ for the uncertainty scenario.
We can use the edges from the Hamiltonian cycle in $G$ and use the edges $\{v_i^a, v_{(j \ \text{mod} \ |V(G)|)}^b\}$ of both of the corresponding edges in the uncertainty scenario.
Thus, the vertex paths are connected to a Hamiltonian cycle as well.

One the other hand, if there is no Hamiltonian cycle, then there is no path of the form $(s, v_1^{a}, v_2^{x}, v_2^{y}, \ldots, v_{|V(G)|}^{x}, v_{|V(G)|}^{y}, v_1^{b}, t)$, where $x, y \in \{a, b\}$ and $x \neq y$.
This is due to the fact that the base scenario $\sigma_0$ in combination with the too small $\kappa = |V(G)|$ enables the possibility to switch only away from the edges that connect the vertex paths.
It is not possible to switch away completely from a vertex path as there are $|V(G)|+1$ edges in each vertex path (including that with $s$ and $t$).
Thus, at least one edge that has to be in the $s$-$t$-path would not be correctly included into the $s$-$t$-path or $s$ and $t$ are not connected by a path.
\end{proof}
\section{Recoverable Robust Problems and the Polynomial Hierarchy}\label{sec:sigma3}
In this section, we investigate the connection between multi-stage Hamming distance recoverable robust problems and the polynomial hierarchy.
For this, we introduce two succinct encodings: $xor$-dependencies and $\Gamma$-set scenarios.
We first prove that the Hamming distance recoverable robust version of problems, which are in \NP, are in $\Sigma^p_3$ for both encodings.
Then, we prove the hardness of the Hamming distance recoverable robust \sat{} for both encodings.

\begin{definition}[Hamming Distance Recoverable Robust \sat{}]
	The problem \hdrr[\sat{}] with Hamming distance over the literals $L$ is defined as follows.
	\begin{description}
		\item[Input:] Literals $L$, clauses $C$, base scenario $\sigma_0 \subseteq L$, uncertainty scenarios $\textsf{S} \subseteq 2^L$, $\kappa \in \mathbb N$
		\item[Question:] Are there solutions $\texttt{s}_{0} \subseteq \sigma_0$ and $\texttt{s}_\sigma \subseteq \sigma$ for all $\sigma \in \textsf{S}$ such that $H(\texttt{s}_0, \texttt{s}_\sigma) \leq \kappa$ for all $\sigma \in \textsf{S}$ and setting $\texttt{s}_0$ and $\texttt{s}_\sigma$ to true, all corresponding formulae of clauses $C|_{\sigma_0}$ and $C|_{\sigma}$ are satisfied?
	\end{description}
\end{definition}

At last, we extend these results to the multistage recoverable robustness case by showing the $\Sigma^p_{2m+1}$-completeness of the Hamming distance recoverable robust \sat{} with $m$ uncertainty and recovery stages.
We begin with $xor$-dependency scenarios.

\begin{definition}[$xor$-Dependency Scenarios]
    Let $\sigma_0$ be the base scenario.
    The encoding of $xor$-dependencies is a tuple $(E', \{(E_{1,1}, E_{1,2}),$ $\ldots,$ $(E_{n,1}, E_{n,2})\})$, where $E'$ and all $E_{i,j}$ are pairwise disjoint sets of combinatorial elements for all $i \in \{1,\ldots,n\}, j \in \{1,2\}$.
    Then the scenario set $\textsf{S}$ includes all $\sigma$ of the form $\sigma = \sigma_0 \vartriangle (E' \cup E_1 \cup \ldots \cup E_n)$ with either $(E_i = E_{i,1})$ or $(E_i = E_{i,2})$ for all $i \in \{1, \ldots, n\}$.
\end{definition}

Observe that with a linear sized encoding, exponentially many scenarios may be encoded.
We study this combinatorial explosion with the result that it introduces more complexity for Hamming distance recoverable robust problems in comparison to the base problem.
Concretely, we use \sat{} as base problem and show the $\Sigma^p_3$-hardness of \hdrr[\sat{}] with a linear number of $xor$-dependencies.
From that point on, we can derive hardness results for further problems.
Before we start the analysis of the hardness, we shall show that if $P_A \in \NP$, then \hdrr{} with a linear number of $xor$-dependencies is in $\Sigma^p_3$.

\begin{theorem}\label{thm12}
	If $P_A \in \NP$, then \hdrr{} with $xor$-dependencies is in $\Sigma^p_3$.
\end{theorem}
\begin{proof}
We present a polynomial time verifier that receives an ($\exists$-quantified) string $y_1$, a ($\forall$-quantified) string $y_2$, and an ($\exists$-quantified) string $y_3$ as input together with the instance.
The first string $y_1$ encodes the solution $\texttt{s}_0$ to the base scenario.
The second string $y_2$ encodes the scenario $\sigma$ for all $\sigma \in \textsf{S}$.
The third string encodes the solution $\texttt{s}_\sigma$ for the selected scenario $\sigma$.

The solution to the scenarios $\texttt{s}_0$ and $(\texttt{s}_\sigma)_{\sigma \in \textsf{S}}$ are encoded as a subset of of active elements in the corresponding scenario.
The scenarios $\sigma_0$ and $\sigma \in \textsf{S}$ can be computed in polynomial time from the input encoding encoded as sets, because the number of $xor$-dependencies is limited by the input length.
Furthermore, the solutions $\texttt{s}_0$ and $(\texttt{s}_\sigma)_{\sigma \in \textsf{S}}$ are subsets of $\sigma_0$ and $\sigma \in \textsf{S}$ correspondingly.
Consequently, the length of the input to the verifying algorithm is at most polynomial in the input length.

We can now construct the following algorithm that runs in polynomial time to verify the correctness of the strings.
First we compute the explicit encodings of the base scenario and the scenario $\sigma \in \textsf{S}$ encoded in $y_2$ in polynomial time.
We then verify whether the solution $\texttt{s}_0$ encoded by $y_1$ is a solution to $\sigma_0$ and whether $\texttt{s}_\sigma$ encoded by $y_3$ is a solution to scenario $\sigma$.
This is doable by using the existing verifier for the base problem that exists because the problem is in \NP.
At last, we check $H(\texttt{s}_0, \texttt{s}_\sigma) \leq \kappa$.
\end{proof}
	
\begin{theorem}\label{sat-hdrr-sigma-3-hard}\label{thm13}
    \hdrr[\sat{}] with $xor$-dependency scenarios is $\Sigma^p_3$-hard.
\end{theorem}
\begin{proof}
We reduce $\exists\forall\exists\sat{}$ to \hdrr[\sat{}].
For this, let $(X, Y, Z, C)$ be the $\exists\forall\exists\sat{}$ instance, where $\exists X \forall Y \exists Z~C(X, Y, Z)$ is the formula with clauses $C(X, Y, Z)$. We denote the \hdrr[\sat{}] instance as $I$.

\begin{description}
\item[Variables] We modify the variable set as follows.
The variable set $X$ remains the same.
We substitute $Y$ by $\{y^t_i, \ y^f_i \mid y_i \in Y\} =: Y'$.
At last, we define $Z' := Z \cup \{y^t_{i,0}, y^t_{i,1}, y^f_{i,0}, y^f_{i,1} \mid y_i \in Y\}$.

\item[Clauses]
The clauses are then modified as follows.
For all $y_i \in Y$, we add $y^f_i \leftrightarrow 0$ and $y^t_i \leftrightarrow 1$ to the formula.
Furthermore for all $y_i \in Y$, we add $y^t_i \leftrightarrow y^t_{i,1}, \ y^t_i \leftrightarrow \overline y^t_{i,0}, \ y^f_i \leftrightarrow y^f_{i,0}, \ y^f_i \leftrightarrow \overline y^f_{i,1}$ to the formula.
At last, we do the following substitutions: For every clause $c = (a, b, y_i) \in C$ with $a, b \in X \cup Y \cup Z$, we substitute $c$ by the clauses $(a, b, y^t_{i,1})$ and $(a,b,y^f_{i,0})$ and for clauses $c = (a, b, \overline y_i) \in C$ with $a, b \in X \cup Y \cup Z$ we substitute $c$ by the clauses $(a, b, y^t_{i,0})$ and $(a,b,y^f_{i,1})$.
We denote the set of modified clauses from $C$ by $C'$.
This is possible in polynomial time because we have a \sat{} instance and we are introducing at most eight new clauses per existing clause.

\item[Scenarios]
In the base scenario of $I$ only the variables from $X$ are active.
The uncertainty scenarios are encoded with $xor$-dependencies.
For this, we introduce $xor$-dependencies on the variables and clauses from $y^t_i$ and $y^f_i$ for all $i \in \{1, \ldots, |Y|\}$.
Concretely, we define the set $E' = Z' \cup C'$ and for each $i \in \{1,\ldots,n\}$, we define $E_{i,1} = \{y^t_i, (y^t_i \leftrightarrow y^t_{i,1}), (y^t_i \leftrightarrow \overline y^t_{i,0}), (y^t_i \leftrightarrow 1)\}$ as well as $E_{i,2} = \{y^f_i, (y^f_i \leftrightarrow y^f_{i,0}), (y^f_i \leftrightarrow \overline y^f_{i,1}), (y^f_i \leftrightarrow 0)\}$.
At last we set the maximum Hamming distance between the literals to $\kappa = |Y| + |Z'|$.

\item[Polynomial Time] This transformation is computable in polynomial time because for each literal and each clause in $(X, Y, Z, C)$ a fixed amount of literals and clauses in $I$ are created. Furthermore, the formula can be transformed into CNF by substituting $a \leftrightarrow b$ with clauses $(\overline a \lor b)$ and $(a \lor \overline b)$.

\item[Correctness]
For the correctness, we have to prove that the constructed instance over the variable sets $X$, $Y'$, and $Z'$ together with the $xor$-dependency scenarios are logically equivalent to the $\exists\forall\exists\sat{}$ formula.
First, we focus on the $\exists X$ part.
Any assignment to the variables from $X$ is a valid solution to the base scenario.
Because $\kappa = |Y| + |Z'|$ and $|Y| + |Z'|$ new variables appear in all of the uncertainty scenarios, the decision on the variables from $X$ is made while choosing a solution to the base scenario and cannot be changed in any uncertainty scenario.
Thus the decision on the variables from $X$ are the same in both the base scenario and the chosen uncertainty scenario.

Next, we concentrate on the $\forall Y$ part.
First for all $i \in \{1, \ldots, |Y|\}$, the clauses $1 \leftrightarrow y^t_i$ and $0 \leftrightarrow y^f_i$, force the variable $y^t_i$ to be always true and the variable $y^f_i$ to be always false if they are active.
The $xor$-dependencies activate exactly one of $y^t_i$ and $y^f_i$ for all $i \in \{1, \ldots, |Y|\}$.
Furthermore, if $y^t_i$ is active, then $y^t_{i,0}$ evaluates to $0$ and $y^t_{i,1}$ evaluates to $1$, and if $y^f_i$ is active, then $y^f_{i,0}$ evaluates to $0$ and $y^f_{i,1}$ evaluates to $1$.
Thus the clauses containing $y^t_{i,0}$ and $y^t_{i,1}$ (resp. $y^f_{i,0}$ and $y^f_{i,1}$) have the same satisfaction behavior than the clauses that contain $y_i$ (resp. $\overline y_i$) in the $\exists\forall\exists\sat{}$ formula.
If on the other hand, $y^t_i \in E_{i,1}$ is inactive, then
also the clauses $y^t_i \leftrightarrow y^t_{i,1}$ and $y^t_i \leftrightarrow \overline y^t_{i,0}$ are deactivated
such that both $y^t_{i,0}$ and $y^t_{i,1}$ can be set to $1$.
This allows all clauses containing $y^t_{i,0}$ or $y^t_{i,1}$ to be trivially fulfilled, whenever $y^t_i$ is inactive. 
The same argument holds for $y^f_i \in E_{i,2}$, i.e. the clauses $y^f_i \leftrightarrow y^f_{i,0}$ and $y^f_i \leftrightarrow \overline y^f_{i,1}$ are deleted and both $y^t_{i,0}$ and $y^t_{i,1}$ can be set to $1$.
Because the combinations allowed by the $xor$-dependencies are all $2^{|Y|}$ possible truth assignments to variables $Y$, the $xor$-dependency scenarios are equivalent to a $\forall Y$ for the variables $Y$.
Thus, we also have a one-to-one correspondence between the variables in $Y$ and $Y'$ in both instances.

At last, we have to consider the $\exists Z'$ part.
All variables from $X$ and $Y$ in the instance of $\hdrr[\sat{}]$ are already set equivalently to the assignment to the variables from $X$ and $Y$ in the $\exists\forall\exists\sat{}$ formula.
The variables from the set $\{y^t_{i,0}, y^t_{i,1}, y^f_{i,0}, y^f_{i,1} \mid y_i \in Y\}$ are assigned according to $Y$.
All variables of the $\hdrr[\sat{}]$ instance that are not yet assigned are free variables from $Z$.
The clauses $C'$, however, are equivalent to the clauses from the  $\exists\forall\exists\sat{}$ formula.
Thus the rest of the variables (in both instances these are the variables from $Z$) is one-to-one correspondent.

In conclusion, the instance from $\exists\forall\exists\sat{}$ is equisatisfiable to the constructed instance of $\hdrr[\sat{}]$ because the assignments on the set $X$, $Y$, and $Z$ correspondent to the assignments in $X$, $Y'$, and $Z'$.
\end{description}
\end{proof}

While the other parts of the paper are developed independent from Goerigk et al. \cite{DBLP:journals/corr/abs-2209-01011}, the results for $\Gamma$-set scenarios are built upon it.
The results based on $xor$-dependencies are adaptable to the $\Gamma$-set scenarios as described in this section.
For the $\Gamma$-set scenarios, we use the definition over sets instead of elements as in $\Gamma$-scenarios, which is defined as follows.

\begin{definition}[$\Gamma$-set Scenarios]\label{def:gamma-set-scenarios}
    Let $\sigma_0$ be the base scenario.
    The encoding of $\Gamma$-set scenarios is a tuple $(E', \{E_{1}, E_{2}, \ldots E_{n}\})$, where $E'$ and all $E_{i}$ are pairwise disjoint sets of combinatorial elements for all $i \in \{1,\ldots,n\}$.
    Then, the corresponding scenario set $\textsf{S}$ includes all $\sigma$ of the form $\sigma = \sigma_0 \vartriangle (E' \cup \bigcup_{E \in \mathcal{E}} E)$ with $\mathcal{E} \subseteq \{E_{1}, E_{2}, \ldots, E_{n}\}, |\mathcal{E}| \leq \Gamma$.
\end{definition}

Again, with a linear sized encoding, exponentially many scenarios may be encoded.
We show $\Sigma^p_3$-hardness of \hdrr[\sat{}] with $\Gamma$-set scenarios.
A proof on the so-called \textsc{Robust Adjustable Sat} was already conducted by Goerigk et al. \cite{DBLP:journals/corr/abs-2209-01011}.
This version of \sat{} uses uncertainties over the costs instead of the elements as in $\Sigma^p_3$-hardness of \hdrr[\sat{}] with $\Gamma$-set scenarios.
Thus, the proof is not analogous as it is different in technicalities, nevertheless, we reuse their basic idea of introducing the cheat detection gadget (modeled by the $s$-variables) for our proof.
Furthermore, we show also that if $P_A \in \NP$, then \hdrr{} with $\Gamma$-set scenarios is in $\Sigma^p_3$.

\begin{theorem}\label{thm121}
If $P_A \in \NP$, then \hdrr{} with $\Gamma$-set scenarios is in $\Sigma^p_3$.
\end{theorem}
\begin{proof}
Each scenario from the $\Gamma$-set scenarios is encodable in polynomial space because the number of sets in $\mathcal{E}$ from \Cref{def:gamma-set-scenarios} is limited by the input length.
Thus, this proof is analogous to the proof for $xor$-dependencies.
\end{proof}
	
\begin{theorem}\label{sat-hdrr-sigma-3-hard2}\label{thm131}
    \hdrr[\sat{}] with $\Gamma$-set scenarios is $\Sigma^p_3$-hard.
\end{theorem}
\begin{proof}
We heavily reuse the transformation for $xor$-dependencies.
Nevertheless, we have to introduce a mechanism to accommodate the less structured $\Gamma$-set scenarios in comparison to $xor$-dependencies.
At last, the scenarios have to be adapted to $\Gamma$-set scenarios.

We reduce $\exists\forall\exists\sat{}$ to \hdrr[\sat{}].
For this, let $(X, Y, Z, C)$ be the $\exists\forall\exists\sat{}$-instance, where $\exists X \forall Y \exists Z~C(X, Y, Z)$ is the formula with clauses $C(X, Y, Z)$. We denote the \hdrr[\sat{}] instance as $I$.

\begin{description}
\item[Variables] We modify the variable set as follows.
The variable set $X$ remains the same.
We substitute the set $Y$ by $\{y^t_i, \ y^f_i \mid y_i \in Y\} =: Y'$.
Moreover, we define set $Z' := Z \cup \{y^t_{i,0}, y^t_{i,1}, y^f_{i,0}, y^f_{i,1} \mid y_i \in Y\} \cup \{s, s_i \mid y_i \in Y\}$.
The added variables $s_i$ for each $y_i \in Y$ and the additional variable $s$ fulfill the same function as in the proof of Goerigk, Lendl and Wulf \cite{DBLP:journals/corr/abs-2209-01011}.

\item[Clauses]
The clauses are then modified as follows.
For all $y_i \in Y$, we add $y^f_i \leftrightarrow 0$ and $y^t_i \leftrightarrow 1$ to the formula.
Furthermore for all $y_i \in Y$, we add $y^t_i \leftrightarrow y^t_{i,1}, \ y^t_i \leftrightarrow \overline y^t_{i,0}, \ y^f_i \leftrightarrow y^f_{i,0}, \ y^f_i \leftrightarrow \overline y^f_{i,1}$ to the formula.
Then, we do the following substitutions: For every clauses $c = (a, b, y_i) \in C$ with $a, b \in X \cup Y \cup Z$ we substitute $c$ by the clauses $(a, b, y^t_{i,1})$ and $(a,b,y^f_{i,0})$ and for clauses $c = (a, b, \overline y_i) \in C$ with $a, b \in X \cup Y \cup Z$ we substitute $c$ by the clauses $(a, b, y^t_{i,0})$ and $(a,b,y^f_{i,1})$.
This is possible in polynomial time because we have a \sat{} instance and we are introducing at most eight clauses per clause.
Moreover, we add $\overline s$ to all clauses $c \in C$, such that we obtain a formula equivalent to $s \rightarrow C(X, Y, Z)$.
We denote the set of modified clauses from $C$ by $C'$.
At last, we add $\overline y^t_i \lor s_i$ and $y^f_i \lor s_i$ as well as $s \lor \overline s_1 \lor \overline s_2 \lor \ldots \lor \overline s_{|Y|}$ to the clauses.

\item[Scenarios]
The first scenario of $I$ consists only of the variables from $X$.
Based on this, we encode the uncertainty scenarios with $\Gamma$-set scenarios.
For this, we include the variable $y^t_i$ (respectively $y^f_i$) together with its clauses in one of the $E_i$.
Concretely, we define $E' = Z' \cup C' \cup \{(s \lor \overline s_1 \lor \ldots \lor \overline s_{|Y|})\}$.
Furthermore, we define $E_{2i-1} = \{y^t_i, (y^t_i \leftrightarrow y^t_{i,1}), (y^t_i \leftrightarrow \overline y^t_{i,0}), (y^t_i \leftrightarrow 1), (\overline y^t_i \lor s_i)\}$ and $E_{2i} = \{y^f_i, (y^f_i \leftrightarrow y^f_{i,0}), (y^f_i \leftrightarrow \overline y^f_{i,1}), (y^f_i \leftrightarrow 0), (y^f_i \lor s_i)\}$ for $i \in \{1,\ldots,|Y|\}$.
At last, set $\kappa = |Y| + |Z'|$ and $\Gamma = |Y|$.

\item[Polynomial Time] This transformation is computable in polynomial time because for each literal and each clause in $(X, Y, Z, C)$ a fixed amount of literals and clauses in $I$ are created. Furthermore, the formula can be transformed into 3CNF by substituting $a \leftrightarrow b$ with clauses $(\overline a \lor b)$ and $(a \lor \overline b)$ and using Karp's reduction from \textsc{Sat} to \textsc{3Sat} \cite{DBLP:conf/coco/Karp72}.

\item[Correctness]
For the correctness, we have to prove that the $\Gamma$-set scenarios within the construction are logically equivalent to $xor$-dependencies.
Indeed the introduction of the cheat detection gadget, i.e. the $s$-variables, ensures this.
For this, observe that whenever the set of uncertain elements $\mathcal E$ is smaller than $\Gamma = |Y|$,
there is a pair of variables $y^t_i, y^f_i$ that is not active.
Consequently, the clauses $\overline y^t_i \lor s_i$ and $y^f_i \lor s_i$ are inactive and $s_i$ can be assigned to $0$.
It follows that $\overline s_i$ satisfies the clause $s \lor \overline s_1 \lor \overline s_2 \lor \ldots \lor \overline s_{|Y|}$ such that $s$ can be assigned $0$.
Then all clauses are fulfilled by the addition of $\overline s$ to all clauses from $C$.
This also holds, whenever there is an active pair of $y^f_i$ and $y^t_i$ because by the pigeonhole principle there is a $j \in \{1, \ldots, |Y|\}$ such that neither $y^f_j$ nor $y^t_j$ is active such that $s_j$ can be assigned to $0$.
Therefore, all non-trivial cases require exactly one of $y^f_i$ and $y^t_i$ to be active, which is equivalent to $xor$-dependencies.
\end{description}
\end{proof}

\subsection{Multi-Stage Recoverable Robustness}
In \emph{multi-stage recoverable robustness}, the uncertainty is not only modeled by one set of scenarios but multiple sets that are connected inductively.

\begin{definition}[Multi-Stage Recoverable Robust Problem]
    A multi-stage recoverable robust problem with $m$ recoveries $\mhdrr[P_A]$ is inductively defined as
    \begin{align*}
        \mhdrr{} &:= P_A& \ \text{for} \ m = 0,\\
        \mhdrr{} &:= \hdrr[(\mminushdrr{})]& \ \text{for} \ m >1.
    \end{align*}
\end{definition}

The complexity results naturally extend to the multiple recoverable robustness concept.
We make use of the inductive nature of the definition by proving the following theorems by induction.
For this, we reuse \Cref{thm12,thm13,thm121,thm131} as induction base.

\begin{theorem}
    \mhdrr[\sat] with $xor$-dependency scenarios is in $\Sigma^p_{2m+1}$.\\
    \mhdrr[\sat] with $\Gamma$-set scenarios is in $\Sigma^p_{2m+1}$.
\end{theorem}
\begin{proof}
We reuse the argument from \Cref{thm12} and generalize it to multiple stages.
For this, we present a polynomial time verifier that receives the instance together with the following strings as input:
an ($\exists$-quantified) string $y$,
and for each stage $i \in \{1, \ldots, m\}$
a ($\forall$-quantified) string $y^i_1$
and an ($\exists$-quantified) string $y^i_2$.
The first string $y$ encodes the solution $\texttt{s}_0$ to the base scenario.
The string $y^i_1$ encodes the scenario $\sigma_i \in \textsf{S}_i$.
The string $y^i_2$ encodes the solution $\texttt{s}_{\sigma_i}$ for the selected scenario $\sigma_i$.
The solution to the scenarios $\texttt{s}_0$ and all $\texttt{s}_{\sigma_i}$ for $i \in \{1, \ldots, m\}$ are encoded as a subset of of active elements in the corresponding scenario.

The scenarios $\sigma_0$ and $\sigma_i \in \textsf{S}_i$ for each $i \in \{1, \ldots, m\}$ can be computed in polynomial time from the input encoding encoded as sets, because the number stages and the number of $xor$-dependencies is limited by the input length.
Furthermore, the solutions $\texttt{s}_0$ and $\texttt{s}_{\sigma_i}$ for each $i \in \{1, \ldots, m\}$ are subsets of $\sigma_0$ and $\sigma_i \in \textsf{S}_i$ correspondingly.
Consequently, the length of the input to the verifier is at most polynomial in the input length.

We can now construct the following algorithm that runs in polynomial time to verify the correctness of the strings.
First we compute the explicit encodings of the base scenario and scenarios $\sigma_i \in \textsf{S}_i$ encoded in $y^i_1$ for each $i \in \{1, \ldots, m\}$ in polynomial time.
We then verify whether the solution $\texttt{s}_0$ encoded by $y$ is a solution to $\sigma_0$ and whether $\texttt{s}_{\sigma_i}$ encoded by $y^i_2$ is a solution to scenario $\sigma_i$.
This is doable by using the existing verifier for the base problem that exists because the problem is in \NP.
At last, we check whether $H(\texttt{s}_0, \texttt{s}_{\sigma_1}) \leq \kappa$ and whether $H(\texttt{s}_{\sigma_{i-1}}, \texttt{s}_{\sigma_i}) \leq \kappa$ for each $i \in \{2, \ldots, m\}$.
\end{proof}

\begin{theorem}\label{inductionmhdrr}
    \mhdrr[\sat] with $xor$-dependency scenarios is $\Sigma^p_{2m+1}$-hard.
    \mhdrr[\sat] with $\Gamma$-set scenarios is $\Sigma^p_{2m+1}$-hard.
\end{theorem}
\begin{proof}
    The proof is by induction over $m$.
    For the induction base, we consider for $m=0$ the \NP-complete problem \sat{} and for $m=1$ the $\Sigma^p_{3}$-complete problem \hdrr[\sat] (\Cref{sat-hdrr-sigma-3-hard}).

    For the induction step from $m$ to $m+1$, we extend the argument from \Cref{sat-hdrr-sigma-3-hard}.
    By induction hypothesis, we know that \mhdrr{} is $\Sigma^p_{2m+1}$-hard.
    More precisely, the induction hypothesis yields that $(\exists\forall)^{m}\exists$-\sat{} is reducible to \mhdrr{}.
    Thus, we need to model the $m+1$st alternation with the additional $m+1$st uncertainty stage.
    For this, let
    $$
        X_1, Y_1, X_2, Y_2, \ldots, X_{m+1}
    $$
    be the variable sets of the $(\exists\forall)^{m+1}\exists$-$\sat$ instance, where
    $$
        \exists X_1 \forall Y_1 \exists X_2 \forall Y_2 \ldots \exists X_{m+1} \ C(X_1, Y_1, X_2, Y_2, \ldots, X_{m+1})
    $$
    is the formula.
    By interpreting the variable sets $X_2, Y_2, \ldots, X_{m+1}$ as the variable set $Z$, which is not altered in any way, $Y_1$ as variable set $Y$ and $X_1$ as $X$, the additional alternation of the $(\exists\forall)^{m+1}\exists\sat$ formula can be modeled by one more uncertainty stage.
\end{proof}
	
\section{Classes of Recoverable Robust Problems}\label{sec:class}
We have shown that \hdrr[\sat{}] is the canonical $\Sigma^p_3$-complete Hamming distance recoverable robust problem.
The goal is to ``upgrade'' the existing reductions on the \NP-level to reduce the corresponding Hamming distance recoverable robust problems to each other.
If we are additionally able to guarantee transitivity, we are also able to easily achieve complexity results for a large class of problems.
Essentially, the reduction between Hamming distance recoverable robust problems needs to preserve the structure of the scenarios.
For this, consider problems $P_A$ and $P_B$.
We need to achieve that a combinatorial element $e_A$ in $P_A$ is active if and only if the combinatorial elements $E_B$, to which $e_A$ is mapped in $P_B$, are active.
Then, we can use this one-to-many correspondence to (de)activate the corresponding elements in the instance of $P_B$.

Many of the properties from above are already constituted by the informal concept of gadget reductions.
Gadget reductions describe that each part of the problem $P_A$ is mapped to a specified part of the problem $P_B$ that inherits the behavior in problem $P_A$.
We adjust this concept to combinatorial elements, that is universe elements and relation elements, for our purpose.
The goal is that a gadget is a subset of combinatorial elements in $P_B$ for every combinatorial element in $P_A$.
Furthermore, we preserve the (in)activeness of elements in a scenario.
We call reductions that fulfill this property \textit{modular} in the sense that all gadgets are easily (de)activatable.
Furthermore, the solution size, which is the number of universe elements in a solution, has to adapt accordingly while being easy to compute in order to define the Hamming distance in the reduction correctly.
We approach this later by demanding that the solution size of every gadget has to be a constant, i.e. it does not change when (de)activating other gadgets.

\subsection{Universe Gadget Reduction}
Let $P_A$ be a combinatorial decision problem with instance tuples $(U_A, R_A, F_A)$ and $P_B$ a combinatorial decision problem with instance tuples $(U_B, R_B, F_B)$. A Universe Gadget Reduction $f^{}_{\preceq}$ that many-one-reduces $P_A$ to $P_B$ is composed of a (possibly empty) constant gadget $Y_{const}$, which is the same for every instance, and of the independent mappings:
	$
		f_{R^i_{A}, R^j_{B}}: R^i_{A} \rightarrow 2^{R^j_{B}} \ \text{for all} \ (i,j) \in I_A \times I_B.
	$
We, then, call the substructure
    $$
        Y_x = f^{}_{\preceq}(x) = \bigcup_{(i, j) \ \in \ I_A \times I_B} f_{R_A^i, R_B^j}(x)
    $$
    the gadget for the specific universe element or relation element $x \in \bigcup_i R_A^i$.
    Additionally, we denote the set of all gadgets by $\Upsilon(R_A) = \{Y_r \mid r \in R^i_A \text{ with } i \in I_A\} \cup \{Y_{const}\}$ for the instance $R_A$.
    The mappings must fulfill the following properties.
    \begin{enumerate}
        \item Pre-image uniqueness:
            Let $y \in R^j_{B}$ for some $j \in I_B$, then either $y \in Y_{const}$ or there is exactly one $(i, j) \in I_A \times I_B$ and exactly one $x \in R_A^i$ such that $y \in f_{R_A^i, R_B^j}(x)$.
        \item Modularity:
            If a combinatorial element $r \in R^i_A$ from $(U_A, R_A, F_A)$ is removed to form a new instance $(U'_A,R'_A,F'_A)$, the removal of the gadget of $r$ in $(U_B,R_B,F_B)$ induces a correct reduction instance $(U'_B,R'_B,F'_B)$.
        	A removal of $r \in R^i_A$ corresponds to the substitution by a (possibly empty) removal gadget $Y^{rem}_r$ in $P_B$:
        	$$
        		f^{}_{\preceq}(R_A \setminus R(r)) = (R_B \setminus f^{}_{\preceq}(R(r))) \cup Y^{rem}_r.
        	$$
        	If the removal gadget is empty for all combinatorial elements, we call the modularity \emph{strong}, otherwise \emph{weak}.
        	We substitute the gadgets $Y_x$, for $x \in R(r)$, with the removal gadget $Y^{rem}_r$ in $\Upsilon(R_A)$ correspondingly.
            We consider the elements of a removal gadget to be disjoint from the elements of the original gadgets in order to guarantee pre-image uniqueness.
    \end{enumerate}

This definition of a gadget reduction for combinatorial decision problems ensures that the gadgets are uniquely relatable to the generating combinatorial elements and every element is easily deactivatable.
Note that only combinatorial elements from $P_A$ can be removed such that the new instance $P'_A$ is a validly encoded instance.
That is, combinatorial elements cannot be removed in general as this may void the validity of the instance, e.g. in \textsc{UstCon} the universe elements $s$ and $t$ cannot be deleted.

For the sake of simplicity, we only use gadget reductions originating from \sat{}.
Therefore, we consider the following properties of solutions in a gadget reduction from \sat{}.
These have to be proven individually for each reduction from \sat{}.
For this, let $(L,C)$ be a \sat{} instance that consists of literals $L$ and clauses $C$.
We introduce \emph{variable gadgets} and \emph{clause gadgets}.
\sat{} has literals as universe elements.
Furthermore, it includes the following relations not exclusively:

\begin{tabular}{ll}
	literals and negated literals & $\{(\ell, \overline \ell) \mid \ell \in L\}$\\
	clauses & $\{(\ell^i, \ell^j, \ell^k) \mid (\ell^i, \ell^j, \ell^k) = c \in C \subseteq L^3 \}$\\
	literal and clause & $\{(\ell, c) \mid \ell \in c \in C\}$\\
	negated literal and clause & $\{(\overline \ell, c) \mid \ell \in c \in C\}$
\end{tabular}

A \emph{variable gadget} exists for each literal pair $\ell, \overline \ell$ and consists of the literal gadgets of $\ell$ and $\overline \ell$ as well as the gadget for the relation element $(\ell, \overline \ell)$ of the literals and negated literals relation.
A \emph{clause gadget} simulates a clause.
For this, all gadgets for relations that include a clause (clause, literal an clause, negated literal and clause, literals in clause, negated literals in clause) build up the clause gadget.

We first assume that the solution on the literals, i.e. the variable assignment is one-to-one correspond to the local solution on the variable gadget.
More precisely, let $\ell_i$, $\overline \ell_i$ be the literals corresponding variable $x_i$, then there is exactly one local solution on the variable gadget of $x_i$ that corresponds to the assignment of true to variable $x_i$ and exactly one that corresponds to the assignment of false to the variable $x_i$.
Furthermore, the local solution of the constant gadget is always the same.
For weakly modular reductions, we additionally assume the following solution extension property.
Consider $Y^{rem}_x$ for variable $x$ and all $Y^{rem}_z$ for variables $z \in Z$ such that $x$ and $z$ share a clause.
Then, for each assignment to the variables in $Z$, there needs to be a local solution on $Y^{rem}_x$ and $Y^{rem}_z$ for all $z \in Z$ such that if $Y^{rem}_z$ is deactivated and $Y_z$ is activated for all $z \in Z$, while $Y^{rem}_x$ stays active, the following holds:
For all extending solutions to the assignment to $Z$ in the {\sc 3Sat}-instance, there is an extending solution to the corresponding local solutions on $Y_z$ and the fixed local solution on $Y^{rem}_x$ in the reduction instance.
\footnote{These solution properties were not stated in the conference version. We added them here and as a prerequisite to \Cref{hdrr_class_reduction,hdrr_gamma}.}
Additionally, the solution size has to adapt to the modularity of the gadgets in the universe gadget reduction.
That is, if a combinatorial element in $P_A$ is removed such that the corresponding gadgets in $P_B$ are removed, the solution size of the instance of $P_B$ is well-defined.

\paragraph*{Solution Size}\label{sss:solutionSize}
In order to correctly define the Hamming distance $\kappa$ for a reduction from a problem $\hdrr[P_A]$ to $\hdrr[P_B]$ based on a universe gadget reduction from $P_A$ to $P_B$, we need to find a solution size function.
We demand that each gadget $Y \in \Upsilon$ has a constant \emph{local solution size}, which is defined by the universe gadget reduction.
A Yes-instance has a solution size, which is defined by the sum of all local solution sizes defined as follows.

\begin{definition}[\sat{}-Reduction Solution Size Function]
	Let $P_B$ be a problem such that a universe gadget reduction $f$ from \sat{} to $P_B$ exists.
	Let $(L, C)$ be a \sat{}-instance.
	The gadgets have a local solution size of $size(Y)$ for each $Y \in \Upsilon(L,C)$.
	The function
	$$
		size_f: \sats{} \rightarrow \mathbb{N}: (L, C) \mapsto \sum_{Y \in \Upsilon(L,C)} size(Y)
	$$
	describes the target solution size over universe elements of $f(L, C) = R_B$ for $R_B$ to be a YES-instance of $P_B$.
\end{definition}

In the following, we only consider universe gadget reductions that have such a solution size function.
We assume that the local solution size of each gadget is a constant independent of the generating combinatorial element and which combinatorial elements are active.
That is, all literal/variable gadgets and each gadget of a $k$-clause gadget have the same solution size.
Thus, the solution size function is computable in polynomial time.
While this is necessary, it is not a serious restriction as we see later.
All of the reductions that we present later inherently have this property.

\subsection{Properties of Universe Gadget Reductions}
The definitions of universe gadget reductions and its solutions size function imply the following three properties, which are specifically desired as illustrated before.

\begin{lemma}\label{UGD_properties}
    A universe gadget reduction is total and one-to-many. The inverse to a universe gadget reduction is many-to-one.
\end{lemma}
\begin{proof}
   Let $P_A$ and $P_B$ combinatorial problems with $P_A \preceq^{UGR} P_B$.
   For every relation element $x \in \bigcup_i R_A^i$, the mappings $f_{R_A^i, R_B^j}(x)$ map to corresponding relation elements of $P_B$.
   By definition of a universal gadget reductions every relation element of $P_B$ is generated by such a mapping or is part of the constant gadget $Y_{const}$ such that universal gadget reductions are total.
   By the definition of the mappings and the constant gadget, universe gadget reductions are one-to-many because a relation element $y \in \bigcup_j R_B^j$ of $P_B$ can be only mapped by one mapping from a relation element $x \in \bigcup_i R_A^i$ or is part of $Y_{const}$.
   Analogously, the inverse mapping of the universal gadget reduction is many-to-one.
\end{proof}

Thus by definition, it is ensured that each element $y \in Y_{const} \cup \bigcup_j R_B^j$ of $P_B$ is left unique and thus belongs to exactly one gadget.
Another desirable property is transitivity.
While strongly modular universe gadget reductions are transitive, we have to pay more attention to weakly modular reductions.
This is due to the introduced removal gadgets.
In the case that a strongly modular reduction is chained after a weakly modular reduction, the removal gadget can be transformed again into a removal gadget, making the resulting reduction weakly modular.
In the case that two weakly modular reductions are chained together, the removal gadgets of both reductions may interact with each other.
Then, however, it is not clear how to transform the removal gadgets into a working removal gadget in general.
Thus in general, we do not reach transitivity for weakly modular reductions.

\begin{lemma}\label{UGD_transitive}\footnote{In the conference version, it was stated that weakly modular reductions are transitive. However, this is not the case in general. We adapted the lemma accordingly. \Cref{thm:reductions} is not influenced by this because in all reduction chains there is at most one weakly modular reduction.}
    Polynomial universe gadget reductions are transitive in the following sense:
    \begin{description}
        \item[(1)] strongly modular universe gadget reductions are transitive
        \item[(2)] a strongly modular reduction followed by a weakly modular reduction results in a weakly modular reduction
        \item[(3)] a weakly modular reduction followed by a strongly modular reduction results in a weakly modular reduction
    \end{description}
\end{lemma}
\begin{proof}
    Let $P_A$ be a combinatorial decision problem with relations $R_A$, $P_B$ a combinatorial decision problem with relations $R_B$ and $P_C$ a combinatorial decision problem with relations $R_C$.
    Firstly, we prove that the pre-image uniqueness is upheld.
    Formally, the concatenation of the mappings $f_{R^i_{A}, R^j_{B}}: R^i_{A} \rightarrow R^j_{B}$ and $f_{R^j_{B}, R^k_{C}}: R^j_{B} \rightarrow R^k_{C}$ has to preserve the following property: Let $z \in R^k_{C}$ for some $k \in I_C$, then either $z \in Y^{A \rightarrow C}_{const}$ or there is exactly one $(i, k) \in I_A \times I_C$ and exactly one $x \in R_A^i$ such that $z \in f_{R_A^i, R_C^k}(x)$.\\
    Let $z \in R^k_{C}$ for some $k \in I_C$.
    \begin{description}
    \item[Case 1] $z \in Y^{B \rightarrow C}_{const}$.
    Then $z$ is generated as part of the constant gadget of the reduction from $P_B$ to $P_C$.
    Thus, $z \in Y^{A \rightarrow C}_{const}$.
    \item[Case 2] $z \notin Y^{B \rightarrow C}_{const}$.
            There is exactly one $(j, k) \in I_B \times I_C$ and exactly one $y \in R_A^j$ such that $z \in f_{R_A^j, R_C^k}(y)$.
            Then, $y \in R_A^j$ for some $j \in I_B$.
            \begin{description}
                \item[Case 2.1] $y \in Y^{A \rightarrow B}_{const}$. Then $z$ is generated by exactly one element of $y \in Y^{A \rightarrow B}_{const}$.
                Thus, $z \in Y^{A \rightarrow C}_{const}$.
                \item[Case 2.2] There is exactly one $(i, j) \in I_A \times I_B$ and exactly one $x \in R_A^i$ such that $y \in f_{R_A^i, R_B^j}(x)$.
                Thus by definition, of the universe gadget reduction, $z$ is generated by exactly on $(i, j, k) \in I_A \times I_B \times I_C$ and exactly on $x$ with $z = f_{R_B^j, R_C^k}(f_{R_A^i, R_B^j}(x))$.
            \end{description}
    \end{description}
    It follows that all gadgets of relation $r \in P_A$ elements are pre-image uniquely mapped in the instance $f_{P_B, P_C} \circ f_{P_A, P_B} (R_A)$.
    Furthermore, the modularity of the gadgets is preserved.
    For this, we have to consider the following three cases:
    \begin{description}
        \item[Case 1] For strongly modular reductions $f_{P_A, P_B}$ and $f_{P_B, P_C}$, the concatenation of $f_{P_A, P_B}$ and $f_{P_B, P_C}$ is still strongly modular.
        Specifically if a relation element $r$ in $P_A$ is deleted, its gadgets are deleted from $P_B$ according to the reduction $f_{P_A, P_B}$ and the instance of $P_B$ is the correct instance.
        Because the elements are deleted in $P_B$, the reduction $f_{P_B, P_C}$ continues to delete the corresponding gadgets in $P_C$, whereby the the instance in $P_C$ stays correct for all deletions.
        Accordingly, strongly modular universe gadget reductions are transitive.
        \item[Case 2] Let $f_{P_A, P_B}$ a strongly modular universe gadget reduction and $f_{P_B, P_C}$ a weakly modular universe gadget reduction.
        If a relation element $r$ in $P_A$ is deleted, then its gadgets are deleted from $P_B$ as well according to the reduction $f_{P_A, P_B}$.
        The deletion of the gadgets of $r$ in the instance of $P_B$ results in the introduction of (potentially empty) removal gadgets in $P_C$ according to the weakly modular reduction $f_{P_B, P_C}$.
        This still yields a correct universe gadget reduction, which is weakly modular.
        \item[Case 3] Let $f_{P_A, P_B}$ a weakly modular universe gadget reduction and $f_{P_B, P_C}$ a strongly modular universe gadget reduction.
        The deletion of a relation element $r$ in the instance of $P_A$ results in the introduction of (potentially empty) removal gadgets in $P_B$.
        By definition of weakly universe gadget reductions, this results in a correct instance to which we can apply the strongly modular reduction $f_{P_B,P_C}$.
        The resulting reduction of the concatenation $f_{P_A, P_B} \circ f_{P_B, P_C}$ is weakly modular, where the removal gadgets are defined by
        $$
        f_{P_B, P_C}(Y^r_{rem}) = \bigcup_{\substack{(i, j) \ \in \ I_B \times I_C \\ y \in Y^r_{rem}}} f_{R_B^i, R_C^j}(y)
    $$
    \end{description}

It follows that the concatenation of two strongly modular universe gadget reductions fulfill the pre-image uniqueness and strong modularity.
Furthermore, the concatenation of a weakly modular gadget reduction with a strongly modular gadget reduction (independent of the order of concatenation) results in a reduction that is pre-image unique and fulfills weak modularity.
\end{proof}

Furthermore, the solution size function adheres to the modularity of the universe gadget reduction.

\begin{lemma}
	The solution size function adheres to modularity.
	In other words, let $(L, C)$ and $(L', C')$ be instances of \sat{} with $L' \subseteq L$ and $C' \subseteq C$.
	Furthermore, let $f$ be a universe gadget reduction from \sat{} to $P_B$ such that $f(L',C')$ results from $f(L,C)$ by removing the corresponding gadgets.
	Then,
		$$
				size_f(L',C') = \sum_{Y \in \Upsilon(L',C')} size(Y).
		$$
\end{lemma}
\begin{proof}
    If a gadget $Y_r$ is removed, the solution is decreased by $size(Y_r)$ and increased by the local solution size of the removal gadget $size(Y^{rem}_r)$.
    Because a solution size function is the sum of the local solutions size of each gadget, the following holds:
	\begin{flalign*}
		size_f(L',C') &= \sum_{Y \in \Upsilon(L',C')} size(Y)&&\\
		&= size(Y_{const}) + \sum_{x \in (L', C')} size(Y_x) + \sum_{x \in ((L, C) \setminus (L',C'))} size(Y^{rem}_x). &&
	\end{flalign*}
    Accordingly, by the definition of the solution size function to be the sum of the constant local solution sizes of the gadgets, it adheres to modularity.
\end{proof}

Now, we present a general reduction from $\exists\forall\exists$\sat{} to the Hamming distance recoverable robust $\hdrr[P_B]$ based on the structure that a universe gadget reduction provides.
That is if there is a polynomial time universe gadget reduction $f$ from \sat{} to $P_B$ such that the solution properties hold and a corresponding polynomial time solution size function $size_f$ exists, then there is a polynomial time reduction for the Hamming distance recoverable robust version of $P_B$ with Hamming distance over the universe elements, transforming the scenarios accordingly.

\begin{theorem}\label{UGD_polytime}\label{hdrr_class_reduction}
    If $\sat$ is universe gadget reducible to $P_B$ in polynomial time such that
    there is a corresponding solution size function, and the solution properties hold,
    then there is a polynomial time reduction from $\exists\forall\exists$\sat{} to $\hdrr[P_B]$, where the Hamming distance is defined over the solution ground set and the scenario encodings are $xor$-dependency scenarios.
\end{theorem}
\begin{proof}
	In the following, we prove that $\exists\forall\exists$\sat{} is reducible to $\hdrr[P_B]$.
    For this, we reuse the reduction from $\exists\forall\exists$\sat{} to \hdrr[\sat] together with the universe gadget reduction from \sat{} to $P_B$.
    The basic idea is to substitute the variables and clauses by the corresponding gadgets.
    With this reduction, we also directly prove the induction base as in the proof of \Cref{inductionmhdrr}.
    Furthermore, we can apply the same argument for the induction step to this reduction as well.
    That is the set $Z$ is able to absorb the lower levels of uncertainty for an additional stage of recoverable robustness.

	Now, let $\exists X \forall Y \exists Z~C(X, Y, Z)$ be the $\exists\forall\exists$\sat{} instance of variable sets $X$, $Y$, and $Z$ as well as clauses $C$.
    In order to construct the instance of $P_B$, we store the gadgets defined by each mapping $f_{R_{\sats}^i, R_B^j}(r)$ in a table, for every element $r \in R_{\sats}^i$ of all relations $R_{\sats}^i$ residing in the \sat{} input.
    We can compute this table, because we have a polynomial time universe gadget reduction between \sat{} and $P_B$.
	With this table, we can now compute the scenarios in polynomial time with the following principle.
	The idea is to activate the variable and clause gadgets, whenever the variable or clause is active.
    As universe gadget reductions are modular, the (de)activation of a variable or clause is easily translatable into the instance of $P_B$:
    We remove the corresponding gadgets (and introduce necessary removal gadgets).
    While the (de)activation of variables from $X$ and $Z$ is straightforward, we have to take care about the (de)activation of the gadgets for variables from $Y$ according to the given structure of uncertainty.
	More precisely for each variable $y_i \in Y$, we deactivate the variable gadget for $y^t_i$ or $y^f_i$ and thus also the corresponding clauses as in \Cref{tab:deactivationClauses}.
    We model this operation with the $xor$-dependencies by adding
    \begin{itemize}
        \item the gadget of variable $y^t_i$ and its clauses together with the variable removal gadget of $y^t_i$ into the corresponding set $E_{i,1}$
        \item the gadget of variable $y^f_i$ and its clauses together with the variable removal gadget of $y^f_i$ into the corresponding set $E_{i,2}$.
    \end{itemize}
    
	\begin{table}[!ht]
		\begin{center}
			\begin{tabular}{l|l|l}
				& $y_i = y^t_i$ & $y_i = y^f_i$ \\
				\hline
				$xor$-dependencies & $y^t_i \leftrightarrow 1$ & $y^f_i \leftrightarrow 0$ \\
                & $y^t_i \leftrightarrow y^t_{i,1}$ & $y^f_i \leftrightarrow y^f_{i,0}$ \\
				& $y^t_i \leftrightarrow \overline y^t_{i,0}$ & $y^f_i \leftrightarrow \overline y^f_{i,1}$ \\
			\end{tabular}
		\end{center}
		\caption{The clauses to (de)activate for $xor$-dependencies.}
		\label{tab:deactivationClauses}
	\end{table}
    
	These (de)activations are possible because the reduction is modular.
    In the base scenario, the variable gadgets of $y^t_i$ and $y^f_i$ are inactive, while the removal gadgets of them are active.
    Consequently if the variable $y^t_i$ (respectively $y^f_i$) is activated by the set $E_{i,1}$ (respectively $E_{i,2}$) in an uncertainty scenario, the removal gadget of $y^t_i$ (respectively $y^f_i$) is removed and the variable gadget is added.
    Additionally, the corresponding clauses that contain $y^t_i$ (respectively $y^f_i$) are activated and we obtain the instance where the gadgets of $y^t_i$ (respectively $y^f_i$) are active simulating the variable accordingly.
    Note that these are exactly the clauses from \Cref{tab:deactivationClauses}, which also have the same structure (six clauses of two variables each by substituting $a \leftrightarrow b$ with clauses $(\overline a \lor b)$ and $(a \lor \overline b)$).
    Therefore, the variable gadgets respectively the variable removal gadgets of $y^t_i$ and $y^f_i$ have the same overall local solution size for each $i \in \{1, \ldots, |Y|\}$.
    By the correctness of the reduction from $\exists\forall\exists$\sat{} to \hdrr[\sat] and the universe gadget reduction from \sat{} to $P_B$, the reduction for each scenario remains correct and the one-to-one correspondence of the activeness between the variable and the variable gadget is uphold.

    In order to describe the scenarios formally, we summarize the reduction from $\exists\forall\exists$\sat{} to \hdrr[\sat] with $xor$-dependencies.
    First, we substituted all variables $y_i \in Y$ by two variables $y^t_i$ and $y^f_i$.
    We then replaced the occurrences of the literal $y_i$ in all clauses by $y^t_{i, 1}, y^t_{i, 0}, y^f_{i, 1}, y^f_{i, 0}$ resulting in a duplication of clauses.
    The additional variables $y^t_{i, 1}, y^t_{i, 0}, y^f_{i, 1}, y^f_{i, 0}$ were added to the set $Z'$, which was then defined by $Z' = Z \cup \{y^t_{i, 1}, y^t_{i, 0}, y^f_{i, 1}, y^f_{i, 0} \mid y_i \in Y\}$.
    The $xor$-dependencies were defined over the variables $y^t_i$ and $y^f_i$ and their clauses (compare \Cref{tab:deactivationClauses}).

    We are now ready to describe the $xor$-dependencies formally.
    For this, we have to define the base scenario $\sigma_0$ and a set $E'$ together with pairs of sets $(E_{i,1}, E_{i,2})$.
    For the base scenario, we compute the reduction instance of all variables and clauses that are available in any of the scenarios.
    We now deactivate the necessary gadgets to obtain an instance that corresponds to the base scenario.
    The base scenario contains only the variable gadgets from $X$.
    All other gadgets are removed and replaced by the corresponding removal gadgets, where the variable removal gadgets of $Y' = Y^t \cup Y^f$ are deactivated first.
    Accordingly, we define $\sigma_0 = \{Y_x \cup Y^{rem}_a \mid x \in X \text{ and } a \in Y' \cup Z' \cup C'|_X\}$.
    To construct the set $E'$, we observe that the variables of $X$ and $Z'$ and all clauses that contain variables from $X$ or $Z'$ are available in all uncertainty scenarios.
    Thus, the set $E'$ contains the variable gadgets of variables from $Z'$ as well as the clause gadgets for clauses containing $X$ and $Z'$.
    Remember that variables $y^t_{i, 1}, y^t_{i, 0}, y^f_{i, 1}, y^f_{i, 0}$, for $1 \leq i \leq |Y|$, were added to set $Z'$ and are thus active in exactly all uncertainty scenarios.
    Accordingly, we define $E' = \{Y^{rem}_a \cup Y_b \mid a \in Z' \cup C'|_X \text{ and } b \in Z', C'|_{X,Z'}\}$.
    
    At last, we define the pairs of sets $(E_{i,1}, E_{i,2})$, where either $E_{i,1}$ or $E_{i,2}$.
    For this, we consider a variable $y_i \in Y$.
    Each of these variables was split into two variables $y^t_i$ and $y^f_i$ and the corresponding clauses from \Cref{tab:deactivationClauses} were introduced.
    On the one hand, we define
    $$
        E_{i,1} = \{Y^{rem}_{y^t_i}, Y_{y^t_i}, Y_{(y^t_i \leftrightarrow y^t_{i,1})}, Y_{(y^t_i \leftrightarrow \overline y^t_{i,0})}, Y_{(y^t_i \leftrightarrow 1)}\}
    $$
    to include the variable gadget of $y^t_i$ and its (possibly empty) removal gadget.
    On the other hand, we define
    $$
        E_{i,2} = \{Y^{rem}_{y^f_i}, Y_{y^f_i}, Y_{(y^f_i \leftrightarrow y^f_{i,0})}, Y_{(y^f_i \leftrightarrow \overline y^f_{i,1})}, Y_{(y^f_i \leftrightarrow 0)}\}
    $$
    to include the variable gadget of $y^f_i$ as well as its (possibly empty) removal gadget.
    Note that in the base scenario the removal gadget of $y^t_i$ (respectively $y^f_i$) was active and the variable gadget inactive.
    Thus if $y^t_i$ (respectively $y^f_i$) is activated, then the variable gadget of $y^t_i$ (respectively $y^f_i$) is activated and the removal gadget is deactivated.

	As we have considered $f$ to be a modular reduction based on the $\exists\forall\exists$\sat{} instance $(L,C)$, we can set
    $$
        \kappa = size_f(X \cup \overline X, \emptyset) + size_f(X \cup \overline X \cup Y' \cup Z' \cup \overline Z', C') - 2 \sum_{x \in X} size(Y_x).
    $$
    This is correct, because the gadgets of variables from $X$ as well as the removal gadgets of $|Y|$ many variables from $Y'$ are present in all scenarios while all other variables are only present either in the base scenario or the uncertainty scenarios.
    Furthermore for each variable $y_i$, there is one variable gadget active (e.g. of $y^t_i$) and one variable removal gadget (e.g. of $y^f_i$) active.
    This guarantees that the solution size for all uncertainty scenarios is the same because there are always $|Y|$ variable gadgets active and $|Y|$ variable removal gadgets active and we assume the gadgets to have a constant local solution size.
    Now consider the solution on the variables on $X$.
    The base scenario activates only the variable gadgets of the set $X$.
    The solution size of $size_f(X \cup \overline X, \emptyset) + size_f(X \cup \overline X \cup Y' \cup Z' \cup \overline Z', C')$ is exactly the sum of the solution sizes of the base scenario and any of the uncertainty scenarios without considering $|Y|$ of the gadgets corresponding to variables from $Y'$.
    Since we subtract $2 \sum_{x \in X} size(Y_x)$ to define $\kappa$ and all elements but the elements from $X$ are $Y$ many variable gadgets from $Y'$, the partial solution to the variables of $X$ has to stay the same by switching from the base scenario to any of the uncertainty scenarios.
    In conclusion, the partial solution on $X$ is the same for all of the uncertainty scenarios.

    By the correctness of the underlying gadget reduction, we obtain the corresponding reduction instance for the base scenario and for each of the uncertainty scenarios.
    We can now use the same argumentation as in the reduction from $\exists\forall\exists$\sat{} to $\hdrr[\sat]$.
    First in the base scenario, a partial solution to the gadgets of the $\exists$-quantified variables $X$ is fixated because the Hamming distance is chosen accordingly.
    Therefore, the partial solution on the gadgets of $X$ is the same for the base scenario and all uncertainty scenarios such that the order of quantification is followed.
    This solution corresponds to an assignment of the variables $X$.
    In the uncertainty scenarios, the partial solution on the variable gadgets of $X$ has to be extended to a complete solution of the reduction instance.
    Furthermore, each of the possible truth assignments on the variables $Y$ is simulated by exactly the corresponding uncertainty scenario, in which the corresponding gadgets are activated.
    Again the local solutions on the variable gadgets of $Y'$ correspond to an assignment of the variables of $Y$.
    At last, we need to extend the solution on the variable gadgets of $X$ and $Y'$ to a complete solution.
    For this, a local solution to the variable gadgets of $Z'$ and the clause gadgets need to be found.
    Note that this can be freely chosen because the Hamming distance $\kappa$ is large enough.
    If and only if this is possible, we have a \textsc{yes}-instance for the corresponding uncertainty scenario, where the assignment on the variable gadgets of $X, Y'$, and $Z'$ in $\hdrr[P_B]$ correspond to an assignment to the variables $X, Y$, and $Z$ in $\exists\forall\exists$\sat{}.
    Note that in the case of weak modularity, the removal gadgets in $E_{i,1}$ for $y^t_i$ ($E_{i,2}$ for $y_i^f$, analogously) might stay active in the base and uncertainty scenarios.
    Since the variables $y^t_{i,0}$ and $y^t_{i,1}$ are assigned to true if $y^t_i$ is inactive, an existing local solution on the variable removal and clause removal gadgets can be extended in the uncertainty scenarios due to the solution extension property.
    For this, the local solution corresponding to the assignment to true is applied to the gadgets for $y^t_{i,0}$ and $y^t_{i,1}$ if $E_{i,1}$ is inactive in the uncertainty scenarios.
    Otherwise if the removal gadgets of $E_{i,1}$ are deactivated and the actual variable and clause gadgets activated, no solution on the now active variable and clause gadgets is fixed and thus can be freely assigned.
    It follows that the $\exists\forall\exists$\sat{} is satisfiable if and only if the base scenario as well as all uncertainty scenarios of the $\hdrr[P_B]$ instance are \textsc{yes}-instances such that the Hamming distance is at most $\kappa$.
\end{proof}

We can derive a similar result for $\Gamma$-set scenarios by reusing the construction from above.

\begin{theorem}\label{hdrr_gamma}
    If $\sat$ is universe gadget reducible to $P_B$ in polynomial time such that
    there is a corresponding solution size function, and the solution properties hold,
    then there is a polynomial time reduction from $\exists\forall\exists$\sat{} to $\hdrr[P_B]$, where the Hamming distance is defined over the solution ground set and the scenario encodings are $\Gamma$-set scenarios.
\end{theorem}
\begin{proof}
    We modify the construction from \Cref{hdrr_class_reduction} by introducing gadgets for the $s$-variables as in \Cref{sat-hdrr-sigma-3-hard2}.
    Let $E^t_{i} = E_{i,1}$ and $E^f_{i} = E_{i,2}$.
    Instead of having one $s_i$ for the pair of variables $y^t_i$ and $y^f_i$, we split the $s_i$ into two variables $s^t_i$ and $s^f_i$ to add the corresponding clauses into the sets $E^t_{i}$ corresponding to activate $y^t_i$ (i.e. setting $y_i$ to $1$) and $E^f_{i}$ corresponding to activate $y^f_i$ (i.e. setting $y_i$ to $0$).
    Specifically, we add the gadgets for the variables $s, s^t_1, s^f_i, \ldots, s^t_{|Y|}, s^f_{|Y|}$ and the corresponding clauses for $(s \lor (\overline s^t_1 \wedge \overline s^f_1) \lor \ldots \lor (\overline s^t_{|Y|} \wedge \overline s^f_{|Y|}))$, and $(\overline y^t_i \lor s^t_i)$ and $(y^f_i \lor s^f_i)$ for all $i \in \{1, \ldots, |Y|\}$.
    Because the variables $s, s_1, \ldots, s_{|Y|}$ are part of $Z'$ in \Cref{sat-hdrr-sigma-3-hard2}, we add the variable gadgets for $s, s^t_1, s^f_i, \ldots, s^t_{|Y|}, s^f_{|Y|}$ to the set $E'$.
    We additionally add the gadget of clause $(\overline y^t_i \lor s^t_i)$ to $E^t_{i}$, where all other gadgets (variable, clause, and variable removal) of $y^t_i$ still reside and we add analogously the gadget of clause $(y^f_i \lor s^f_i)$ to $E^f_{i}$.
    Accordingly if $E^t_{i}$ (respectively $E^f_i$) is not activated, $s^t_i$ (respectively $s^f_i$) can be set to $0$ because the clause $(\overline y^t_i \lor s^t_i)$ (respectively $(\overline y^f_i \lor s^f_i)$) is removed.
    Furthermore, the subformula $(s \lor (\overline s^t_1 \wedge \overline s^f_1) \lor \ldots \lor (\overline s^t_{|Y|} \wedge \overline s^f_{|Y|}))$ works equivalently to the clause $(s \lor \overline s_1 \lor \ldots \lor \overline s_{|Y|})$ in \Cref{sat-hdrr-sigma-3-hard2} because both $s^t_i$ and $s^f_i$ need to be set to $0$ such that $s$ can be set to $0$.
    We have to analyze the following three cases.
    
    If exactly one of $E^t_{i}$ and $E^f_{i}$ is active, then we have the same situation as in the proof of \Cref{hdrr_class_reduction} for $xor$-dependencies.
    Accordingly, we still have to prove the correctness for the other two cases.
    
    If there is an $i \in \{1, \ldots, |Y|\}$ such that neither of $E^t_{i}$ and $E^f_{i}$ are active, then the removal gadgets of both $y^t_i$ and $y^f_i$ are still active as in the base scenario.
    Thus, a correct instance of $P_B$ is induced.
    Because the \sat{} formula is trivially solvable by setting $s = 0$, the resulting instance of $P_B$ is a \textsc{yes}-instance.
    The Hamming distance of $\kappa$ is sufficient because at most $\Gamma = |Y|$ many gadgets of variables from $Y'$ are activated to which the solution has to be switched.

    If there is an $i \in \{1, \ldots, |Y|\}$ such that both $E^t_{i}$ and $E^f_{i}$ are active, the variable and clause gadgets of both $y^t_i$ and $y^f_i$ are active while the removal gadgets of both $y^t_i$ and $y^f_i$ are inactive.
    Again a correct instance of $P_B$ is induced, where both $y^t_i$ and $y^f_i$ are simulated to be active.
    The instance is, however, also a \textsc{yes}-instance due to the pigeonhole principle and $\Gamma \leq |Y|$.
    That is, there is a $j \in \{1, \ldots, |Y|\}$ such that neither $E^t_{j}$ nor $E^f_{j}$ are active.
    Accordingly, the corresponding \sat{} instance is satisfiable by setting $s = 0$ and a correct \textsc{yes}-instance of $P_B$ is induced as in the case above.

    Observe that the subformula $(s \lor (\overline s^t_1 \wedge \overline s^f_1) \lor \ldots \lor (\overline s^t_{|Y|} \wedge \overline s^f_{|Y|}))$ can be transformed into a CNF by Tseitin's transformation \cite{tseitin1983complexity}.
    Furthermore all clauses of length greater than four can be transformed into clauses of length three by Karp's reduction from \textsc{CNF-Sat} to \textsc{3Sat} \cite{DBLP:conf/coco/Karp72}.
    The newly introduced variables can be added to the set $E'$.
\end{proof}

With these structural properties in mind, we can construct a whole set of Hamming distance recoverable robust problems.
Note that the transitivity of the universe gadget reduction can be used to deduce further reductions.

\subsection{Gadget Reductions for Various Combinatorial Decision Problems}
In this section, we examine various but not all problems that are universe gadget reducible from \sat.
The reductions are all well-known results or modifications of well-known results.
We adapt these results to the universe gadget reduction framework to indicate that \Cref{hdrr_class_reduction,hdrr_gamma} are general statements.
We prove the following theorem by showing that a universe gadget reduction from \sat{} exists for all the problems.
For this, we use the transitivity of the reductions as illustrated in \Cref{fig:reductions}.

\begin{restatable}{theorem}{reductions}\label{thm:reductions}
	The $m$-Hamming distance recoverable robust version of the following problems are \NP-complete with polynomially computable scenarios and $\Sigma^p_{2m+1}$-complete with \textit{xor}-dependency scenarios or $\Gamma$-set-scenarios: 
	\vc,
	\ds,
	\fas,
	\fvs,
	\hs,
	\is,
	\cl,
	\sss,
	\ks,
	\pa,
	\tms,
	\udhc,
	\udhp,
	\tsp,
	\tdm,
	\xtc,
	\kddp \ $(k \geq 2)$,
	\tcol,
	\kcol,
	\cc.
\end{restatable}

\begin{figure}[!ht]
	\centering
	\scalebox{0.56}{
    	\begin{tikzpicture}[scale=1]
    		\node[text width=1cm,align=center] (0) at (0, 0) {\textsc{3Sat}};
    
    		\node[] (1) at (-6, -1) {\textsc{VC}};
    		\node[] (2) at (-3.75, -1) {\textsc{IS}};
    		\node[] (3) at (-1, -1) {\textsc{Subset Sum}};
    		\node[] (4) at (6, -1) {\textsc{3DM}};
    		\node[] (5) at (2, -1) {\textsc{DHP}};
    		\node[] (6) at (4, -1) {\textsc{2DDP}};
            \node[] (7) at (8, -1) {\textsc{3Col}};
    
    		\node[] (11) at (-8, -1.75) {\textsc{DS}};
    		\node[] (12) at (-7, -1.75) {\textsc{FAS}};
    		\node[] (13) at (-6, -1.75) {\textsc{FVS}};
    		\node[] (14) at (-5, -1.75) {\textsc{HS}};
    
    		\node[] (21) at (-3.75, -1.75) {\textsc{Clique}};
    
    		\node[] (31) at (-2, -1.75) {\textsc{Partition}};
    		\node[] (32) at (0, -1.75) {\textsc{Knapsack}};
    
    		\node[] (41) at (6, -1.75) {\textsc{X3C}};
    
    		\node[] (51) at (1.5, -1.75) {\textsc{DHC}};
    		\node[] (52) at (2.5, -1.75) {\textsc{UHP}};
    
    		\node[] (61) at (4, -1.75) {k\textsc{DDP}};
            
            \node[] (71) at (8, -1.75) {k\textsc{Col}};
            
    		\node[] (311) at (-2, -2.5) {\textsc{Scheduling}};
      
    		\node[] (511) at (1.5, -2.5) {\textsc{UHC}};
    		\node[] (512) at (1.5, -3.25) {\textsc{TSP}};
    	
            \node[] (711) at (8, -2.5) {\textsc{Clique Cover}};
            
    		\path[->] (0) edge (1);
    		\path[->] (0) edge (2);
    		\path[->] (0) edge (3);
    		\path[->] (0) edge (4);
    		\path[->] (0) edge (5);
    		\path[->] (0) edge (6);
    
            \path[->] (1) edge (11);
    		\path[->] (1) edge (12);
    		\path[->] (1) edge (13);
    		\path[->] (1) edge (14);
    
    		\path[->] (2) edge (21);
    
    		\path[->] (3) edge (31);
    		\path[->] (3) edge (32);
    		\path[->] (31) edge (311);
    
    		\path[->] (4) edge (41);
    
    		\path[->] (5) edge (51);
    		\path[->] (5) edge (52);
    		\path[->] (51) edge (511);
    		\path[->] (511) edge (512);
    
    		\path[->] (6) edge (61);
    
    		\path[->] (0) edge (7);
    		\path[->] (7) edge (71);
    		\path[->] (71) edge (711);
    	\end{tikzpicture}
	}
	\caption{The tree of gadget reductions for all considered problems.}
	\label{fig:reductions}
\end{figure}

\subsubsection{\textsc{Vertex Cover}}
As an introductory example, we take a close look at a universe gadget reduction of \sat{} to \vc, which was initially developed by Garey and Johnson \cite{DBLP:books/fm/GareyJ79}.
This example directly proves \Cref{vc_gadget_reducible}.
For the \vc{} reduction we use the very fine-grained universe gadget reduction for each combinatorial element.
In the following reductions, however, we directly use variable and clause gadgets as described in \Cref{sss:solutionSize}, to shorten our argumentation.

\begin{lemma}\label{vc_gadget_reducible}
    \sat{} is strongly modular universe gadget reducible to \vc{} such that the solution properties hold and a solution size function for this reduction exists.
\end{lemma}

\begin{proof}
	The problem \sat{} consists of the universe $L$ for the literals and the relations
	\begin{itemize}
		\item $R_{\ell,\overline{\ell}}$ that relates a literal $\ell$ with its negation $\overline{\ell}$,
		\item $R_{\ell, c}$ that relates a literal $\ell$ to a clause $c$, iff $\ell \in c$ and
		\item $R_{\ell, \ell', c}$ that relates literals $\ell$ and $\ell'$, iff $\ell, \ell' \in c$.
	\end{itemize}
	The problem \vc{}, on the other hand, consists of vertices $V$ and edges $E$ that form a graph $G =(V, E)$.
	Based on these universe and relations, the gadgets as in \Cref{fig:sat-vc:reduction:gadgets} can be found.
	Therefore, we define the mappings:
	$$
		f_{L, V}, f_{L, E}, f_{R_{\ell,\overline{\ell}}, V}, f_{R_{\ell,\overline{\ell}}, E}, f_{R_{\ell, c}, V}, f_{R_{\ell, c}, E}, f_{R_{\ell, \ell', c}, V}, f_{R_{\ell, \ell', c}, E}
	$$
	The dashed vertices in \Cref{fig:sat-vc:reduction:gadgets} indicate that these are part of a different gadget.
	\begin{figure}[!ht]
		\centering
		\begin{subfigure}[b]{0.45\textwidth}
			\centering
			\scalebox{1}{
				\begin{tikzpicture}[scale=1,
						node/.style = {shape=circle, draw, inner sep=0pt, minimum size=0.25cm},
						dashednode/.style = {shape=circle, draw=blue, dashed, inner sep=0pt, minimum size=0.25cm},
						textnode/.style = {shape=circle, draw, inner sep=0pt, minimum size=0.4cm},
						dashedtextnode/.style = {shape=circle, draw=blue, dashed, inner sep=0pt, minimum size=0.4cm},
						box/.style = {rectangle, fill=gray!20, rounded corners, fill opacity=1, inner sep=1pt}]
					\node[node] (l) at (0, 0) {}; \node[] at ($(l) + (0,0.33)$) {$v_\ell$};
				\end{tikzpicture}
			}
			\caption{
				Literal Gadget for literal $\ell \in L$.
				The corresponding mappings are $f_{L, V}: \ell \mapsto \{v_\ell\}$ and $f_{L, E}: \ell \mapsto \emptyset$.
			}
			\label{fig:vc:literal_gadget}
		\end{subfigure}
		\hfill
		\begin{subfigure}[b]{0.45\textwidth}
			\centering
			\scalebox{1}{
				\begin{tikzpicture}[scale=1,
						node/.style = {shape=circle, draw, inner sep=0pt, minimum size=0.25cm},
						dashednode/.style = {shape=circle, draw=blue, dashed, inner sep=0pt, minimum size=0.25cm},
						textnode/.style = {shape=circle, draw, inner sep=0pt, minimum size=0.4cm},
						dashedtextnode/.style = {shape=circle, draw=blue, dashed, inner sep=0pt, minimum size=0.4cm},
						box/.style = {rectangle, fill=gray!20, rounded corners, fill opacity=1, inner sep=1pt}]
					\node[dashednode] (t) at (0, 0) {}; \node[] at ($(t) + (0,0.33)$) {$v_\ell$};
					\node[dashednode] (f) at (1, 0) {}; \node[] at ($(f) + (0,0.33)$) {$v_{\overline{\ell}}$};
					\path[-] (t) edge node[right] {} (f);
				\end{tikzpicture}
			}
			\caption{
				Gadget for relation $R_{\ell,\overline{\ell}}$ for some literal $\ell \in L$.
				The corresponding mappings are $f_{R_{\ell,\overline{\ell}}, V}: (\ell, \overline{\ell}) \mapsto \emptyset$ and $f_{R_{\ell,\overline{\ell}}, E}: (\ell, \overline{\ell}) \mapsto \{\{v_\ell, v_{\overline{\ell}}\}\}$.
			}
			\label{fig:vc:literalnegation}
		\end{subfigure}
		\hfill\\
		\hfill\\
		\begin{subfigure}[b]{0.45\textwidth}
			\centering
			\scalebox{1}{
				\begin{tikzpicture}[scale=1,
						node/.style = {shape=circle, draw, inner sep=0pt, minimum size=0.25cm},
						dashednode/.style = {shape=circle, draw=blue, dashed, inner sep=0pt, minimum size=0.25cm},
						textnode/.style = {shape=circle, draw, inner sep=0pt, minimum size=0.4cm},
						dashedtextnode/.style = {shape=circle, draw=blue, dashed, inner sep=0pt, minimum size=0.4cm},
						box/.style = {rectangle, fill=gray!20, rounded corners, fill opacity=1, inner sep=1pt}]
					\node[dashednode] (l) at (0, 0) {}; \node[] at ($(l) + (0,0.33)$) {$v_\ell$};
					\node[node] (lc) at (1, -1) {}; \node[] at ($(lc) + (0.5,0)$) {$v_{\ell, c}$};
					\path[-, out=270, in=90] (l) edge (lc);
				\end{tikzpicture}
			}
			\caption{
				Gadget for relation $R_{\ell, c}$ for literal $\ell \in L$ with $\ell \in c \in C$.
				The corresponding mappings are $f_{R_{\ell, c}, V}: (\ell, c) \mapsto \{v_{\ell, c}\}$ and $f_{R_{\ell, c}, E}: (\ell, c) \mapsto \{\{v_\ell, v_{\ell, c}\}\}$.
			}
			\label{fig:literalclause}
		\end{subfigure}
		\hfill
		\begin{subfigure}[b]{0.45\textwidth}
			\centering
			\scalebox{1}{
				\begin{tikzpicture}[scale=1,
						node/.style = {shape=circle, draw, inner sep=0pt, minimum size=0.25cm},
						dashednode/.style = {shape=circle, draw=blue, dashed, inner sep=0pt, minimum size=0.25cm},
						textnode/.style = {shape=circle, draw, inner sep=0pt, minimum size=0.4cm},
						dashedtextnode/.style = {shape=circle, draw=blue, dashed, inner sep=0pt, minimum size=0.4cm},
						box/.style = {rectangle, fill=gray!20, rounded corners, fill opacity=1, inner sep=1pt}]
					\node[dashednode] (lc) at (0, 0) {}; \node[] at ($(lc) + (-0.1,0.33)$) {$v_{\ell, c}$};
					\node[dashednode] (l'c) at (1, 1) {}; \node[] at ($(l'c) + (0,0.33)$) {$v_{\ell', c}$};
					\path[-] (lc) edge node[right] {} (l'c);
				\end{tikzpicture}
			}
			\caption{
				Gadget for relation $R_{\ell, \ell', c}$ for literals $\ell, \ell' \in L$ with $\ell, \ell' \in c \in C$.
				The corresponding mappings are $f_{R_{\ell, \ell', c}, V}: (\ell, \ell', c) \mapsto \emptyset$ and $f_{R_{\ell, \ell', c}, E}: (\ell, \ell', c) \mapsto \{\{v_{\ell, c}, v_{\ell', c}\}\}$.
			}
			\label{fig:literalliteralclause}
		\end{subfigure}
		\centering
		\caption{The gadgets for the universe and all relations for the reduction from \sat{} to \vc{}.}
		\label{fig:sat-vc:reduction:gadgets}
	\end{figure}

    A complete example can be found in \Cref{fig:sat-vc:reduction}.
	\begin{figure}[!ht]
		\centering
		\scalebox{1}{
				\begin{tikzpicture}[scale=1,
						node/.style = {shape=circle, draw, inner sep=0pt, minimum size=0.25cm},
						dashednode/.style = {shape=circle, draw=blue, dashed, inner sep=0pt, minimum size=0.25cm},
						textnode/.style = {shape=circle, draw, inner sep=0pt, minimum size=0.4cm},
						dashedtextnode/.style = {shape=circle, draw=blue, dashed, inner sep=0pt, minimum size=0.4cm},
						box/.style = {rectangle, fill=gray!20, rounded corners, fill opacity=1, inner sep=1pt}]
				\node[node] (lt1) at (0, 0) {}; \node[] at ($(lt1) + (0,0.33)$)  {$v_{\ell_1}$};
				\node[node] (lf1) at (1, 0) {}; \node[] at ($(lf1) + (0,0.33)$)  {$v_{\overline{\ell}_1}$};
				\node[node] (lt2) at (2, 0) {}; \node[] at ($(lt2) + (0,0.33)$)  {$v_{\ell_2}$};
				\node[node] (lf2) at (3, 0) {}; \node[] at ($(lf2) + (0,0.33)$)  {$v_{\overline{\ell}_2}$};
				\node[node] (lt3) at (4, 0) {}; \node[] at ($(lt3) + (0,0.33)$)  {$v_{\ell_3}$};
				\node[node] (lf3) at (5, 0) {}; \node[] at ($(lf3) + (0,0.33)$)  {$v_{\overline{\ell}_3}$};
				\path[-] (lt1) edge node[right] {} (lf1);
				\path[-] (lt2) edge node[right] {} (lf2);
				\path[-] (lt3) edge node[right] {} (lf3);
				\node[node] (c1lt1) at (1, -1.5) {}; \node[] at ($(c1lt1) + (0.5,0.2)$) {$v_{\ell_1, c_1}$};
				\node[node] (c1lt2) at (0.25, -2.25) {}; \node[] at ($(c1lt2) + (0,-0.33)$)  {$v_{\ell_2, c_1}$};
				\node[node] (c1lt3) at (1.75, -2.25) {}; \node[] at ($(c1lt3) + (0,-0.33)$)  {$v_{\ell_3, c_1}$};
				\node[node] (c2lf1) at (4, -1.5) {}; \node[] at ($(c2lf1) + (-0.45,0.25)$)  {$v_{\overline{\ell}_1, c_2}$};
				\node[node] (c2lf2) at (3.25, -2.25) {}; \node[] at ($(c2lf2) + (0,-0.33)$)  {$v_{\overline{\ell}_2, c_2}$};
				\node[node] (c2lt3) at (4.75, -2.25) {}; \node[] at ($(c2lt3) + (0,-0.33)$)  {$v_{\ell_3, c_2}$};
				\path[-] (c1lt1) edge node[right] {} (c1lt2);
				\path[-] (c1lt1) edge node[right] {} (c1lt3);
				\path[-] (c1lt2) edge node[right] {} (c1lt3);
				\path[-] (c2lf1) edge node[right] {} (c2lf2);
				\path[-] (c2lf1) edge node[right] {} (c2lt3);
				\path[-] (c2lf2) edge node[right] {} (c2lt3);
				\path[-, out=270, in=90] (lt1) edge node[right] {} (c1lt1);
				\path[-, out=270, in=90] (lt2) edge node[right] {} (c1lt2);
				\path[-, out=270, in=90] (lt3) edge node[right] {} (c1lt3);
				\path[-, out=270, in=90] (lf1) edge node[right] {} (c2lf1);
				\path[-, out=270, in=90] (lf2) edge node[right] {} (c2lf2);
				\path[-, out=270, in=90] (lt3) edge node[right] {} (c2lt3);
			\end{tikzpicture}
		}
		\caption{The reduction graph for \sat{} formula $C=\{\{\ell_1, \ell_2, \ell_3\}, \{\overline{\ell_1}, \overline{\ell_2}, \ell_3\}\}$.}
		\label{fig:sat-vc:reduction}
	\end{figure}
    On the other hand, the reduction based on variable and clause gadgets can also be established.
    For this, the relations from above are combined in the gadgets.
	\begin{itemize}
		\item The universe $L$ is combined with relation $R_{\ell,\overline{\ell}}$ to a \textit{variable gadget for variable $x \in X$}.
		\item The relations $R_{\ell, c}$ and $R_{\ell, \ell', c}$ are combined to one clause gadget that connects the corresponding variable gadget correctly to a clause $c \in C$.
	\end{itemize}
	These gadgets are depicted in \Cref{fig:sat-vc:reduction:gadgets:coarse}, in which the dashed vertices indicate that these are part of a different gadget.
    \begin{figure}[!ht]
        \centering
        \begin{subfigure}[b]{0.45\textwidth}
            \centering
	        \scalebox{1}{
				\begin{tikzpicture}[scale=1,
						node/.style = {shape=circle, draw, inner sep=0pt, minimum size=0.25cm},
						dashednode/.style = {shape=circle, draw=blue, dashed, inner sep=0pt, minimum size=0.25cm},
						textnode/.style = {shape=circle, draw, inner sep=0pt, minimum size=0.4cm},
						dashedtextnode/.style = {shape=circle, draw=blue, dashed, inner sep=0pt, minimum size=0.4cm},
						box/.style = {rectangle, fill=gray!20, rounded corners, fill opacity=1, inner sep=1pt}]
                    \node[node] (t) at (0, 0) {}; \node[] at ($(t) + (0,0.33)$) {$v_\ell$};
                    \node[node] (f) at (1, 0) {}; \node[] at ($(f) + (0,0.33)$) {$v_{\overline{\ell}}$};
                    \path[-] (t) edge node[right] {} (f);
                \end{tikzpicture}
	        }
            \caption{
                Variable Gadget representing literals $\ell, \overline{\ell} \in L$.
                The corresponding mappings are $f_{L, V}: \ell \mapsto \{v_\ell, v_{\overline{\ell}}\}$ and $f_{L, E}: \ell \mapsto \{\{v_\ell, v_{\overline{\ell}}\}\}$.
            }
            \label{fig:vc:variable_gadget}
        \end{subfigure}
        \hfill
        \begin{subfigure}[b]{0.45\textwidth}
            \centering
		    \scalebox{1}{
				\begin{tikzpicture}[scale=1,
						node/.style = {shape=circle, draw, inner sep=0pt, minimum size=0.25cm},
						dashednode/.style = {shape=circle, draw=blue, dashed, inner sep=0pt, minimum size=0.25cm},
						textnode/.style = {shape=circle, draw, inner sep=0pt, minimum size=0.4cm},
						dashedtextnode/.style = {shape=circle, draw=blue, dashed, inner sep=0pt, minimum size=0.4cm},
						box/.style = {rectangle, fill=gray!20, rounded corners, fill opacity=1, inner sep=1pt}]
	                \node[node] (1) at (0, 0) {}; \node[] at ($(1) + (0.45,0.2)$) {$v_{\ell_1, c}$}; 
	                \node[node] (2) at (-0.75, -0.75) {}; \node[] at ($(2) + (0,-0.33)$) {$v_{\ell_2, c}$}; 
	                \node[node] (3) at (0.75, -0.75) {}; \node[] at ($(3) + (0,-0.33)$) {$v_{\ell_3, c}$};
	                \path[-] (1) edge node[right] {} (2);
	                \path[-] (1) edge node[right] {} (3);
	                \path[-] (2) edge node[right] {} (3);
	                \node[dashednode] (11) at (0, 0.75) {}; \node[] at ($(11) + (0,0.33)$) {$v_{\ell_1}$};
	                \node[dashednode] (22) at (-1, 0.75) {}; \node[] at ($(22) + (0,0.33)$)  {$v_{\ell_2}$};
	                \node[dashednode] (33) at (1, 0.75) {}; \node[] at ($(33) + (0,0.33)$)  {$v_{\ell_3}$};
	                \path[-, out=90, in=270] (1) edge node[right] {} (11);
	                \path[-, out=90, in=270] (2) edge node[right] {} (22);
	                \path[-, out=90, in=270] (3) edge node[right] {} (33);
	            \end{tikzpicture}
		    }
            \caption{
                Clause gadget representing clause $c \in C$.
                The corresponding mappings are defined by
                $
                    f_{C, V}: \ell \mapsto \{v_{\ell_1, c}, v_{\ell_2, c}, v_{\ell_3, c}\}
                $
                and by
                $
                    f_{C, E}: \ell \mapsto \{\{v_{\ell_1, c}, v_{\ell_2, c}\}, \{v_{\ell_2, c}, v_{\ell_3, c}\}, \{v_{\ell_3, c}, v_{\ell_1, c}\},\\ \{v_{\ell_1}, v_{\ell_1, c}\}, \{v_{\ell_2}, v_{\ell_2, c}\}, \{v_{\ell_3}, v_{\ell_3, c}\}\}.
                $
            }
            \label{fig:vc:clausegadget}
        \end{subfigure}
        \centering
        \caption{\centering Gadgets for universe and relations for the \sat-\vc{} reduction}
        \label{fig:sat-vc:reduction:gadgets:coarse}
    \end{figure}
    Observe that the gadgets only combine the more fine-grained relations and the overall reduction stays the same.
    That is, the reduction is overall the same for both views and can be found in \Cref{fig:sat-vc:reduction} as well.
    The existence of these variable gadgets and a clause gadgets also shows the strong modularity of this reduction:
    One can easily remove the variable gadget and all clauses gadgets containing that variable or removing just one clause gadget.
    The resulting graph is the correct reduction graph of the corresponding \sats{} instance.

    The solution size function for each gadget can be defined as follows.
    A solution includes one vertex for each literal in the solution of the \sats{} instance: $v_{\ell_i}$ is included if and only if $x_i$ is assigned to true, and $v_{\overline \ell_i}$ is included if and only if $x_i$ is assigned to false.
    Additionally to satisfy the clause gadgets, two of the vertices of the 3-clique need to be included.
    If a vertex $v_\ell$ is in the solution, then the edge $\{v_\ell, v_{\ell,c}\}$ is covered.
    Therefore, the vertices $v_{\ell',c}$ and $v_{\ell'',c}$ can be taken into the solution in order to cover the 3-clique and the incident edges.
    If the clause is not satisfied, then non of $v_\ell$, $v_{\ell'}$, and $v_{\ell''}$ are in the solution and all three vertices $v_{\ell,c}$,$v_{\ell',c}$, and $v_{\ell'',c}$ need to be included to cover all edges invalidating the solution.
    Thus, $size_f(L, C) = |L|/2 + 2|C|$.
    This also fulfills the necessary conditions by the modularity of the gadget reduction.
    Correspondingly, the feasible solutions of the scenarios are defined with the help of the solution size function.
    Concretely, a solution to the \textsc{Vertex Cover} instance is feasible if the vertex cover has size at most $size_f(L,C)$.
\end{proof}

\begin{lemma}
    \vc{} is strongly modular universe gadget reducible to \ds{} such that the solution properties hold and a solution size function for this reduction exists.
\end{lemma}
\begin{proof}
	There is a folklore reduction that is a universe gadget reduction.
	For every edge $\{u, v\}$ between vertices $u, v \in V$, the reduction adds two vertices $uv^1$ and $uv^2$ together with edges $\{uv^i, u\}, \{uv^i, v\}$, $i \in \{1,2\}$.
	The universe elements of both problems are the vertices $V$ and the relations are the edges $E$.
	Thus, there are vertex gadgets, see \Cref{fig:dsvc:vertex_gadget}, defined by
	$$
		f_{V, V}, f_{V, E}
	$$
	and the edge gadgets, see \Cref{fig:dsvc:edge_gadget},
	$$
		f_{E, V}, f_{E, E}.
	$$

	\begin{figure}[!ht]
		\begin{subfigure}[b]{0.45\textwidth}
			\centering
			\scalebox{1}{
				\begin{tikzpicture}[scale=1,
						node/.style = {shape=circle, draw, inner sep=0pt, minimum size=0.25cm},
						dashednode/.style = {shape=circle, draw=blue, dashed, inner sep=0pt, minimum size=0.25cm},
						textnode/.style = {shape=circle, draw, inner sep=0pt, minimum size=0.4cm},
						dashedtextnode/.style = {shape=circle, draw=blue, dashed, inner sep=0pt, minimum size=0.4cm},
						box/.style = {rectangle, fill=gray!20, rounded corners, fill opacity=1, inner sep=1pt}]
					\node[node] (v) at (0, 0) {}; \node[] at ($(v) + (0,0.33)$) {$v$};
				\end{tikzpicture}
			}
			\caption{
				Vertex Gadget for $v \in V$.
				The corresponding mappings are $f_{V, V}: v \mapsto \{v\}$ and $f_{V, E}: \{u, v\} \mapsto \emptyset$.
			}
			\label{fig:dsvc:vertex_gadget}
		\end{subfigure}
		\hfill
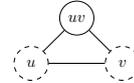
		\begin{subfigure}[b]{0.45\textwidth}
			\centering
			\scalebox{1}{
				\begin{tikzpicture}[scale=1,
						node/.style = {shape=circle, draw, inner sep=0pt, minimum size=0.25cm},
						dashednode/.style = {shape=circle, draw=blue, dashed, inner sep=0pt, minimum size=0.25cm},
						textnode/.style = {shape=circle, draw, inner sep=0pt, minimum size=0.4cm},
						dashedtextnode/.style = {shape=circle, draw=blue, dashed, inner sep=0pt, minimum size=0.4cm},
						box/.style = {rectangle, fill=gray!20, rounded corners, fill opacity=1, inner sep=1pt}]
					\node[dashednode] (u) at (0, 0) {}; \node[] at ($(u) + (0,0.33)$) {$u$};
					\node[dashednode] (v) at (2, 0) {}; \node[] at ($(v) + (0,0.33)$) {$v$};
					\node[node] (uv1) at (1, 1.33) {}; \node[] at ($(uv1) + (0,0.33)$) {$uv^1$};
					\node[node] (uv2) at (1, 0.33) {}; \node[] at ($(uv2) + (0,0.33)$) {$uv^2$};
					\path[-] (u) edge node[right] {} (v);
					\path[-] (u) edge node[right] {} (uv1);
					\path[-] (uv1) edge node[right] {} (v);
					\path[-] (u) edge node[right] {} (uv2);
					\path[-] (uv2) edge node[right] {} (v);
				\end{tikzpicture}
			}
			\caption{
				Edge Gadget for $v \in V$.
				The corresponding mappings are $f_{E, V}: \{u, v\} \mapsto \{\{uv^1\},\{uv^2\}\}$ and $f_{E, E}: \{u, v\} \mapsto \{\{u, v\}, \{u, uv^1\}, \{uv^1, v\}, \{u, uv^2\}, \{uv^2, v\}\}$.
			}
			\label{fig:dsvc:edge_gadget}
		\end{subfigure}
	\end{figure}
	
	It is easy to see that both properties of a universe gadget reduction are fulfilled.
	The gadgets are disjoint.
	Furthermore, removing a vertex (and its incident edges) results in removing the corresponding vertex gadget and edge gadgets.
    Accordingly, the resulting instance remains correct for this reduction.
    Removing only an edge results in removing the edge gadget, which is also correct.
	At last, we consider the solution.
    The solution size function remains $size_f(L, C) = |L|/2 + 2|C|$ as in the \textsc{Vertex Cover} reduction because the solution of the \textsc{Dominating Set} and \textsc{Vertex Cover} build up a one-to-one correspondence.
    That is, a vertex $v$ is part of a vertex cover if and only if $v$ is part of a dominating set.
    The feasible solutions are accordingly defined by the dominating sets of size at most $size_f(L,C)$.
\end{proof}

\begin{lemma}
    \vc{} is strongly modular universe gadget reducible to \fas{} such that the solution properties hold and a solution size function for this reduction exists.
\end{lemma}
\begin{proof}
	The reduction of Karp \cite{DBLP:conf/coco/Karp72} is a universe gadget reduction.
	\vc{} consists of vertices $V$ and edges $E$.
	\fas{} consists of vertices $V'$ and arcs $A'$.
	The reduction maps every vertex $v \in V$ to two vertices $v_0, v_1 \in V'$ and one arc $(v_0, v_1) \in A'$.
	Furthermore, we map each edge $\{u,v\} \in E$ to two arcs $(u_1, v_0),(v_0,u_1) \in A'$.
	Thus, each edge $\{u,v\} \in E$ induces a cycle of four arcs, which has to be disconnected by removing one of the arcs.
	Observe that the arc $(v_0, v_1) \in A'$ is contained in all cycles induced by incident edges $e \in E$ of $v \in V$.
	Thus, it is always favorable to remove arcs $(v_0, v_1) \in A'$ which corresponds to taking $v \in V$ into the vertex cover.
	Because of this one-to-one correspondence between an arc $(v_0, v_1)$ in \fas{} and a vertex $v$ in \vc{} in the solution, the solution size remains $size_f(L, C) = |L|/2 + 2|C|$.
    Furthermore, the one-to-one correspondence between the elements and their gadgets guarantees the modularity and that all pre-images of all gadgets are unique.
    The feasible solutions are accordingly defined by the feedback arc sets of size at most $size_f(L,C)$.
\end{proof}

\begin{lemma}
    \vc{} is strongly modular universe gadget reducible to \fvs{} such that the solution properties hold and a solution size function for this reduction exists.
\end{lemma}
\begin{proof}
	The reduction of Karp \cite{DBLP:conf/coco/Karp72} is a universe gadget reduction.
	Again, \vc{} consists of vertices $V$ and edges $E$.
	On the other hand, \fvs{} consists of vertices $V'$ and arcs $A'$.
	The reduction maps every vertex $v \in V$ to vertex $v \in V'$ and every edge $\{u, v\} \in E$ is mapped to two arcs $(u, v)$ and $(v, u)$ in $A'$.
	Because of the one-to-one correspondence in the solution between vertex $v \in V$ in \vc{} and vertex $v \in V'$ in \fvs{}, the solution size remains $size_f(L, C) = |L|/2 + 2|C|$.
	Furthermore, the correspondence between the elements and their gadgets guarantees modularity and that all pre-images of all gadgets are unique.
    The feasible solutions are accordingly defined by all feedback vertex sets of size at most $size_f(L,C)$.
\end{proof}

\begin{lemma}
    \vc{} is strongly modular universe gadget reducible to \hs{} such that the solution properties hold and a solution size function for this reduction exists.
\end{lemma}
\begin{proof}
The reduction of Karp \cite{DBLP:conf/coco/Karp72} from \vc{} to \hs{} is a universe gadget reduction.
\vc{} consists of vertices $V$ and edges $E$.
The universe of \hs{} is a set $U$ and the relations are subsets $s_i \subseteq U$ for $1 \leq i \leq r$.
Every vertex $v \in V$ is mapped to a corresponding element $v \in U$ and every edge $(u, v) \in E$ is mapped to a subset $s = \{u,v\} \subseteq U$.
By the one-to-one correspondence of vertex $v$ and element $v$ in the universe of \hs{} as well as the edge $\{u,v\}$ and the subset $s = \{u,v \}$, the gadgets are disjoint, uniquely retraceable to their origin, and the reduction is modular.
At last, we consider the solution.
The solution size function remains $size_f(L, C) = |L|/2 + 2|C|$ because the elements of \hs{} and the vertices in \vc{} build up a one-to-one correspondence.
The feasible solutions are accordingly defined by the hitting sets of size at most $size_f(L,C)$.
\end{proof}

\subsubsection{\textsc{Independent Set}}
\begin{lemma}
    \sat{} is strongly modular universe gadget reducible to \is{} such that the solution properties hold and a solution size function for this reduction exists.
\end{lemma}
\begin{proof}
   For \is{}, we reuse the reduction from \sat{} to \vc{} by Garey and Johnson \cite{DBLP:books/fm/GareyJ79}.
   For \sat{}, we use the literals as universe elements and the relations $R_{\ell, \overline{\ell}}$, which relates a literal and its negation, $R_{\overline{\ell}, c}$, which relates a clause with the negation of the its literals, $R_{\ell, \overline{\ell}, c}$, which relates the literal and its negation with the clauses the literal is in.
   \is{}, on the other side, consists of vertices $V$ and edges $E$.
   This results in the mappings for the variable gadget, see \Cref{fig:is:variable_gadget},
   $$
       f_{L, V}, f_{L, E}, f_{R_{\ell,\overline{\ell}}, V}, f_{R_{\ell,\overline{\ell}}, E},
   $$
   and the clause gadget, see \Cref{fig:is:clause_gadget},
   $$
       f_{R_{\overline{\ell}, c}, V}, f_{R_{\overline{\ell}, c}, E}, f_{R_{\ell, \ell', c}, V}, f_{R_{\ell, \ell', c}, E}.
   $$

	\begin{figure}[!ht]
		\centering
		\begin{subfigure}[b]{0.45\textwidth}
			\centering
			\scalebox{1}{
				\begin{tikzpicture}[scale=1,
						node/.style = {shape=circle, draw, inner sep=0pt, minimum size=0.25cm},
						dashednode/.style = {shape=circle, draw=blue, dashed, inner sep=0pt, minimum size=0.25cm},
						textnode/.style = {shape=circle, draw, inner sep=0pt, minimum size=0.4cm},
						dashedtextnode/.style = {shape=circle, draw=blue, dashed, inner sep=0pt, minimum size=0.4cm},
						box/.style = {rectangle, fill=gray!20, rounded corners, fill opacity=1, inner sep=1pt}]
                    \node[node] (t) at (0, 0) {}; \node[] at ($(t) + (0,0.33)$) {$v_\ell$};
                    \node[node] (f) at (1, 0) {}; \node[] at ($(f) + (0,0.33)$) {$v_{\overline{\ell}}$};
                    \path[-] (t) edge node[right] {} (f);
                \end{tikzpicture}
	        }
			\caption{
				Variable Gadget representing literals $\ell, \overline{\ell} \in L$.
				The corresponding mappings are $f_{L, V}: \ell \mapsto \{v_\ell, v_{\overline{\ell}}\}$ and $f_{L, E}: \ell \mapsto \{\{v_\ell, v_{\overline{\ell}}\}\}$.
			}
			\label{fig:is:variable_gadget}
		\end{subfigure}
		\hfill
		\begin{subfigure}[b]{0.45\textwidth}
			\centering
			\scalebox{1}{
				\begin{tikzpicture}[scale=1,
						node/.style = {shape=circle, draw, inner sep=0pt, minimum size=0.25cm},
						dashednode/.style = {shape=circle, draw=blue, dashed, inner sep=0pt, minimum size=0.25cm},
						textnode/.style = {shape=circle, draw, inner sep=0pt, minimum size=0.4cm},
						dashedtextnode/.style = {shape=circle, draw=blue, dashed, inner sep=0pt, minimum size=0.4cm},
						box/.style = {rectangle, fill=gray!20, rounded corners, fill opacity=1, inner sep=1pt}]
	                \node[node] (1) at (0, 0) {}; \node[] at ($(1) + (0.45,0.2)$) {$v_{\ell_1, c}$}; 
	                \node[node] (2) at (-0.75, -0.75) {}; \node[] at ($(2) + (0,-0.33)$) {$v_{\ell_2, c}$}; 
	                \node[node] (3) at (0.75, -0.75) {}; \node[] at ($(3) + (0,-0.33)$) {$v_{\ell_3, c}$};
	                \path[-] (1) edge node[right] {} (2);
	                \path[-] (1) edge node[right] {} (3);
	                \path[-] (2) edge node[right] {} (3);
	                \node[dashednode] (11) at (0, 0.75) {}; \node[] at ($(11) + (0,0.33)$) {$v_{\overline\ell_1}$};
	                \node[dashednode] (22) at (-1, 0.75) {}; \node[] at ($(22) + (0,0.33)$)  {$v_{\overline\ell_2}$};
	                \node[dashednode] (33) at (1, 0.75) {}; \node[] at ($(33) + (0,0.33)$)  {$v_{\overline\ell_3}$};
	                \path[-, out=90, in=270] (1) edge node[right] {} (11);
	                \path[-, out=90, in=270] (2) edge node[right] {} (22);
	                \path[-, out=90, in=270] (3) edge node[right] {} (33);
	            \end{tikzpicture}
		    }
			\caption{
                Clause gadget representing clause $c \in C$.
                The corresponding mappings are defined by $f_{C, V}: \ell \mapsto \{v_{\ell_1, c}, v_{\ell_2, c}, v_{\ell_3, c}\}$ and by $f_{C, E}: \ell \mapsto \{\{v_{\ell_1, c}, v_{\ell_2, c}\}, \{v_{\ell_2, c}, v_{\ell_3, c}\}, \{v_{\ell_3, c}, v_{\ell_1, c}\},\\ \{v_{\overline{\ell_1}}, v_{\ell_1, c}\}, \{v_{\overline{\ell_2}}, v_{\ell_2, c}\}, \{v_{\overline{\ell_3}}, v_{\ell_3, c}\}\}$.
			}
			\label{fig:is:clause_gadget}
		\end{subfigure}
		\centering
		\caption{Gadgets for universe and relations for the reduction from \sat{} to \is{}}
		\label{fig:sat-is:reductiongadgetscoarse}
	\end{figure}

	Analogously to the \vc{} reduction, this reduction is a universe gadget reduction.
	Furthermore, the solution size function includes one vertex for each variable gadget and one vertex for each clause gadget.
    A vertex $v_{\ell_i}$ is included in the solution if and only if $x_i$ is assigned to true,
    and $v_{\overline \ell_i}$ is included in the solution if and only if $x_i$ is assigned to false.
    Because $v_\ell$ is in the solution and thus not $v_{\overline \ell}$, the vertex $v_{\ell, c}$ can be included simulating the satisfaction of the clause.
    If a clause $c \in C$ is not satisfied, then all vertices $v_{\overline \ell}$, $v_{\overline \ell'}$, and $v_{\overline \ell''}$ for $\ell,\ell', \ell'' \in c$ are in the independent set such that non of $v_{\ell,c}$, $v_{\ell',c}$, $v_{\ell'',c}$ can be taken into the solution.
	Thus, $size_f(L,C) = |L|/2 + |C|$.
	With the same arguments as for the \vc{} reduction, the solution size function is modular.
    We define the feasible solutions by the independent sets of size at least $size_f(L,C) = |L|/2 + |C|$.
\end{proof}

\begin{lemma}
    \is{} is strongly modular universe gadget reducible to \cl{} such that the solution properties hold and a solution size function for this reduction exists.
\end{lemma}
\begin{proof}
	For \cl{}, we reuse the duality between \vc{}, \is{}, and \cl{} as described by Garey and Johnson \cite{DBLP:books/fm/GareyJ79}.
	The problem \is{} consists of a graph with vertices $V$ and edges $E$.
	On the other hand, we define \cl{} with vertices $V'$ as universe but a different relation $\overline E \subseteq V' \times V'$, the set of non-edges.
	This definition of \cl{} allows us to use the equivalence as universe gadget reduction.

	For the reduction, we map every vertex $v \in V$ to the vertex $v' \in V'$ and we map every edge $e \in E$ to a non-edge $\overline e \in \overline E$.
	Thus, we have a one-to-one correspondence between the vertices and the edges and non-edges.
	This one-to-one correspondence also holds for the solution.
	That is, every solution of one problem is also a solution to the other problem.
	By this one-to-one correspondence, the modularity, the pre-image uniqueness and the solution size of $size_f(L,C) = |L|/2 + |C|$ remains.
    The feasible solutions are accordingly defined by cliques of size at least $size_f(L,C)$.
\end{proof}

\subsubsection{\textsc{Subset Sum}}
\begin{lemma}
    \sat{} is weakly modular universe gadget reducible to \sss{} such that the solution properties hold and a solution size function for this reduction exists.
\end{lemma}
\begin{proof}
    We use a modification of the reduction by Sipser \cite{DBLP:books/daglib/0086373} from \sat{} to \sss{}.
    For \sat{}, we use the literals as universe elements and the relations $R_{\ell, \overline{\ell}}$, which relates a literal and its negation, the clause relation $R_c$, which is a unary relation on the clauses, and $R_{\ell, c}$, which relates a clause with the negation of its literals.
    
    \sss{}, on the other side, consists of binary numbers of $\{0, 1\}^{t}$.
    For the sake of simplicity, the reduction description uses non-binary numbers.
    Numbers that are bigger than one are easily translatable in corresponding binary numbers with an offset such that a possible carry has no influence.
    This results in the mappings for the variable gadget, see \Cref{fig:ss:literal_gadget},
    $$
        f_{L, \{0, 1\}^{t}},
    $$
    and the clause gadget, see both \Cref{fig:ss:clause_gadget,fig:ss:literal_clause_gadget},
    $$
        f_{R_{c}, \{0, 1\}^{t}}, f_{R_{\ell, c}, \{0, 1\}^{t}}.
    $$
    \begin{figure}[!ht]
        \centering
        \begin{tabular}{c|c|c|c|c||c|c|c|c||c|c|c|c}
            $\ell_1$ & $\overline{\ell_1}$ & $\ldots$ & $\ell_n$ & $\overline{\ell_n}$ & $x_1$ & $x_2$ & $\ldots$ & $x_n$ & $c_1$ & $c_2$ & $\ldots$ & $c_m$\\
            \hline
            1 & 0 & \ldots & 0 & 0 & 1 & 0 & \ldots & 0 & 0 & 0 & $\ldots$ & 0\\
            \hline
            0 & 1 & \ldots & 0 & 0 & 1 & 0 & \ldots & 0 & 0 & 0 & $\ldots$ & 0\\
        \end{tabular}
        \caption{
            Variable Gadget representing literals $\ell, \overline{\ell} \in L$ (here for $\ell_1$ and $\overline{\ell}$).
        }
        \label{fig:ss:literal_gadget}
    \end{figure}
    \begin{figure}[!ht]
        \centering
        \begin{tabular}{c|c|c|c|c||c|c|c|c||c|c|c|c}
            $\ell_1$ & $\overline{\ell_1}$ & $\ldots$ & $\ell_n$ & $\overline{\ell_n}$ & $x_1$ & $x_2$ & $\ldots$ & $x_n$ & $c_1$ & $c_2$ & $\ldots$ & $c_m$\\
            \hline
            0 & 0 & \ldots & 0 & 0 & 0 & 0 & \ldots & 0 & 11 & 0 & $\ldots$ & 0\\
            \hline
            0 & 0 & \ldots & 0 & 0 & 0 & 0 & \ldots & 0 & 12 & 0 & $\ldots$ & 0\\
            \hline
            0 & 0 & \ldots & 0 & 0 & 0 & 0 & \ldots & 0 & 13 & 0 & $\ldots$ & 0\\
        \end{tabular}
        \caption{
            Clause Gadget for $c \in C$ (here for $c_1$).
        }
        \label{fig:ss:clause_gadget}
    \end{figure}
    \begin{figure}[!ht]
        \centering
        \begin{tabular}{c|c|c|c|c||c|c|c|c||c|c|c|c}
            $\ell_1$ & $\overline{\ell_1}$ & $\ldots$ & $\ell_n$ & $\overline{\ell_n}$ & $x_1$ & $x_2$ & $\ldots$ & $x_n$ & $c_1$ & $c_2$ & $\ldots$ & $c_m$\\
            \hline
            1 & 0 & \ldots & 0 & 0 & 0 & 0 & \ldots & 0 & 1 & 0 & $\ldots$ & 0\\
        \end{tabular}
        \caption{
            Literal Clause Gadget for $\ell \in c \in C$ (here for $\ell_1$ with $\ell_1 \in c_1$ and $\ell_1 \notin c_2, c_m$).
        }
        \label{fig:ss:literal_clause_gadget}
    \end{figure}

	The target sum in \sss{} plays a crucial role to simulate the satisfaction of the clause correctly.
	In \Cref{fig:ss:k}, the target sum is depicted.
    In a solution, the row corresponding to $\overline \ell_i$ is taken into the solution if and only if variable $x_i$ is assigned true.
    This enables the solution to include the literal clause gadget for $\ell_i$.
    Otherwise, the row corresponding to $\ell_i$ is part of the solution and the literal clause gadget for $\overline \ell_i$.
    Thus all clauses that include $\ell_i$ have at least 1 in their row.
    After this assignment, the suitable row from the clause gadget is included such that the sum of the clause column reaches 14.
    If a clause is not satisfied, the corresponding column cannot reach 14.
    
    \begin{figure}[!ht]
        \centering
        \begin{tabular}{c||c|c|c|c|c||c|c|c|c||c|c|c|c}
            & $\ell_1$ & $\overline{\ell_1}$ & $\ldots$ &  $\ell_n$ & $\overline{\ell_n}$ & $x_1$ & $x_2$ & $\ldots$ & $x_n$ & $c_1$ & $c_2$ & $\ldots$ & $c_m$\\
            \hline
            $\sum$ & 1 & 1 & $\ldots$ & 1 & 1 & 1 & 1 & $\ldots$ & 1 & 14 & 14 & $\ldots$ & 14\\
        \end{tabular}
        \caption{
            The target value $k$ of the sum.
        }
        \label{fig:ss:k}
    \end{figure}

	The modularity of this problem is weak.
    The removal gadget of a clause is the addition of a number that simulates the satisfaction of that clause.
	This gadget is depicted in \Cref{fig:ss:removal_gadget:clause}.
    Note that we leave the numbers from \Cref{fig:ss:clause_gadget} available.
    We can use this clause gadget to also construct a literal removal gadget for $\ell$ or $\overline \ell$.
	This gadget simulates the fulfillment of the clauses that contain $\ell$ or $\overline \ell$.
    Consequently, we add the clause removal gadget for clauses that contain $\ell$ or $\overline \ell$.
    Additionally, we leave the literal gadget within the instance.
    Then together with the clause gadget from \Cref{fig:ss:clause_gadget}, the sum of the clause column containing $\ell$ or $\overline \ell$ can be set to $14$, depending on the other variables.
    The partial solution in the base scenario is extendable to a full solution in the uncertainty scenarios because the clause removal gadget only fixes the solution on the literals that are part of that clause.

\begin{figure}[!ht]
        \centering
        \begin{tabular}{c|c|c|c|c||c|c|c|c||c|c|c|c}
            $\ell_1$ & $\overline{\ell_1}$ & $\ldots$ & $\ell_n$ & $\overline{\ell_n}$ & $x_1$ & $x_2$ & $\ldots$ & $x_n$ & $c_1$ & $c_2$ & $\ldots$ & $c_m$\\
	\hline
            0 & 0 & \ldots & 0 & 0 & 0 & 0 & \ldots & 0 & 14 & 0 & $\ldots$ & 0\\
        \end{tabular}
        \caption{
            The Removal Gadget for $c_1 \in C$.
        }
        \label{fig:ss:removal_gadget:clause}
    \end{figure}

	At last, we describe the solution size function over the numbers.
	Due to weak modularity, we use a reduction $f$ from an instance $(L',C')$ with $L \subseteq L'$ and $C \subseteq C'$ such that
	$$
		size_f(L,C) = |L'| + |C'|.
	$$
	Overall, we include one number for each literal and one number for each clause.
	If a variable is removed, we include the removal gadgets for the clauses containing this variable and leave the gadgets for the literals and the literal clause relation in the instance.
	The removal of a clause (without also removing a variable) does not change the number of elements in the solution because $11, 12, 13$, or $14$ still need to be added such that the clause column adds up to $14$.
    Accordingly the solution size of a variable gadget is the same as for the variable removal gadget.
\end{proof}

\begin{lemma}
    \sss{} is strongly modular universe gadget reducible to \ks{} such that the solution properties hold and a solution size function for this reduction exists.
\end{lemma}
\begin{proof}
	The reduction from \sss{} to \ks{} is by generalization.
	The numbers in \sss{} are mapped to objects of the weight and price corresponding to the value of the number.
	By setting the knapsack capacity and the price threshold to the target sum of \sss{}, the reduction is complete.

	Overall, this is a one-to-one correspondence between all combinatorial elements and the solutions.
	Thus, the modularity, the pre-image uniqueness, solution properties, and the solution size function of \sss{} directly applies to \ks{} as well.
\end{proof}

\begin{lemma}
    \sss{} is strongly modular universe gadget reducible to \pa{} such that the solution properties hold and a solution size function for this reduction exists.
\end{lemma}
\begin{proof}
	The reduction from \sss{} to \pa{} by Karp \cite{DBLP:conf/coco/Karp72} is a universe gadget reduction.
	The numbers $A$ in \sss{} are transferred to the \pa{} instance and remain unchanged.
	Furthermore, let $k$ be the target sum of \sss{}, then $k+1$ and $1 - k + \sum_{a \in A} a$ are added to the \pa{} instance as well.
	These two numbers build up the constant gadget.
    W.l.o.g. we assume that $1-k+\sum_{a \in A} a$ is in the first set of the partition.

	Overall, the numbers from \sss{} and \pa{} are one-to-one correspondent.
    Thus modularity and pre-image uniqueness hold.
    Furthermore, note that  the solution of \sss{} is the same as the set of the partition that additionally includes the element $1-k+\sum_{a \in A} a$.
    Thus, we also have a one-to-one correspondence between the elements in the solutions.
\end{proof}

\begin{lemma}
    \pa{} is strongly modular universe gadget reducible to \tms{} such that the solution properties hold and a solution size function for this reduction exists.
\end{lemma}
\begin{proof}
	\pa{} is a special case of \tms{}.
	By interpreting the numbers in the \pa{} instance to be the job times in \tms{} and by interpreting the sets of the partition as two identical machines, we have a one-to-one correspondence between the combinatorial elements and the solutions as well.
    Thus, pre-image uniqueness, modularity, and solution properties hold and the solution size function remains the same.
\end{proof}

\subsubsection{\textsc{Hamiltonian Path}}
\begin{lemma}
    \sat{} is weakly modular universe gadget reducible to \dhp{} such that the solution properties hold and a solution size function for this reduction exists.
\end{lemma}
\begin{proof}
    A modification of the reduction by Arora and Barak \cite{DBLP:books/daglib/0023084} is a universe gadget reduction.
    For \sat{}, we use the literals as universe elements and the relations $R_{\ell, \overline{\ell}}$, which relates a literal and its negation, and $R_{\ell, c}$, which relates a clause with the negation of its literals.
    
    \hc{}, on the other side, consists of vertices $V$ and arcs $A$.
    This results in the mappings for the variable gadget, see \Cref{fig:dhc:variable_gadget},
    $$
        f_{L, V}, f_{L, A}, f_{R_{\ell,\overline{\ell}}, V}, f_{R_{\ell,\overline{\ell}}, A},
    $$
    and the clause gadget, see \Cref{fig:dhc:clause_gadget} and \Cref{fig:dhc:literal-clause_gadget},
    $$
        f_{R_{\ell, c}, V}, f_{R_{\ell, c}, A}.
    $$
    In order to connect the variable gadgets, we also need a constant gadget defined by mappings $f_{const, V}$ and $f_{const, E}$, see \Cref{fig:dhc:constant_gadget} .
    \begin{figure}[!ht]\centering
	    \scalebox{1}{
			\begin{tikzpicture}[scale=1,
						node/.style = {shape=circle, draw, inner sep=0pt, minimum size=0.25cm},
						dashednode/.style = {shape=circle, draw=blue, dashed, inner sep=0pt, minimum size=0.25cm},
						textnode/.style = {shape=circle, draw, inner sep=0pt, minimum size=0.4cm},
						dashedtextnode/.style = {shape=circle, draw=blue, dashed, inner sep=0pt, minimum size=0.4cm},
						box/.style = {rectangle, fill=gray!20, rounded corners, fill opacity=1, inner sep=1pt}]
	            \node[node] (s) at (0, 0) {}; \node[] at ($(s) + (0,0.33)$) {$s$};
	            \node[node] (t) at (0, -2.5) {}; \node[] at ($(t) - (0,0.33)$) {$t$};
	            \node[dashednode] (x11) at (-1, -0.5) {}; \node[] at ($(x11) + (0,0.33)$) {\textcolor{blue}{$x_1$}};
	            \node[dashednode] (x12) at (1, -0.5) {}; \node[] at ($(x12) + (0,0.33)$) {\textcolor{blue}{$x'_1$}};
	            \node[dashednode] (xn1) at (-1, -1.5) {}; \node[] at ($(xn1) + (0,0.33)$) {\textcolor{blue}{$x_{n}$}};
	            \node[dashednode] (xn2) at (1, -1.5) {}; \node[] at ($(xn2) + (0,0.33)$) {\textcolor{blue}{$x'_{n}$}};
	            \node[dashednode] (xnt) at (0, -2) {};
	            \path[->,dashed,blue] (s) edge node[right] {} (x11);
	            \path[->,dashed,blue] (s) edge node[right] {} (x12);
	            \path[<-,dashed,blue] (xnt) edge node[right] {} (xn1);
	            \path[<-,dashed,blue] (xnt) edge node[right] {} (xn2);
	            \path[->] (xnt) edge node[right] {} (t);
	        \end{tikzpicture}
	    }
        \caption{
            Constant Gadget for the reduction.
        }
        \label{fig:dhc:constant_gadget}
    \end{figure}
    \hfill
    \begin{figure}[!ht]
        \centering
		\scalebox{1}{
			\begin{tikzpicture}[scale=1,
						node/.style = {shape=circle, draw, inner sep=0pt, minimum size=0.25cm},
						dashednode/.style = {shape=circle, draw=blue, dashed, inner sep=0pt, minimum size=0.25cm},
						textnode/.style = {shape=circle, draw, inner sep=0pt, minimum size=0.4cm},
						dashedtextnode/.style = {shape=circle, draw=blue, dashed, inner sep=0pt, minimum size=0.4cm},
						box/.style = {rectangle, fill=gray!20, rounded corners, fill opacity=1, inner sep=1pt}]
            \node[node] (xstart) at (0, 0) {}; \node[] at ($(xstart) + (0,0.33)$) {$x_i$};
            \node[node] (x1) at (2, 0) {}; \node[] at ($(x1) + (0,0.33)$) {$x^1_i$};
            \node[node] (x2) at (4, 0) {}; \node[] at ($(x2) + (0,0.33)$) {$x^2_i$};
            \node[node] (xend) at (9, 0) {}; \node[] at ($(xend) + (0,0.33)$) {$x'_i$};
            \node[node] (xc) at (7, 0) {}; \node[] at ($(xc) + (0.22,0.36)$) {$x^{4|C|}_i$};
            \node[] (dots) at (5.5,0) {$\ldots$};
            \node[dashednode] (xsplit) at (4.5, 1) {};
            \node[node] (xjoin) at (4.5,-0.75) {};
            \path[->, out=30, in=150] (xstart) edge node[right] {} (x1);
            \path[<-, out=330, in=210] (xstart) edge node[right] {} (x1);
            \path[->, dashed, draw=blue, out=30, in=150] (x1) edge node[right] {} (x2);
            \path[<-, dashed, draw=blue, out=330, in=210] (x1) edge node[right] {} (x2);
            \path[->, out=30, in=150] (x2) edge node[right] {} (dots);
            \path[<-, out=330, in=210] (x2) edge node[right] {} (dots);
            \path[->, out=30, in=150] (dots) edge node[right] {} (xc);
            \path[<-, out=330, in=210] (dots) edge node[right] {} (xc);
            \path[->, out=30, in=150] (xc) edge node[right] {} (xend);
            \path[<-, out=330, in=210] (xc) edge node[right] {} (xend);
            \path[->, out=180,in=60,looseness=0.3] (xsplit) edge (xstart);
            \path[->, out=0,in=120,looseness=0.3] (xsplit) edge (xend);
            \path[->,out=270,in=150,looseness=0.25] (xstart) edge (xjoin);
            \path[->,out=270,in=30,looseness=0.25] (xend) edge (xjoin);

            \node[dashednode] (x1start) at (0, -1.0) {}; \node[] at ($(x1start) + (0,-0.33)$) {\textcolor{blue}{$x_{i+1}$}};
            \node[dashednode] (x1end) at (9, -1.0) {}; \node[] at ($(x1end) + (0,-0.33)$) {\textcolor{blue}{$x'_{i+1}$}};
	
            \path[->, dashed, draw=blue] (xjoin) edge node[right] {} (x1start);
            \path[->, dashed, draw=blue] (xjoin) edge node[right] {} (x1end);
        \end{tikzpicture}
	}
        \caption{
            Variable Gadget representing literals $\ell, \overline{\ell} \in L$.
        }
        \label{fig:dhc:variable_gadget}
    \end{figure}
    \begin{figure}[!ht]
        \centering
        \begin{subfigure}[b]{0.48\textwidth}
            \centering
            \scalebox{1}{
    			\begin{tikzpicture}[scale=1,
    						node/.style = {shape=circle, draw, inner sep=0pt, minimum size=0.25cm},
    						dashednode/.style = {shape=circle, draw=blue, dashed, inner sep=0pt, minimum size=0.25cm},
    						textnode/.style = {shape=circle, draw, inner sep=0pt, minimum size=0.4cm},
    						dashedtextnode/.style = {shape=circle, draw=blue, dashed, inner sep=0pt, minimum size=0.4cm},
    						box/.style = {rectangle, fill=gray!20, rounded corners, fill opacity=1, inner sep=1pt}]
                    \node[dashednode] (x1) at (-1,-1) {}; \node[] at ($(x1) + (-0,0.33)$) {\textcolor{blue}{$x^{4j-2}_i = x^{4j-1}_i$}};
                \end{tikzpicture}
	       }
            \caption{
                If the literals $\ell$ and $\overline \ell$ are not in the $j$-th clause $c_j \in C$, we merge the vertices $x^{4j-2}_i$ and $x^{4j-1}_i$.
            }
            \label{fig:dhc:clause_gadget}
        \end{subfigure}
        ~
        \begin{subfigure}[b]{0.48\textwidth}
            \centering
            \scalebox{1}{
    			\begin{tikzpicture}[scale=1,
    						node/.style = {shape=circle, draw, inner sep=0pt, minimum size=0.25cm},
    						dashednode/.style = {shape=circle, draw=blue, dashed, inner sep=0pt, minimum size=0.25cm},
    						textnode/.style = {shape=circle, draw, inner sep=0pt, minimum size=0.4cm},
    						dashedtextnode/.style = {shape=circle, draw=blue, dashed, inner sep=0pt, minimum size=0.4cm},
    						box/.style = {rectangle, fill=gray!20, rounded corners, fill opacity=1, inner sep=1pt}]
                    \node[node] (c) at (0, 0) {}; \node[] at ($(c) + (0,0.33)$) {$c_j$};
                    \node[dashednode] (x1) at (-1,-1) {}; \node[] at ($(x1) + (-0.2,0.33)$) {\textcolor{blue}{$x^{4j-2}_i$}};
                    \node[dashednode] (x2) at (1, -1) {}; \node[] at ($(x2) + (0.5,0.33)$) {\textcolor{blue}{$x^{4j-1}_i$}};
                    \path[->] (x1) edge node[right] {} (c);
                    \path[->] (c) edge node[right] {} (x2);
                    \path[<-, out=30, in=150] (x1) edge node[right] {} (x2);
                    \path[<-, out=210, in=330] (x2) edge node[right] {} (x1);
                \end{tikzpicture}
    	    }
            \caption{
                Clause Gadget for $\ell \in c \in C$.
            }
            \label{fig:dhc:literal-clause_gadget}
        \end{subfigure}
        \caption{
                Clause Gadget for $c \in C$.
            }
    \end{figure}

    If we include the path from $x_i$ to $x_i'$, we simulate that the variable is assigned true.
    Vice versa, if the path from $x'_i$ to $x_i$ is included, the variable $x_i$ is simulated to be false.
    Therefore, it is possible to include the vertices of clauses satisfied by the assignment into the Hamiltonian cycle.

	This reduction is only weakly modular because removing a variable $x_i$ results in a disconnected graph.
    For this, we can employ a clause removal gadget, which can be found in \Cref{fig:dhc:clause_removal_gadget}.
    Then for a variable removable gadget of variable $x_i$, we leave the variable gadget in the instance, while introducing the clause removal gadget for all clauses that contain the variable $x_i$.
    A solution in the base scenario is extendable to a full solution because the clause removal gadget only fixes the solution on the literals that are part of the clause.
    \begin{figure}[!ht]
        \centering
        \scalebox{1}{
			\begin{tikzpicture}[scale=1,
						node/.style = {shape=circle, draw, inner sep=0pt, minimum size=0.25cm},
						dashednode/.style = {shape=circle, draw=blue, dashed, inner sep=0pt, minimum size=0.25cm},
						textnode/.style = {shape=circle, draw, inner sep=0pt, minimum size=0.4cm},
						dashedtextnode/.style = {shape=circle, draw=blue, dashed, inner sep=0pt, minimum size=0.4cm},
						box/.style = {rectangle, fill=gray!20, rounded corners, fill opacity=1, inner sep=1pt}]
            \node[node] (c) at (0, 0) {}; \node[] at ($(c) + (0,0.33)$) {$c_j$};
            \node[dashednode] (x1) at (-1,-1) {}; \node[] at ($(x1) + (-0.5,0.33)$) {\textcolor{blue}{$x^{4j-2}_i$}};
            \node[dashednode] (x2) at (1, -1) {}; \node[] at ($(x2) + (0.5,0.33)$) {\textcolor{blue}{$x^{4j-1}_i$}};
            \path[->,out=90,in=165] (x1) edge node[right] {} (c);
            \path[->,out=345,in=120] (c) edge node[right] {} (x2);
            \path[->,out=90,in=15] (x2) edge node[right] {} (c);
            \path[->,out=195,in=60] (c) edge node[right] {} (x1);
            \path[<-, out=30, in=150] (x1) edge node[right] {} (x2);
            \path[<-, out=210, in=330] (x2) edge node[right] {} (x1);
        \end{tikzpicture}
	}
        \caption{
            Clause Removal Gadget for clause $c \in C$.
        }
        \label{fig:dhc:clause_removal_gadget}
    \end{figure}
	
	At last, we describe the solution size function over the arcs.
	Due to weak modularity, we use a reduction $f$ from an instance $(L',C')$ with $L \subseteq L'$ and $C \subseteq C'$ such that
	$$
		size_f(L,C) = 1 + (1+3 \cdot |L'|/2) \cdot |C'| + \sum_{c \in C'} (1+|c|).
	$$
	Overall, we include one arc for the constant gadget.
	Furthermore for each variable, we include $3 |C'|$ arcs.
    For all clauses $c \in C$ of size $|c|$, we need to include $1+|c|$ arcs.
    The size of a clause gadget and clause removal gadget are the same.
    Thus we do not have to change the solution size function in this regard.
\end{proof}

\begin{lemma}
    \dhp{} is strongly modular universe gadget reducible to \dhc{} such that the solution properties hold and a solution size function for this reduction exists.
\end{lemma}
\begin{proof}
	Adding an arc from $t$ to $s$ as constant gadget closes the cycle.
	This reduction is a universe gadget reduction because the combinatorial elements are mapped one-to-one such that the solutions are one-to-one translatable as well.
	This directly proves the pre-image uniqueness, the modularity, and the solution properties.
    The solution size function includes an additional term of one for the arcs $(t,s)$, to be correct.
\end{proof}

\begin{lemma}\label{reduction:dhpuhp}
    \dhp{} is strongly modular universe gadget reducible to \uhp{} such that the solution properties hold and a solution size function for this reduction exists.
\end{lemma}
\begin{proof}
	Karp's redution \cite{DBLP:conf/coco/Karp72} is a universe gadget reduction.
	It triples the vertices and connects the triplets as depicted in \Cref{fig:dhpuhp:vertex_gadget}.
	Furthermore, each arc $(u,v)$ in the graph is mapped to an edge $\{u^{out}, v^{in}\} \in E'$.

    \begin{figure}[!ht]
        \centering
        \scalebox{1}{
			\begin{tikzpicture}[scale=1,
						node/.style = {shape=circle, draw, inner sep=0pt, minimum size=0.25cm},
						dashednode/.style = {shape=circle, draw=blue, dashed, inner sep=0pt, minimum size=0.25cm},
						textnode/.style = {shape=circle, draw, inner sep=0pt, minimum size=0.4cm},
						dashedtextnode/.style = {shape=circle, draw=blue, dashed, inner sep=0pt, minimum size=0.4cm},
						box/.style = {rectangle, fill=gray!20, rounded corners, fill opacity=1, inner sep=1pt}]
                \node[node] (1) at (0, 0) {}; \node[] at ($(1) + (0,0.33)$) {$v^{in}$};
                \node[node] (2) at (1, 0) {}; \node[] at ($(2) + (0,0.33)$) {$v^{m}$};
                \node[node] (3) at (2, 0) {}; \node[] at ($(3) + (0,0.33)$) {$v^{out}$};                   
                \path[-] (1) edge node[right] {} (2);
                \path[-] (2) edge node[right] {} (3);
            \end{tikzpicture}
        }
        \caption{
            The vertex gadget for the reduction from \dhp{} to \uhp{}.
        }
        \label{fig:dhpuhp:vertex_gadget}
    \end{figure}
	
	The pre-image uniqueness and the modularity remain.
    The solutions on arcs $(u,v)$ and edges $\{u^{out}, v^{in}\}$ are one-to-one correspondent, while all edges $\{v^{in},v^m\}$ and $\{v^{m},v^{out}\}$ have to be in the solution.
    Thus, the solution properties still hold, while the solution size function needs to take the two edges from $v^{in}$ over $v^{m}$ to $v^{out}$ into account for every vertex.
	That is, the number of used edges in a solution of the variable is tripled.
	Overall, we get
	$$
		size_f(L,C) = 6 + 9|L'|/2 \cdot |C'| + \sum_{c \in C'} (1+3|c|).
	$$
\end{proof}

\begin{lemma}
    \dhc{} is strongly modular universe gadget reducible to \uhc{} such that the solution properties hold and a solution size function for this reduction exists.
\end{lemma}
\begin{proof}
	This reduction is completely analogous to Karp's reduction from \dhp{} to \uhp{} (\Cref{reduction:dhpuhp}).
\end{proof}

\begin{lemma}
    \uhc{} is strongly modular universe gadget reducible to \tsp{} such that the solution properties hold and a solution size function for this reduction exists.
\end{lemma}
\begin{proof}
	We consider \tsp{} to be defined over an undirected weighted graph.
	This graph does not have to be complete.
	Then, the graph $G = (V, E)$ of the \uhc{} instance can be mapped to a weighted graph $G' = (V',E',w')$,
    where $V = V'$ and $E = E'$. The weights are set to $0$ and the weight threshold is set to $0$.
	
	The reductions yields a one-to-one correspondence between the vertices and edges and thus between the solutions.
	It follows that the solution size function remains the same and the solution properties, the pre-image uniqueness as well as the modularity hold.
\end{proof}

\subsubsection{\textsc{2-Disjoint Path}}
\begin{lemma}
    \sat{} is weakly modular universe gadget reducible to \tddp{} such that the solution properties hold and a solution size function for this reduction exists.
\end{lemma}
\begin{proof}
A modification of the reduction by Fortune, Hopcroft and Wyllie \cite{DBLP:journals/tcs/FortuneHW80} is a universe gadget reduction.
This reduction is much more complex than earlier reductions such that we first explain the construction and then explain how this construction can be divided into variable and clause gadgets.

First of all, the reduction introduces a so-called switch gadget, which is visualized in \Cref{fig:reduction:3sat-directed-two-disjoint-path:switch}.
\begin{figure}
\centering
\scalebox{1}{
			\begin{tikzpicture}[scale=0.5,
						node/.style = {shape=circle, draw, inner sep=0pt, minimum size=0.25cm},
						dashednode/.style = {shape=circle, draw=blue, dashed, inner sep=0pt, minimum size=0.25cm},
						textnode/.style = {shape=circle, draw, inner sep=0pt, minimum size=0.4cm},
						dashedtextnode/.style = {shape=circle, draw=blue, dashed, inner sep=0pt, minimum size=0.4cm},
						box/.style = {rectangle, fill=gray!20, rounded corners, fill opacity=1, inner sep=1pt},
                        arc/.style = {->, draw}
                    ]

\node[dashednode] (c) at (0,0) {};
\node[node] (a) at (0,-6) {};
\draw[arc, draw=blue, dashed] ($(c) + (0,1)$) to node[left] {\textcolor{blue}{$C$}} (c);
\draw[arc] (a) to node[left] {$A$} ($(a) - (0,1)$);

\node[node] (01) at ($(c) + (-1,-1)$) {};
\node[node] (02) at ($(c) + (-1,-2)$) {};
\node[node] (03) at ($(c) + (-1,-3)$) {};
\node[node] (04) at ($(c) + (-1,-4)$) {};
\node[node] (05) at ($(c) + (-1,-5)$) {};
\draw[arc] (c) to (01);
\draw[arc] (01) to (02);
\draw[arc] (02) to (03);
\draw[arc] (03) to (04);
\draw[arc] (04) to (05);
\draw[arc] (05) to (a);

\node[node] (11) at ($(c) + (1,-1)$) {};
\node[node] (12) at ($(c) + (1,-2)$) {};
\node[node] (13) at ($(c) + (1,-3)$) {};
\node[node] (14) at ($(c) + (1,-4)$) {};
\node[node] (15) at ($(c) + (1,-5)$) {};
\draw[arc] (c) to (11);
\draw[arc] (11) to (12);
\draw[arc] (12) to (13);
\draw[arc] (13) to (14);
\draw[arc] (14) to (15);
\draw[arc] (15) to (a);

\node[node] (08) at (-4,-1) {};
\node[node] (09) at (-3,-1) {};
\node[node] (00) at (-2,-1) {};
\draw[arc] ($(08) - (1,0)$) to node[above] {$W$} (08);
\draw[arc] (08) to (09);
\draw[arc] (09) to (00);
\draw[arc] (00) to (01);
\draw[arc, out=135, in=270, looseness=0.5] (15) to (09);

\node[node, blue, dashed] (18) at (4,-1) {};
\node[node] (19) at (3,-1) {};
\node[node] (10) at (2,-1) {};
\draw[arc, draw=blue, dashed] ($(18) + (1,0)$) to node[above] {\textcolor{blue}{$Y$}} (18);
\draw[arc] (18) to (19);
\draw[arc] (19) to (10);
\draw[arc] (10) to (11);
\draw[arc, out=45, in=270, looseness=0.5] (05) to (19);

\node[dashednode] (b) at (2,-6) {};
\draw[arc] (b) to (04);
\draw[arc] (b) to (14);
\draw[arc, draw=blue, dashed] ($(b) - (0,1)$) to node[right] {\textcolor{blue}{$B$}} (b);

\node[node] (d) at (2,0) {};
\draw[arc] (00) to (d);
\draw[arc] (10) to (d);
\draw[arc] (d) to node[right] {$D$} ($(d) + (0,1)$);

\node[node] (011) at (-4,-2) {};
\draw[arc] (02) to (011);
\draw[arc] (011) to node[below] {$X$} ($(011) - (1,0)$);

\node[dashednode] (111) at (4,-2) {};
\draw[arc] (12) to (111);
\draw[arc, blue, dashed] (111) to node[below] {\textcolor{blue}{$Z$}} ($(111) + (1,0)$);

\end{tikzpicture}
}
\caption{The switch gadget. The black solid elements are always part of the corresponding clause gadget. The vertices $Y$ and $Z$ are part of the gadget of the variable that belongs to the clause. The vertices $B$ and $C$ are part of the previous clause gadget if it is the first switch gadget of the clause.}
\label{fig:reduction:3sat-directed-two-disjoint-path:switch}
\end{figure}
The idea of this gadget is that two vertex-disjoint paths entering at vertex $B$ (respectively leaving at vertex $A$) need to leave through vertex $D$ (respectively need to enter at vertex $C$) such that either the path from $X$ to $Z$ or the path from $W$ to $X$ is still usable without violating the vertex-disjoint path constraint.
In order to integrate this gadget into the complete construction, we use the schematic view on the gadget as depicted in \Cref{fig:reduction:3sat-directed-two-disjoint-path:switch:schema}.

\begin{figure}
\centering
\scalebox{1}{
			\begin{tikzpicture}[scale=1,
						node/.style = {shape=circle, draw, inner sep=0pt, minimum size=0.25cm},
						dashednode/.style = {shape=circle, draw=blue, dashed, inner sep=0pt, minimum size=0.25cm},
						textnode/.style = {shape=circle, draw, inner sep=0pt, minimum size=0.4cm},
						dashedtextnode/.style = {shape=circle, draw=blue, dashed, inner sep=0pt, minimum size=0.4cm},
						box/.style = {rectangle, fill=gray!20, rounded corners, fill opacity=1, inner sep=1pt},
                        arc/.style = {->, draw}
                    ]

\node[node] (w) at (0,0) {}; \node[] at ($(w) + (0,0.33)$) {$W$};
\node[node] (x) at (0,-1) {}; \node[] at ($(x) + (0,-0.33)$) {$X$};
\node[node] (y) at (1,0) {}; \node[] at ($(y) + (0,0.33)$) {$Y$};
\node[node] (z) at (1,-1) {}; \node[] at ($(z) + (0,-0.33)$) {$Z$};

\draw[->] (w) to (x);
\draw[->] (y) to (z);
\draw[-] (0,-0.5) to (1,-0.5);

\end{tikzpicture}
}
\caption{The schematic switch gadget.}
\label{fig:reduction:3sat-directed-two-disjoint-path:switch:schema}
\end{figure}

We now start the description of the actual construction.
Let $(L, C)$ be the \textsc{3Sat} instance of literal $L$ and clauses $C$.
On the other hand, let $(V, A, s_1, t_1, s_2, t_2)$ be the \textsc{Directed Two Disjoint Path} instance with vertex set $V$, arc set $A$, and the start and end vertices of the two disjoint paths $s_1, t_1, s_2, t_2$.
First, we introduce the four start and end vertices $s_1, t_1, s_2, t_2$.
Second for every literal $\ell \in L$, we create a path $\ell^1, \ldots, \ell^{4|C|}$.
In the original reduction, four vertices and three arcs are inserted for each clause. We modify the original reduction by adding the left and right arc to the variable gadget first. Then, the middle arc is part of the clause gadget if and only if the corresponding literal is in the clause, otherwise the middle of the two vertices are merged together such that only the left and right arc remain.
Let $\ell_i, \overline \ell_i$ the literals corresponding to variable $x_i$.
Then we connect the paths of literals $\ell_i$ and $\overline \ell_i$ by introducing two vertices $x^s_i$ and $x^t_i$ together with arcs $(x^s_i, \ell^1_i), (x^s_i, \overline \ell^1_i)$ and $(\ell^{4|C|}_i, x^t_i), (\overline \ell^{4|C|}_i, x^t_i)$.
These are part of the variable gadget.
We further connect these gadgets by adding the arcs $(x^t_i, x^s_{i+1})$ for all $i \in \{1, \ldots, |X|\}$.
We call this the lobe for variable $x_i$.
Third for each clause $c_j \in C$, we add two vertices $c^1_j$ and $c^2_j$.
We connect these two vertices by three arcs of the form $(c^1_j, c^2_j)$.
Additionally, we connect these vertex pairs by adding the arc $(c^2_j, c^1_{j+1})$ for each $j \in \{1, \ldots, |C|-1\}$.
These elements are also part of the clause gadget for clause $c_j$.
At last, we connect the variable lobes with the clause path with an arc $(x^t_{|L|/2}, c^1_1)$ and we add the arc $(c^2_{|C|}, t_1)$.
These elements are part of the constant gadget.
This summarizes the overall structure of the reduction.

We are now ready to introduce switch gadgets into the construction.
For each $k$-clause $c \in C$, we add $k$ switch gadgets to the construction.
These $k$ switch gadgets are part to the clause gadget of the corresponding clause $c \in C$.
All of the switch gadgets are now stacked by merging the vertex $C$ of one switch gadget with vertex $A$ of the following switch gadget and by merging vertex $B$ of one switch gadget with vertex $D$ of the following switch gadget.
This leaves the vertices $W,X,Y$ and, $Z$ unconnected.
We connect these to the graph in the following way.
Consider the $j$-th clause $c_j$ that contains the literal $\ell$.
We identify $W$ and $X$ with the existing vertices $\ell^{4j-2}$ and $\ell^{4j-1}$ in the lobes of the variables.
Note that we merge the vertices $\ell^{4j-2}$ and $\ell^{4j-1}$ if $\ell \notin c_j$.
Additionally, the vertices $Y$ and $Z$ are identified with the vertices $c^1_j$ and $c^2_j$ induced by clause $c_j$.
Note that if there are two consecutive switch gadgets belonging to two different clauses $c_j$ and $c_{j+1}$, then we define that vertex $A$ (merged with the following $C$) and $D$ (merged with the following $B$) belong to $c_j$.
To finish the construction, we add arcs $(s_2, C_{last}), (A_1, t_2), (s_1, B_1)$, and $(D_{last}, x^s_1)$, which are part of the constant gadget.
The whole construction can be found in \Cref{fig:reduction:3sat-directed-two-disjoint-path} in which we symbolize the usage of a switch gadget as depicted in \Cref{fig:reduction:3sat-directed-two-disjoint-path:switch:schema}.
The construction as described above is pre-image unique because each vertex and arc is induced by one variable or clause.

\begin{figure}
\centering
\scalebox{1}{
			\begin{tikzpicture}[scale=0.5,
						node/.style = {shape=circle, draw, inner sep=0pt, minimum size=0.25cm},
						dashednode/.style = {shape=circle, draw=blue, dashed, inner sep=0pt, minimum size=0.25cm},
						textnode/.style = {shape=circle, draw, inner sep=0pt, minimum size=0.4cm},
						dashedtextnode/.style = {shape=circle, draw=blue, dashed, inner sep=0pt, minimum size=0.4cm},
						box/.style = {rectangle, fill=gray!20, rounded corners, fill opacity=1, inner sep=1pt},
                        arc/.style = {->, draw}
                    ]

\node[node] (x1s) at (1,0) {};
\node[node] (x101) at ($(x1s) + (1,1)$) {};
\node[node] (x102) at ($(x1s) + (2,1)$) {};
\node[node] (x103) at ($(x1s) + (3,1)$) {};
\node[node] (x104) at ($(x1s) + (4,1)$) {}; \node[above left] at (x101) {$\overline \ell_1$};
\node[node] (x111) at ($(x1s) + (1,-1)$) {};
\node[node] (x112) at ($(x1s) + (2,-1)$) {};
\node[node] (x113) at ($(x1s) + (3,-1)$) {};
\node[node] (x114) at ($(x1s) + (4,-1)$) {}; \node[below left] at (x111) {$\ell_1$};
\node[node] (x1t) at ($(x1s) + (5,0)$) {};
\draw[->] (x1s) to (x101);
\draw[->] (x101) to (x102);
\draw[->] (x102) to (x103);
\draw[->] (x103) to (x104);
\draw[->] (x104) to (x1t);
\draw[->] (x1s) to (x111);
\draw[->] (x111) to (x112);
\draw[->] (x112) to (x113);
\draw[->] (x113) to (x114);
\draw[->] (x114) to (x1t);

\node[node] (x2s) at (7,0) {};
\node[node] (x201) at ($(x2s) + (1,1)$) {};
\node[node] (x202) at ($(x2s) + (2,1)$) {};
\node[node] (x203) at ($(x2s) + (3,1)$) {};
\node[node] (x204) at ($(x2s) + (4,1)$) {}; \node[above left] at (x201) {$\overline \ell_2$};
\node[node] (x211) at ($(x2s) + (1,-1)$) {};
\node[node] (x212) at ($(x2s) + (2,-1)$) {};
\node[node] (x213) at ($(x2s) + (3,-1)$) {};
\node[node] (x214) at ($(x2s) + (4,-1)$) {}; \node[below left] at (x211) {$\ell_2$};
\node[node] (x2t) at ($(x2s) + (5,0)$) {};
\draw[->] (x2s) to (x201);
\draw[->] (x201) to (x202);
\draw[->] (x202) to (x203);
\draw[->] (x203) to (x204);
\draw[->] (x204) to (x2t);
\draw[->] (x2s) to (x211);
\draw[->] (x211) to (x212);
\draw[->] (x212) to (x213);
\draw[->] (x213) to (x214);
\draw[->] (x214) to (x2t);

\node[node] (x3s) at (13,0) {};
\node[node] (x301) at ($(x3s) + (1,1)$) {};
\node[node] (x302) at ($(x3s) + (2,1)$) {};
\node[node] (x303) at ($(x3s) + (3,1)$) {};
\node[node] (x304) at ($(x3s) + (4,1)$) {}; \node[above left] at (x301) {$\overline \ell_3$};
\node[node] (x311) at ($(x3s) + (1,-1)$) {};
\node[node] (x312) at ($(x3s) + (2,-1)$) {};
\node[node] (x313) at ($(x3s) + (3,-1)$) {};
\node[node] (x314) at ($(x3s) + (4,-1)$) {}; \node[below left] at (x311) {$\ell_3$};
\node[node] (x3t) at ($(x3s) + (5,0)$) {};
\draw[->] (x3s) to (x301);
\draw[->] (x301) to (x302);
\draw[->] (x302) to (x303);
\draw[->] (x303) to (x304);
\draw[->] (x304) to (x3t);
\draw[->] (x3s) to (x311);
\draw[->] (x311) to (x312);
\draw[->] (x312) to (x313);
\draw[->] (x313) to (x314);
\draw[->] (x314) to (x3t);

\node[node] (s1) at (-3,0) {}; \node[above left] at (s1) {$s_1$};
\node[node] (b0) at (-2,0) {}; \node[above] at (b0) {$B_1$};
\node[node] (d1) at (0,0) {}; \node[below] at (d1) {$D_{\ell\textit{ast}}$};
\node[node] (t11) at (19,0) {}; \node[above right] at (t11) {};
\node[node] (t1) at (0,-4) {}; \node[above left] at (t1) {$t_1$};
\node[node] (s11) at (19,-4) {}; \node[above right] at (s11) {};
\draw[->] (s1) to (b0);
\draw[->] (d1) to (x1s);
\draw[->] (x1t) to (x2s);
\draw[->] (x2t) to (x3s);
\draw[->] (x3t) to (t11);
\draw[->, out=0, in=0, looseness=1] (t11) to (s11);

\node[node] (s2) at (6.5,4) {}; \node[above] at (s2) {$s_2$};
\node[node] (t2) at (12.5,4) {}; \node[above] at (t2) {$t_2$};
\node[node] (c1) at (8,4) {}; \node[below] at (c1) {$C_{\ell\textit{ast}}$};
\node[node] (a0) at (11,4) {}; \node[below] at (a0) {$A_1$};
\draw[->] (s2) to (c1);
\draw[->] (a0) to (t2);

\node[node] (c11) at (8,-4) {}; \node[below] at (c11) {$c^1_1$};
\node[below] (c1) at (9.5,-5) {$c_1$};
\node[node] (c12) at (11,-4) {}; \node[below] at (c12) {$c^2_1$};
\draw[->] (s11) to (c12);
\draw[->] (c11) to (t1);

\draw[->] (c12) to (c11);
\draw[->, in=45, out=135, looseness=0.7] (c12) to (c11);
\draw[->, in=315, out=225, looseness=0.7] (c12) to (c11);

\draw[-] (9.5,-4) to ($(x112)+(0.4,0)$);
\draw[-] (9.5,-3.5) to ($(x212)+(0.4,0)$);
\draw[-] (9.5,-4.5) to ($(x302)+(0.4,0)$);

\end{tikzpicture}
}
\caption{Classical reduction of \textsc{3Sat} to \textsc{Directed Two Disjoint Path} for $\varphi = (\overline \ell_1 \wedge \overline \ell_2 \wedge \ell_3)$.}
\label{fig:reduction:3sat-directed-two-disjoint-path}
\end{figure}

There is a correspondence between the assignment of variables and the lobes of the variables and thus the variable gadgets:
If on the one hand the path corresponding to literal $\ell_i$ is part of the solution, then $x_i$ is assigned to false; if on the other hand the path corresponding to literal $\overline \ell_i$ is part of the solution, then $x_i$ is assigned to true.
Furthermore observe that if the path of $\ell_i$ is part of the solution, all arcs of $\overline \ell_i$ are still unused.
Accordingly, while traveling through the clause gadget of clause $c_j$, which contains literal $\overline \ell_i$, these arcs can be used.
On the other hand, a clause gadget of clause $c_j$ containing $\ell_i$ cannot make use of the arcs of the path of $\ell_i$.

This reduction is only weakly modular because removing a variable $x_i$ results in a disconnected graph.
However, there is a corresponding removal gadget for clauses mitigating this problem.
Thus the removal gadget of a variable is the variable gadget itself together with the clause removal gadgets.
In order to deactivate a clause, one can add a path from $c^1_j$ to $c^2_j$ of length $5$ such that the number of used arcs remains the same.
If this new path is used the clause is simulated to be satisfied.
A solution in the base scenario is extendable to a full solution because the clause removal gadget only fixes the solution on the literals that are part of the clause.
    
At last, we describe the solution size function over the arcs.
Due to weak modularity, we use a reduction $f$ from an instance $(L',C')$ with $L \subseteq L'$ and $C \subseteq C'$ such that
$$
	size_f(L,C) = 5 + \frac{2|L'|}{2} + 2|C'||L'| + \sum_{c \in C'} (\frac{11|c|}{2} + \frac{5|c|}{2} + 5).
$$
Overall, we include $5$ arcs for the constant gadget, which are the arcs $(s_1, B_1)$, $(s_2, C_{last})$, $(c^1_{|C|}, t_1)$, $(A_1, t_2)$, and $(x^t_{|X|}, c^2_1)$.
For each clause, we need to travel the $C$-$A$ path and the $B$-$D$ path, which are $11$ arcs for each variable in the clause.
Furthermore, we have to travel over the $W$-$X$ path of all switch gadgets for each clause.
Accordingly for each clause $c \in C$, we have to include $5|c|/2$ arcs.
In order to travel from $c_j^1$ to $c_j^2$ for each clause $c_j$, the $Y$-$Z$ path has to be used, which are $5$ arcs.
At last for each variable $x_i$, we have to use one of the arcs $(x_i^s,\ell_i)$, $(x^s_i,\overline \ell_i)$ and one of the arcs $(\ell_i,x^t_i)$, $(\overline \ell_i,x^t_i)$.
Additionally, we need to include two arcs per clause to travel through the lobe.
This completes the description of the solution size function.
Note that the clause removal gadgets were designed such that the number of solution elements remains the same.
\end{proof}

\begin{lemma}
    \tddp{} is strongly modular universe gadget reducible to \kddp{} such that the solution properties hold and a solution size function for this reduction exists.
\end{lemma}
\begin{proof}
	The reduction from \textsc{$k$Disjoint Directed Path} to \textsc{$k+1$Disjoint Directed Path} for $k \geq 2$ works as follows.
	The reduction consists only of a constant gadget, which adds an additional path from $s_{k+1}$ to $t_{k+1}$ over the single arc $(s_{k+1}, t_{k+1})$.
	The solution size functions needs to include this additional arc.
	Because, the rest of the instance remains the same, we have a one-to-one correspondence between all combinatorial elements.
	Thus, modularity and pre-image uniqueness remain.
\end{proof}

\subsubsection{\textsc{3-Dimensional Matching}}
\begin{lemma}
    \sat{} is weakly modular universe gadget reducible to \tdm{} such that the solution properties hold and a solution size function for this reduction exists.
\end{lemma}
\begin{proof}
	A modification of the reduction from \sat{} to \tdm{} by Garey and Johnson \cite{DBLP:books/fm/GareyJ79} is a universe gadget reduction.
	For \sat{}, we use the literals as universe elements and the relations $R_{\ell, \overline{\ell}}$, which relates a literal and its negation, $R_{\ell, c}$, which relates the literals with the clauses, and  $R_{\overline \ell, c}$, which relates the negated literals with the clauses.
	\tdm{} consists of a ground set $U$ including all elements of the triples.
	Additionally, there is a 3-ary relation between the triples $T \subseteq U_1 \times U_2 \times U_3$ with $U_1~\dot\cup~U_2~\dot\cup~U_3 = U$ and $|U_1| = |U_2| = |U_3|$.
	A solution is a perfect matching $M \subseteq T$ of $U$.

	We describe the construction and explain how to divide the elements in variable and clause gadgets.
	In the original reduction the sets $T^t_i$, $T^f_i$, and $G$ of triples were introduced.
	The semantics are to include the set $T^t_i$ in the solution if variable $x_i$ was set to true and $T^f_i$ if the variable $x_i$ is set to false.
	The set $G$ is a garbage collection set that has the task to collect all non-matched elements in the sets $T^t_i$ and $T^f_i$.
	
	For the modified version, we introduce a variable gadget for the literal pair $\ell_i, \overline \ell_i$.
	Such a variable gadget consists of four triples $(\overline \ell_i[0], a_i[0], b_i[0])$ (belonging to the set $T^t_i$), $(\ell_i[0], a_i[1], b_i[0])$ (belonging to the set $T^f_i$), $(\overline \ell_i[0], g_1[i], g_2[i])$, and $(\ell_i[0], g_1[i], g_2[i])$ (both belonging to the garbage collection set $G$).
	For each clause, we introduce the following triples as clause gadget.
	We assume that the clauses $C$ are ordered and let $|c|$ denote the number of literals in clause $c \in C$.
    Further, let $\gamma_i$ be the number of clauses containing $\ell_i$ or $\overline \ell_i$.
	\begin{align*}
		(\overline \ell_i[j], a_i[j], b_i[j]),&  \text{ if $c$ is the $j$-th clause with } \ell_i \text{ or } \overline \ell_i \in c\\
		(\ell_i[j], a_i[j+1~\text{mod}~\gamma_i+1], b_i[j]),& \text{ if $c$ is the $j$-th clause with } \ell_i \text{ or } \overline \ell_i \in c,\\
		(\overline \ell_i[j],g^c_1[k], g^c_2[k]),& \text{ if $c$ is the $j$-th clause with } \ell_i \text{ or } \overline \ell_i \in c,\\ & \qquad\qquad\qquad\qquad\qquad\qquad\quad \ \text{ and } 1 \leq k \leq |c|-1\\
		(\ell_i[j],g^c_1[k], g^c_2[k]),& \text{ if $c$ is the $j$-th clause with } \ell_i \text{ or } \overline \ell_i \in c,\\
        & \qquad\qquad\qquad\qquad\qquad\qquad\quad \ \text{ and } 1 \leq k \leq |c|-1\\
		(\ell_i[j], s^c_1, s^c_2),& \text{ if $c$ is the $j$-th clause with } \ell_i \text{ or } \overline \ell_i \in c,
        \text{ and } \ell_i \in c\\
		(\overline \ell_i[j], s^c_1, s^c_2),& \text{ if $c$ is the $j$-th clause with } \ell_i \text{ or } \overline \ell_i \in c, \text{ and } \overline \ell_i \in c\\
	\end{align*}
	The element $(\overline \ell_i[j], a_i[j], b_i[j])$ is part of the set $T^t_i$, the element $(\ell_i[j], a_i[j+1], b_i[j])$ is part of the set $T^f_i$, and the elements $(\overline \ell_i[j],g^c_1[k], g^c_2[k])$ and $(\ell_i[j],g^c_1[k], g^c_2[k])$ are part of the set $G$.
	In comparison to the original reduction, we leave out all elements $(\overline \ell_i[j], a_i[j], b_i[j])$, and $(\ell_i[j], a_i[j+1], b_i[j])$ in $T^t_i$ and $T^f_i$ as well as the additional garbage collection element from $G$ if the literal is not part of the clause.
	For each clause $c_j \in C$, we have added the triples
	\begin{align*}
		C_j &= \{(\ell_i[j], s^c_1, s^c_2) \mid \ell_i \in c_j\}.
	\end{align*}
	A triple from $C_j$ is taken into the solution if $\ell_i$ is able to satisfy the clause $c_j$.
	This is only possible if $T^t_i$ and $T^f_i$ are chosen correspondingly.
	
	The variable removal gadget is the variable gadget itself together with the clause removal gadgets.
	The removal gadget for clause $c_j \in C$ is
	\begin{align*}
		\overline C_j = \{(\ell_i[j],s^c_{1}, s^c_{2}),(\overline\ell_i[j],s^c_{1}, s^c_{2}) \mid \ell_i \in c_j \}.
	\end{align*}
	If a clause $C_j$ is unsatisfied, then one of $\ell_i[j]$ and $\overline\ell_i[j]$ is not matched for some $\ell_i \in C_j$ because $G$ is only able to match at most $k-1$ many $\ell_i[j]$ for a clause of size $k$, while $2k$ many $\ell_i[j]$ are introduced and $k$ being matched by the elements from $T^t_i$ or $T^f_i$.
	In other words, the set $\overline C_j$ simulates that $C_j$ is satisfied and also includes the element $\ell_i[j]$, which cannot be matched by $G$.
	In conclusion, the solution size stays the same if a clause is removed.
	Thus we get
	$$
		size_f(L, C) = |L'| + \sum_{c \in C'} 2|c|,
	$$
	and the reduction is modular by the given one-to-many correspondence of the literals and clauses on the one side and their gadgets on the other side.
    A solution in the base scenario is extendable to a full solution because the clause removal gadget only fixes the solution on the literals that are part of the clause.
\end{proof}

\begin{lemma}
    \tdm{} is strongly modular universe gadget reducible to \xtc{} such that the solution properties hold and a solution size function for this reduction exists.
\end{lemma}
\begin{proof}
	The reduction from \tdm{} to \xtc{} by Garey and Johnson \cite{DBLP:books/fm/GareyJ79} is a universe gadget reduction.
	Because \xtc{} is a generalization of \tdm{}, the \tdm{} instance is just reinterpreted as an instance of \xtc{}.
	This yields a direct one-to-one correspondence between the 3-tuples and 3-sets.
	Thus, the solution size function remains and all necessary properties still hold.
\end{proof}

\subsubsection{\textsc{Coloring} (Partition Problems)}
\begin{lemma}
    \sat{} is strongly modular universe gadget reducible to \tcol{} such that the solution properties hold and a solution size function for this reduction exists.
\end{lemma}
\begin{proof}
    The reduction from \sat{} to \col{} by Garey et al. \cite{DBLP:journals/tcs/GareyJS76} is a universe gadget reduction.
    \col{} has the vertices of the graph $V$ as universe elements and the edges $E$ as relation over the universe elements.
    We therefore have the mappings $f_{const, V}$, $f_{const, E}$, $f_{L, V}$, $f_{L, E}$, $f_{C, V}$, $f_{C, E}$.

    \begin{figure}[!ht]
        \centering
        \scalebox{1}{
	\begin{tikzpicture}[scale=1,
		node/.style = {shape=circle, draw, inner sep=0pt, minimum size=0.25cm},
		dashednode/.style = {shape=circle, draw=blue, dashed, inner sep=0pt, minimum size=0.25cm},
		textnode/.style = {shape=circle, draw, inner sep=0pt, minimum size=0.4cm},
		dashedtextnode/.style = {shape=circle, draw=blue, dashed, inner sep=0pt, minimum size=0.4cm},
		box/.style = {rectangle, fill=gray!20, rounded corners, fill opacity=1, inner sep=1pt}]

        \node[node, minimum size=0.5cm, shape=star,star points=6] (T) at (-0.75,-0.75) {}; \node at ($(T) + (0,0.5)$) {$T$};
        \node[node, minimum size=0.5cm,,shape=star,star points=6] (F) at (0.75,-0.75) {}; \node at ($(F) + (0,0.5)$) {$F$};
        \node[node, minimum size=0.5cm,,shape=star,star points=6] (B) at (0,0) {}; \node at ($(B) + (0,0.5)$) {$B$};
    
        \path [-, black] (B) edge[bend left=0] node[above] {} (T);
        \path [-, black] (B) edge[bend left=0] node[above] {} (F);
        \path [-, black] (T) edge[bend left=0] node[above] {} (F);
    
    \end{tikzpicture}
}
        \caption{
            Constant Gadget for the reduction.
            The corresponding mappings are $f_{const, V}: \emptyset \mapsto \{B, F, T\}$ and $f_{const, E}: \emptyset \mapsto \{\{B,F\},\{B,T\},\{F,T\}\}$
        }
        \label{fig:Class:3Sat_3Coloring_Constant_Gadget}
    \end{figure}

    \begin{figure}[!ht]
        \centering
        \scalebox{1}{
	\begin{tikzpicture}[scale=1,
		node/.style = {shape=circle, draw, inner sep=0pt, minimum size=0.25cm},
		dashednode/.style = {shape=circle, draw=blue, dashed, inner sep=0pt, minimum size=0.25cm},
		textnode/.style = {shape=circle, draw, inner sep=0pt, minimum size=0.4cm},
		dashedtextnode/.style = {shape=circle, draw=blue, dashed, inner sep=0pt, minimum size=0.4cm},
		box/.style = {rectangle, fill=gray!20, rounded corners, fill opacity=1, inner sep=1pt}]
        \node[dashednode, minimum size=0.5cm, shape=star,star points=6] (B) at (-2,0) {}; \node at ($(B) + (0,0.5)$) {$B$};
        
        \node[node] (VT) at (-0.75,-0.75) {}; \node at ($(VT) + (0.33,0)$) {$\ell$};
        \node[node] (VF) at (-0.75,0.75) {}; \node at ($(VF) + (0.33,0)$) {$\overline{\ell}$};
    
        \path [-, black] (B) edge[bend left=0] node[above] {} (VT);
        \path [-, black] (B) edge[bend left=0] node[above] {} (VF);
        \path [-, black] (VT) edge[bend left=0] node[above] {} (VF);
    \end{tikzpicture}
}
        \caption{
            Variable Gadget representing literals $\ell, \overline{\ell} \in L$.
            The corresponding mappings are $f_{L, V}: (\ell, \overline \ell) \mapsto \{v_\ell, v_{\overline{\ell}}\}$ and $f_{L, E}: (\ell, \overline \ell) \mapsto \{\{v_\ell,v_{\overline{\ell}}\},\{v_\ell,B\},\{v_{\overline{\ell}},B\}\}$
        }
        \label{fig:Class:3Sat_3Coloring_Literal_Gadget}
    \end{figure}

    \begin{figure}[!ht]
        \centering
        \scalebox{1}{
    \begin{tikzpicture}[scale=1,
		node/.style = {shape=circle, draw, inner sep=0pt, minimum size=0.25cm},
		dashednode/.style = {shape=circle, draw=blue, dashed, inner sep=0pt, minimum size=0.25cm},
		textnode/.style = {shape=circle, draw, inner sep=0pt, minimum size=0.4cm},
		dashedtextnode/.style = {shape=circle, draw=blue, dashed, inner sep=0pt, minimum size=0.4cm},
		box/.style = {rectangle, fill=gray!20, rounded corners, fill opacity=1, inner sep=1pt}]

        \node[dashednode, minimum size=0.5cm, shape=star,star points=6] (F) at (5,-1) {}; \node at ($(F) + (0,-0.5)$) {F};
        \node[dashednode, minimum size=0.5cm, shape=star,star points=6] (B) at (5,0.5) {}; \node at ($(B) + (0,0.5)$) {B};
    
        \path [-, blue, dashed] (B) edge[bend left=0] node[above] {} (F);
    
        \node[dashednode] (x_1) at (0,1) {}; \node at ($(x_1) + (0,0.33)$) {$\ell_1$};
        \node[dashednode] (x_2) at (0,0) {}; \node at ($(x_2) + (0,0.33)$) {$\ell_2$};
        \node[dashednode] (x_3) at (0,-1) {}; \node at ($(x_3) + (0,0.33)$) {$\ell_3$};
        
        \node[node] (y_1) at (1,1) {}; \node at ($(y_1) + (0,0.33)$) {$c^1_1$};
        \node[node] (y_2) at (1,0) {}; \node at ($(y_2) + (0,-0.33)$) {$c^1_2$};
        \node[node] (y_4) at (2,0.5) {}; \node at ($(y_4) + (0,0.33)$) {$c^1_3$};
        
        \node[node] (y_3) at (3,-1) {}; \node at ($(y_3) + (0,-0.33)$) {$c^2_1$};
        \node[node] (y_5) at (3,0.5) {}; \node at ($(y_5) + (0,0.33)$) {$c^2_2$};
        \node[node] (y_6) at (4,-0.25) {}; \node at ($(y_6) + (0,0.33)$) {$c^2_3$};
        
        \path [-, black, thick] (x_1) edge[bend left=0] node[above] {} (y_1);
        \path [-, black, thick] (x_2) edge[bend left=0] node[above] {} (y_2);
        \path [-, black, thick] (x_3) edge[bend left=0] node[above] {} (y_3);
        
        \path [-, black, thick] (y_1) edge[bend left=0] node[above] {} (y_2);
        \path [-, black, thick] (y_1) edge[bend left=0] node[above] {} (y_4);
        \path [-, black, thick] (y_2) edge[bend left=0] node[above] {} (y_4);
        
        \path [-, black, thick] (y_4) edge[bend left=0] node[above] {} (y_5);
        
        \path [-, black, thick] (y_3) edge[bend left=0] node[above] {} (y_5);
        \path [-, black, thick] (y_3) edge[bend left=0] node[above] {} (y_6);
        \path [-, black, thick] (y_5) edge[bend left=0] node[above] {} (y_6);
        
        \path [-, black, thick] (y_6) edge[bend left=0] node[above] {} (F);
        \path [-, black, thick] (y_6) edge[bend left=0] node[above] {} (B);
    \end{tikzpicture}
}
        \caption{
            Clause Gadget for $c \in C$.
            The corresponding mappings are $f_{C, V}: c \mapsto \{c^1_1, c^1_2, c^1_3, c^2_1, c^2_2, c^2_3\}$ and $f_{C, E}: c \mapsto \{\{c^1_1, c^1_2\},\{c^1_1, c^1_3\},\{c^1_2, c^1_3\},\{c^2_1, c^2_2\},\{c^2_1, c^2_3\},\{c^2_2, c^2_3\},\\ \{\ell_1,c^1_1\},\{\ell_2,c^1_2\},\{\ell_3,c^2_1\},\{c^1_3,c^2_2\},\{c^2_3,B\},\{c^2_3,F\}\}$
        }
        \label{fig:Class:3Sat_3Coloring_Clause_Gadget}
    \end{figure}

    The constant gadget is a 3-clique, see \Cref{fig:Class:3Sat_3Coloring_Constant_Gadget}.
    W.l.o.g we assume that $T$ is always in the first set, $F$ is always in the second set, and $B$ is always in the third set of the partition.
    For the literals, the mapping $f_{L, V}$ maps a literal $\ell \in L$ to two vertices.
    The mapping $f_{L, E}$, on the other hand, maps a literal $\ell \in L$ to three edges connecting the two vertices $\ell$ and $\overline \ell$ and vertex $B$ of the constant gadget, which is generated by the constant mapping $f_{const}$, visualized in \Cref{fig:Class:3Sat_3Coloring_Literal_Gadget}.
    At last, we have the clause gadget. The mapping $f_{C, V}$ maps the clause to six vertices, which are depicted as circles in \Cref{fig:Class:3Sat_3Coloring_Clause_Gadget}. The mapping $f_{C, E}$ maps the clause to the edges as shown as solid edges in \Cref{fig:Class:3Sat_3Coloring_Clause_Gadget}. The dashed vertices are part of different literal gadgets and the vertices $F$ and $B$ and the three dashed edges are from the constant gadget.

    A vertex $v_{\ell_i}$ is assigned the color $T$ if and only if the variable $x_i$ is assigned true.
    One can verify that if one of the literal vertices $v_{\ell_1}$, $v_{\ell_2}$, and $v_{\ell_3}$ is assigned color $T$, the vertex $c^2_3$ is assigned color $T$.
    Accordingly this clause gadget does not violate the coloring constraint.
    If on the other hand, $v_{\ell_1}$, $v_{\ell_2}$, and $v_{\ell_3}$ are assigned color $F$, then $c^2_3$ has to be assigned $F$ violating the coloring constraint.
    
	Overall, all vertices and edges are either generated by the constant function or are attributable to exactly one literal or one clause of the \sat-instance.
	Furthermore, deleting a variable gadget or a clause gadget results in the correct reduction such that we have strong modularity.
	Thus, the reduction fulfills the universal gadget reduction properties.
    The solution size function includes all vertices in one of the partitions (the colors).
    Thus, $size_f(L,C) = 2|L|+6|C|+3$ because every variable introduces two vertices and every clause introduces 6 vertices.
    The 3 additional vertices result from the constant gadget.
\end{proof}

\begin{lemma}
    \tcol{} is strongly modular universe gadget reducible to \kcol{} such that the solution properties hold and a solution size function for this reduction exists.
\end{lemma}
\begin{proof}
	The graph $G=(V,E)$ for \kcol{} remains, but a vertex $v_{new}$ is added and connected to all existing vertices $V$.
	Thus, $v_{new}$ needs to have a different color than all existing vertices in $V$.
	This is a universe gadget reduction because the vertex $v_{new}$ is a constant gadget and every edge to $v_{new}$ is part of the vertex gadget of $v$ together with $v$ itself.
	The pre-image uniqueness and the modularity results from the one-to-two correspondence of vertex $v$ and the vertex gadget consisting of $v$ and the edge $\{v, v_{new}\}$.
	The solution size function needs to include the additional vertex $v_{new}$, thus $1$ is added.
\end{proof}

\begin{lemma}
    \kcol{} is strongly modular universe gadget reducible to \cc{} such that the solution properties hold and a solution size function for this reduction exists.
\end{lemma}
\begin{proof}
	This reduction is analogous to the reduction from \is{} to \cl{} due to the fact that a coloring of a graph is a partition into independent sets while a clique cover is a partitions into cliques.
\end{proof}
\section{Conclusion}\label{sec:conclusion}
We have defined Hamming distance recoverable robust problems with elemental uncertainty and applied this concept to various well-known problems in \NP.
Further, we have defined universe gadget reductions to build a framework for a large class of Hamming distance recoverable robust problems.
The complexity results are that the Hamming distance recoverable robust versions of \NP-complete problems remain \NP-complete if the scenarios are polynomially computable and that the \NP-complete problems are $\Sigma^p_3$-complete for $xor$-dependency scenarios and $\Gamma$-set scenarios if \sat{} is universe gadget reducible to them and a corresponding solution size function exists.
Furthermore, multi-stage problems with $m$ stages result in $\Sigma^p_{2m+1}$-completeness if the encoding of scenarios are $xor$-dependency scenarios or $\Gamma$-set scenarios.

Remaining interesting questions are whether there is a (light-weight) reduction framework for other adversial problems or robustness concepts, for example for interdiction problems or two-stage adjustable problems, to derive completeness for higher levels in the polynomial hierarchy than \NP.
Furthermore, it is of interest whether this concept is adaptable to problems with cost uncertainty and for other distance measures.
A more special question is, which succinct encodings also result in $\Sigma^p_3$-completeness or if there are succinct encodings which result in the \NP-completeness of the problem.

\newpage

\bibliography{bibliography,bib_recoverable_robust,bib_reductions}

\end{document}